\documentclass[prb,aps,nofootinbib,amssymb,twocolumn,superscriptaddress,10pt]{revtex4-2}
\usepackage[T1]{fontenc}
\usepackage{graphicx}
\usepackage{xcolor}
\usepackage{dcolumn}
\usepackage{comment}
\usepackage{amsmath,amssymb,bbold,bm}
\usepackage{hyperref}
\usepackage{amsfonts}
\usepackage{float,latexsym}
\usepackage{multirow}
\usepackage{color}
\usepackage[normalem]{ulem}
\usepackage{cleveref}
\usepackage{physics}
\usepackage{tabularx}
\usepackage{array}
\usepackage[caption=false]{subfig}
\usepackage{siunitx}
\usepackage{gensymb}
\usepackage{pifont}
\usepackage{placeins}
\usepackage{mathtools}
\usepackage{longtable}
\usepackage{multirow}
\usepackage{nicematrix}
\usepackage{tikz}
\usepackage{xstring}
\usepackage{makecell}
\usepackage{xr}
\usepackage{ifthen}
\usepackage{chemformula}
\usepackage{combelow}
\usepackage{amsthm}
\usetikzlibrary{decorations.markings}
\usetikzlibrary{positioning}
\usetikzlibrary{arrows.meta}
\usepackage{tikz-feynman}
\usetikzlibrary{calc}
\tikzfeynmanset{warn luatex=false}
\let\vec\mathbf

\makeatletter
\def\maketag@@@#1{\hbox{\m@th\normalfont\normalsize#1}}
\makeatother

\crefname{appendix}{Appendix}{Appendices}
\crefname{equation}{Eq.}{Eqs.}
\crefname{figure}{Fig.}{Figs.}
\crefname{table}{Table}{Tables}
\crefname{section}{Section}{Sections}
\crefname{enumi}{Point}{Points}
\creflabelformat{appendix}{[#2#1#3]}
\crefmultiformat{figure}{Figs.~#2#1\xdef\mycreffirstarg{#1}#3}%
 { and~#2#1#3}{, #2#1#3}{ and~#2#1#3}

\makeatletter     
\renewcommand\onecolumngrid{
\do@columngrid{one}{\@ne}%
\def\set@footnotewidth{\onecolumngrid}
\def\footnoterule{\kern-6pt\hrule width 1.5in\kern6pt}%
}

\makeatletter
\AtBeginDocument{\let\LS@rot\@undefined}
\makeatother

\newtheorem*{lemma*}{Lemma}

\crefname{appendix}{Appendix}{Appendices}
\crefname{equation}{Eq.}{Eqs.}
\crefname{figure}{Fig.}{Figs.}
\crefname{table}{Table}{Tables}
\crefname{section}{Section}{Sections}
\creflabelformat{appendix}{[#2#1#3]}

\makeatletter
\AtBeginDocument{\let\LS@rot\@undefined}
\makeatother

\makeatletter
\renewcommand\onecolumngrid{\do@columngrid{one}{\@ne}\def\set@footnotewidth{\onecolumngrid}\def\footnoterule{\kern-6pt\hrule width 1.5in\kern6pt}}

\newcommand{\siSection}{appendix}

\allowdisplaybreaks

\begin{document}
\title{Probing the Quantized Berry Phases in 1H-\ch{NbSe2} Using Scanning Tunneling Microscopy}
\author{Dumitru C\u{a}lug\u{a}ru}
	\thanks{These authors contributed equally to this work.}
	\affiliation{Department of Physics, Princeton University, Princeton, New Jersey 08544, USA}
	\affiliation{Rudolf Peierls Centre for Theoretical Physics, University of Oxford, Oxford OX1 3PU, United Kingdom}
	\author{Yi Jiang}
	\thanks{These authors contributed equally to this work.}
	\affiliation{Donostia International Physics Center (DIPC), Paseo Manuel de Lardizábal. 20018, San Sebastián, Spain}
	\author{Haojie Guo}
	\thanks{These authors contributed equally to this work.}
	\affiliation{Donostia International Physics Center (DIPC), Paseo Manuel de Lardizábal. 20018, San Sebastián, Spain}
	\author{Sandra Sajan}
	\affiliation{Donostia International Physics Center (DIPC), Paseo Manuel de Lardizábal. 20018, San Sebastián, Spain}
	\author{Yongsong Wang}
	\affiliation{Donostia International Physics Center (DIPC), Paseo Manuel de Lardizábal. 20018, San Sebastián, Spain}
	\author{Haoyu Hu}
	\affiliation{Department of Physics, Princeton University, Princeton, New Jersey 08544, USA}
	\affiliation{Donostia International Physics Center (DIPC), Paseo Manuel de Lardizábal. 20018, San Sebastián, Spain}
	\author{Jiabin Yu}
	\affiliation{Department of Physics, University of Florida, Gainesville, FL, USA}
	\affiliation{Department of Physics, Princeton University, Princeton, New Jersey 08544, USA}
	\author{B.~Andrei Bernevig}
	\email{bernevig@princeton.edu}
	\affiliation{Department of Physics, Princeton University, Princeton, New Jersey 08544, USA}
	\affiliation{Donostia International Physics Center, P. Manuel de Lardizabal 4, 20018 Donostia-San Sebastian, Spain}
	\affiliation{IKERBASQUE, Basque Foundation for Science, Maria Diaz de Haro 3, 48013 Bilbao, Spain}
	\author{Fernando de Juan}
	\email{fernando.dejuan@dipc.org}
	\affiliation{Donostia International Physics Center, P. Manuel de Lardizabal 4, 20018 Donostia-San Sebastian, Spain}
	\affiliation{Departamento de Fisica de Materiales, Facultad de Ciencias Quimicas, Universidad del Pais Vasco (UPV-EHU), P. Manuel de Lardizabal 3, 20018 Donostia-San Sebastian, Spain}
	\affiliation{IKERBASQUE, Basque Foundation for Science, Maria Diaz de Haro 3, 48013 Bilbao, Spain}
	\author{Miguel M.~Ugeda}
	\email{mmugeda@dipc.org}
	\affiliation{Donostia International Physics Center, P. Manuel de Lardizabal 4, 20018 Donostia-San Sebastian, Spain}
	\affiliation{Centro de Física de Materiales, Paseo Manuel de Lardizábal 5, 20018 San Sebastián, Spain.}
	\affiliation{IKERBASQUE, Basque Foundation for Science, Maria Diaz de Haro 3, 48013 Bilbao, Spain}

\let\oldaddcontentsline\addcontentsline
\renewcommand{\addcontentsline}[3]{}

\begin{abstract}
	Topologically trivial insulators are classified into two primary categories: unobstructed and obstructed atomic insulators. While both types can be described by exponentially localized Wannier orbitals, a defining feature of obstructed atomic insulators is that the centers of charge of these orbitals are positioned at empty sites within the unit cell, rather than on atoms. Despite extensive theoretical predictions, the unambiguous quantitative experimental identification of an obstructed atomic phase has remained elusive. In this work, we present the first direct quantitative experimental evidence of such a phase in 1H-\ch{NbSe2}. We develop a novel method to extract the inter-orbital correlation functions from the local spectral function probed by scanning tunneling microscopy (STM), leveraging the orbital wave functions obtained from \textit{ab initio}{} calculations. Applying this technique to STM images, we determine the inter-orbital correlation functions for the atomic band of 1H-\ch{NbSe2} that crosses the Fermi level. Our results show that this band realizes an optimally compact obstructed atomic phase, providing the unambiguous experimental identification of such a phase. Our approach of deconvolving the STM signal using \textit{ab initio}{} orbital wave functions is broadly applicable to other material platforms, offering a powerful tool for exploring other electronic phases.
\end{abstract}

\maketitle

\textit{Introduction}.~The topology of electron wave functions has become a central theme in modern condensed matter physics. In broad terms, the momentum-space Bloch wave functions in a material can be Fourier transformed to provide a real-space description of the system. For topologically trivial bands, this transformation yields exponentially localized symmetric Wannier orbitals, while for topological bands, the resulting real-space wave functions are highly extended, leading to markedly different physics. Using Topological Quantum Chemistry~\cite{BRA17,ELC21} and related theories~\cite{KRU17,PO17,WAT18}, various types of nontrivial topology, such as stable and fragile~\cite{CAN18a,PO18c}, have been identified. 

Recently, aided by the advent of real-space invariants~\cite{XU24,XU21a,SON20a}, the broad class of topologically trivial bands has been further divided into unobstructed atomic (UA) and obstructed atomic (OA) insulators~\cite{SON20a,XU21a,LI22b,XU24}. While both cases admit exponentially localized real-space orbitals, in an UA insulator, the corresponding Wannier centers of charge are located at the occupied sites in the unit cell. In contrast, in an OA insulator, the center of charge is located \emph{away} from the atomic nuclei. This generically renders the Wannier orbitals of OA phases more extended than their UA counterparts and endows them with a larger quantum geometric tensor or Fubini-Study metric. The quantum geometric tensor has been shown to lower bound the superfluid weight in flat-band superconductors~\cite{PEO15,TOR22}, including those featuring an OA phase~\cite{HER22}, or even to appear in the electron-phonon coupling~\cite{YU24e}.

Despite its purported prevalence~\cite{XU21a}, the unambiguous \emph{quantitative} experimental identification of an OA phase has remained elusive, with only a handful of qualitative experimental signatures reported thus far~\cite{LIU23b,LIU24a}. In this work, we employ scanning tunneling microscopy (STM) to firmly establish that the Fermi-energy band of 1H-\ch{NbSe2} constitutes an OA band. Experimentally, the monolayer 1H-\ch{NbSe2} (henceforth referred to as \ch{NbSe2}) serves as the canonical example of a metallic transition metal dichalcogenide (TMD), extensively studied due to the emergence of correlated electronic phases at low temperatures, such as charge-density wave (CDW) order ($T_{\text{CDW}} \approx \SI{33}{\kelvin}$)~\cite{WIL75,CHA15,UGE16,LIA18,LIN20,DRE21} and superconductivity ($T_{c} \approx \SI{2}{\kelvin}$)~\cite{REV65,CAO15,UGE16,XI16,LIA18,ZHA19a,DRE21,WAN22a}. This rich experimental landscape has, in turn, inspired a wealth of theoretical investigations into the electron-phonon interaction in this material and its bulk counterpart~\cite{CAL09,LER15,FLI15,FLI16}. By measuring constant-height tunneling current maps in STM, we obtain the real-space charge density distribution (CDD) associated with the Fermi-level quasi-flat band of \ch{NbSe2}. We then devise a general method to extract the inter-orbital correlation functions from the experimentally measured CDD by effectively deconvolving the measured STM signal with the atomic orbital wave functions obtained from \textit{ab initio}{} simulations. The extracted correlators quantitatively show that the atomic orbitals of \ch{NbSe2} strongly hybridize within disjoint triangular plaquettes to realize the Fermi-level OA band, whose Wannier center is located at the empty triangular lattice site of \ch{NbSe2}, as predicted by comprehensive \textit{ab initio}{} simulations~\cite{YU24}. Remarkably, the wave function of the OA band has more than $94\%$ overlap with a prototypical theoretical model of an optimally compact OA band defined on the triangular lattice~\cite{SCH21,HER22,YU24}.

\begin{figure}[t]
	\centering
	\includegraphics[width=\columnwidth]{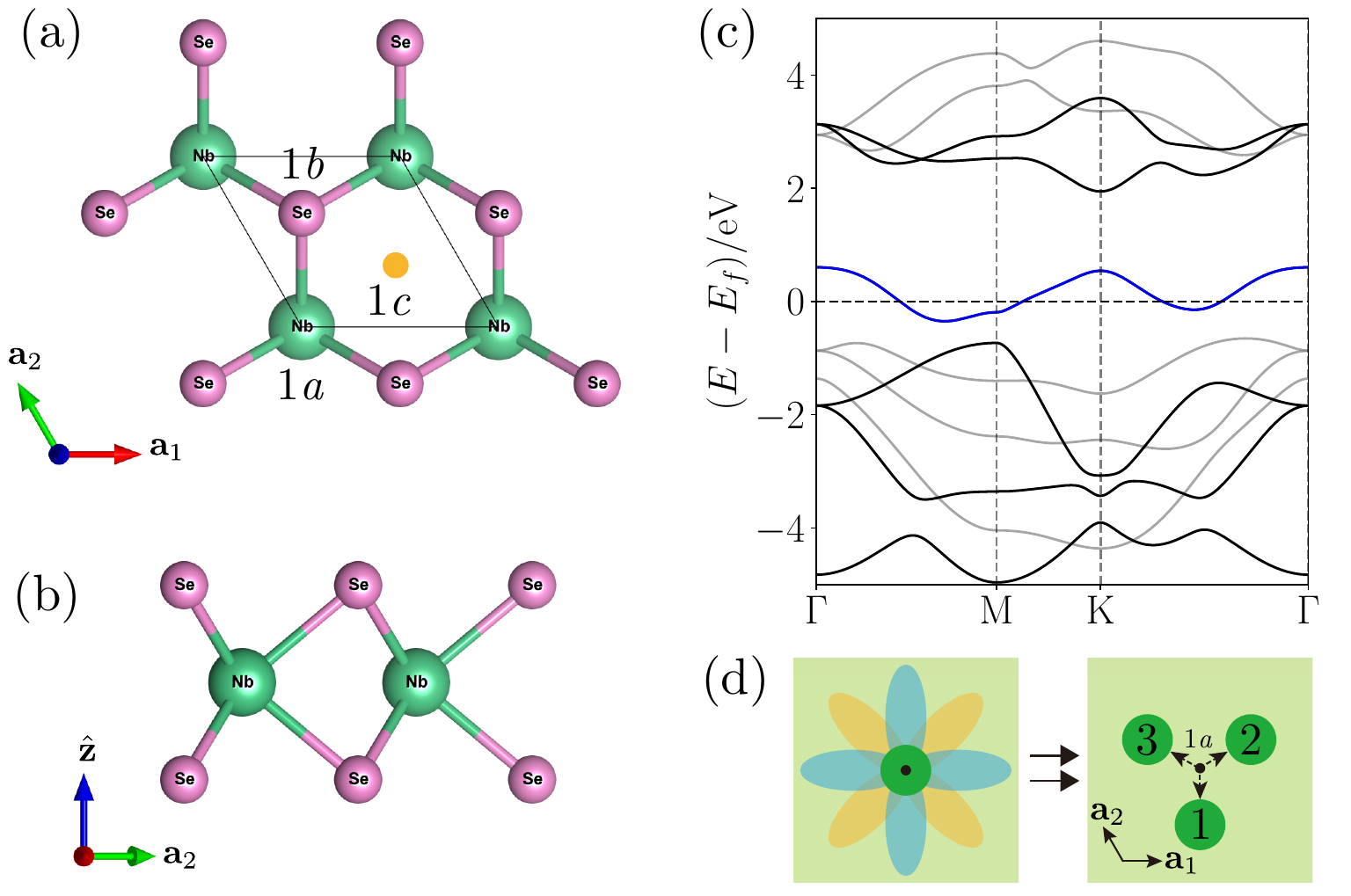}\subfloat{\label{fig:crystal_band_structure:a}}\subfloat{\label{fig:crystal_band_structure:b}}\subfloat{\label{fig:crystal_band_structure:c}}\subfloat{\label{fig:crystal_band_structure:d}}\caption{Crystal structure and electronic band structure of monolayer \ch{NbSe2}. The top and side views of the crystal structure are shown in (a) and (b), respectively. The $C_{3z}$-symmetric sites of the unit cell have been indicated in (a). (c) illustrates the \textit{ab initio}{} electronic band structure of \ch{NbSe2}, with the mirror-even (mirror-odd) bands shown in black or blue (grey). The blue quasi-flat band crossing the Fermi level is obstructed, with its Wannier center located at the empty $1c$ position. (d) highlights the effective $d_{z^2}$ (green), $d_{x^2 - y^2}$ (blue), and $d_{xy}$ (orange) orbitals that span the topmost (mirror-even) conduction bands (including the quasi-flat band). The three $d$-orbitals can be hybridized into three $s$-like orbitals disposed symmetrically around the $1a$ site.}
	\label{fig:crystal_band_structure}
\end{figure}

\textit{Monolayer \ch{NbSe2}}.~The crystalline structure of monolayer $\ch{NbSe2}$ in the pristine ({\it i.e.}{}, non-CDW) phase is shown in \cref{fig:crystal_band_structure:a,fig:crystal_band_structure:b}. It belongs to the $p\bar{6}m2$ layer group, generated by two-dimensional translations, time-reversal symmetry $\mathcal{T}$, three-fold $C_{3z}$ rotations, $m_z$ mirror reflections, and in-plane two-fold $C_{2y}$ rotations. Within the $m_z$-symmetric plane, the system includes three $C_{3z}$-invariant positions. The $1a$ position is occupied by Nb atoms, the Se atoms are situated above and below the $1b$ position, and the $1c$ position remains unoccupied. 

The \textit{ab initio}{} electronic band structure shown in \cref{fig:crystal_band_structure:c} reveals 11 bands around the Fermi level, originating from the five $d$ orbitals of Nb and the three $p$ orbitals of each of the two Se atoms. The energy spectrum can be decomposed into $m_z$ sectors. Our focus is on the topmost three mirror-even conduction bands, which include the quasi-flat band crossing the Fermi energy. These three bands are spanned by molecular $d$-like orbitals located at the $1a$ (Nb) site, which are adiabatically connected to and derive more than two-thirds of their weight from the atomic $d_{z^2}$, $d_{x^2 - y^2}$, and $d_{xy}$ orbitals of Nb. As indicated by symmetry analysis within Topological Quantum Chemistry~\cite{BRA17,ELC21}, the Fermi-level quasi-flat band constitutes an OA band induced by an $s$-like molecular orbital located at the $1c$ (unoccupied) position. 

\begin{figure}[t]
	\centering
	\includegraphics[width=\columnwidth]{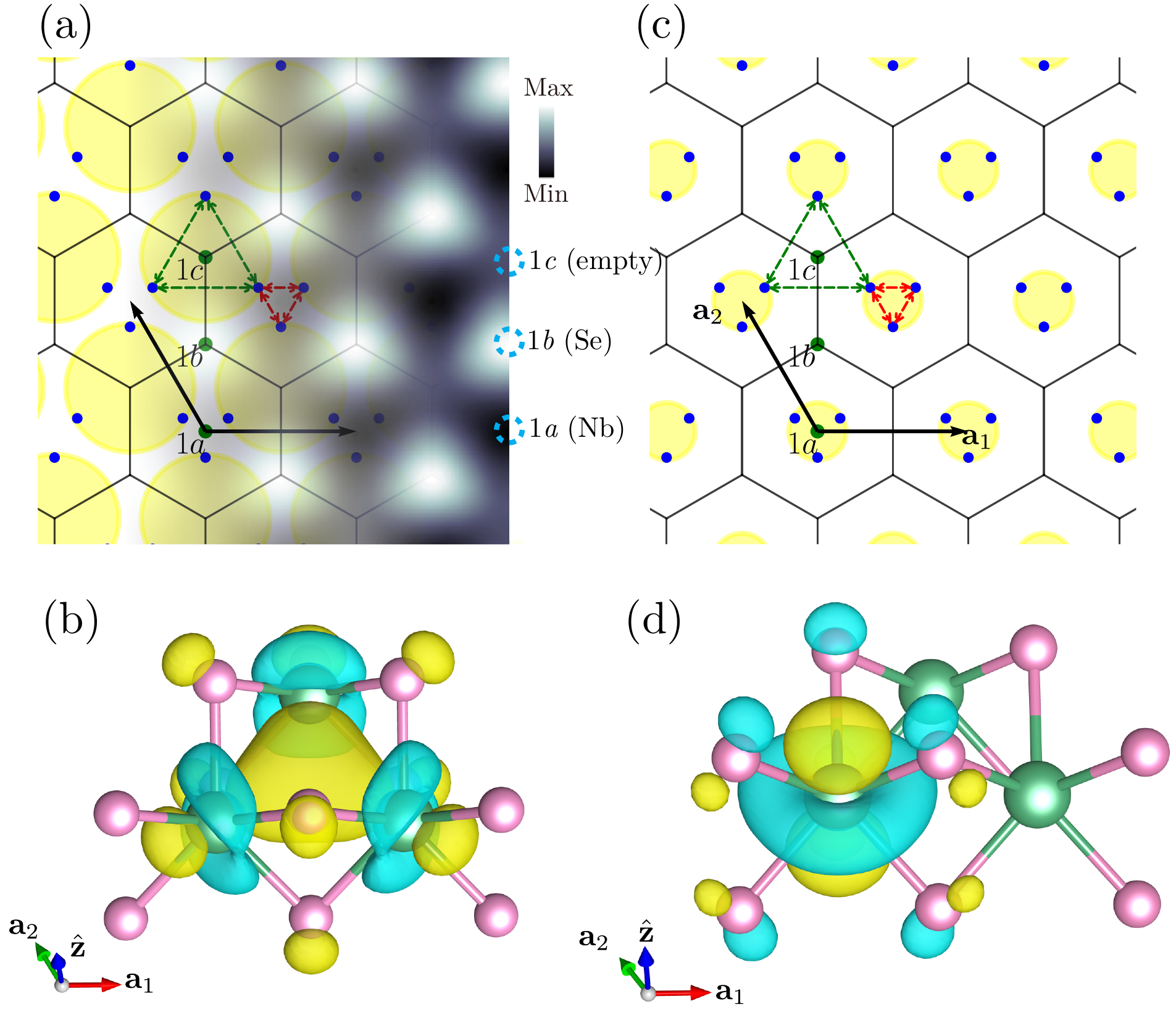}\subfloat{\label{fig:OA_and_AI_limits:a}}\subfloat{\label{fig:OA_and_AI_limits:b}}\subfloat{\label{fig:OA_and_AI_limits:c}}\subfloat{\label{fig:OA_and_AI_limits:d}}\caption{Faithful OA and fictitious UA limits of the quasi-flat band of \ch{NbSe2}. In the correct OA limit depicted in (a), the $s$-like orbitals (dots) hybridize strongly around the $1c$ position, as depicted by the yellow circles. The corresponding CDD computed from \textit{ab initio}{} (as measured by a probe located at height $z=\SI{4}{\angstrom}$ above the topmost Se atoms) is overlaid on the right half of the panel. (b) shows the Wannier orbital corresponding to the obstructed quasi-flat band, with the yellow and blue regions representing positive and negative wavefunction amplitudes, respectively. (c) and (d) are the same as (a) and (b), but for a fictitious UA limit, in which the $s$-like orbitals hybridize strongly around the $1a$ position.}
	\label{fig:OA_and_AI_limits}
\end{figure}

\textit{The OA and UA limits of \ch{NbSe2}}.~A particularly intuitive toy model for the topmost three mirror-even bands of \ch{NbSe2} shown in \cref{fig:crystal_band_structure:c} was formulated in Ref.~\cite{YU24}. In this model, the three effective $d$-like orbitals spanning the these bands are first hybridized into $s$-like orbitals located away from the $1a$ (Nb) position at three $C_{3z}$-related sites, as shown in \cref{fig:crystal_band_structure:d}. We denote the $s$-like orbitals by $\hat{c}^\dagger_{\vec{R},i}$ (for $1 \leq i \leq 3$) within unit cell $\vec{R}$. Next, the $s$-orbitals surrounding each empty $1c$ site (displaced by $\vec{r}_{1c} = \frac{2}{3} \vec{a}_1 + \frac13 \vec{a}_2$ from the unit cell origin) are hybridized by a real hopping amplitude $t$. The corresponding tight-binding Hamiltonian is approximately a sum of commuting terms associated with each unit cell, $\hat{H} = \sum_{\vec{R}} \hat{H} \left( \vec{R} \right) + \dots$, where
\begin{equation}
	\label{eqn:plaquete_ham}
	\hat{H} \left( \vec{R} \right) = \sum^{3}_{i,j=1} \hat{c}^\dagger_{\vec{R}+\delta\vec{R}_i,i} \left[ t + \left(E_0 - t\right) \delta_{ij} \right] \hat{c}_{\vec{R}+\delta\vec{R}_j,j}.
\end{equation}
with the displacements $\delta\vec{R}_i = \vec{r}_{1c} - C^{i}_{3z} \vec{r}_{1c}$ (for $i = 1, 2, 3$), and the dots denoting other terms that we ignore for now. From \textit{ab initio}{} simulations, we find that $t = \SI{-0.784}{\electronvolt}$ and $E_0 = \SI{1.733}{\electronvolt}$~\cite{YU24}. This simplified model features three flat bands: one at $\epsilon_1 = E_0 + 2t$ and two at $\epsilon_{2,3} = E_0 - t$. The Wannier orbitals associated with the lowest flat band at $\epsilon_1 = E_0 + 2t$ are trivially given by $\hat{\gamma}^\dagger_{\vec{R},1} = \frac{1}{\sqrt{3}} \sum^{3}_{i=1} \hat{c}^\dagger_{\vec{R}+\delta\vec{R}_i,i}$ and, from a symmetry standpoint, are akin to $s$-orbitals located at the empty $1c$ position, making this flat band an OA phase. The introduction of the previously ignored terms (with hopping amplitudes smaller than $\SI{0.3}{\electronvolt}$) renders the flat bands dispersive without closing any gaps, reproducing the spectra of the top three mirror-even bands shown in \cref{fig:crystal_band_structure:c}. Importantly, the overlap between the $\hat{\gamma}^\dagger_{\vec{R},1}$ orbitals and the \textit{ab initio}{} Fermi-level band of \ch{NbSe2} is approximately $94\%$~\cite{YU24}.

\begin{figure*}[t]
	\centering
	\includegraphics[width=\textwidth]{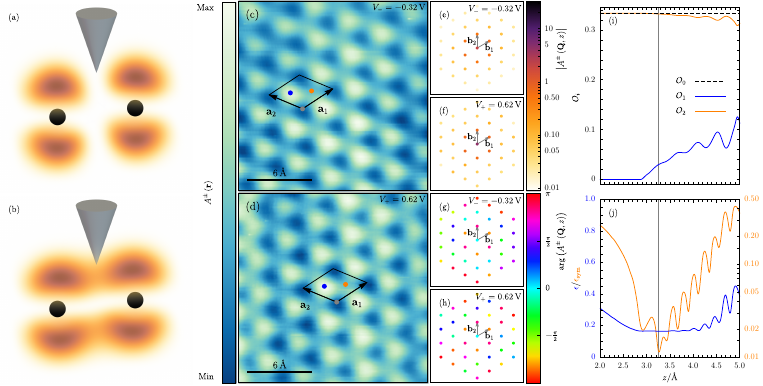}\subfloat{\label{fig:extraction_method:a}}\subfloat{\label{fig:extraction_method:b}}\subfloat{\label{fig:extraction_method:c}}\subfloat{\label{fig:extraction_method:d}}\subfloat{\label{fig:extraction_method:e}}\subfloat{\label{fig:extraction_method:f}}\subfloat{\label{fig:extraction_method:g}}\subfloat{\label{fig:extraction_method:h}}\subfloat{\label{fig:extraction_method:i}}\subfloat{\label{fig:extraction_method:j}}\caption{Extracting the orbital correlators from STM. By probing the electron density in the middle of a bond, an STM tip can distinguish between an antibonding and a bonding $\pi$-bond between two atoms, as depicted in (a) and (b), respectively. The constant-height tunneling current maps of \ch{NbSe2} at positive and negative biases $V_{\pm}$ are shown in (c) and (d). The conventional unit cell from \cref{fig:crystal_band_structure:a}, with the $1a$ (gray), $1b$ (blue), and $1c$ (orange) positions marked with dots, is overlaid on the maps. (e)-(h) display the Fourier transform amplitude of the tunneling current maps at the reciprocal lattice vectors, where the first two panels represent the magnitude and the last two illustrate the phase. (i) depicts the extracted correlators assuming different tip distances $z$ from the uppermost Se atoms. The fitting error relative to the directly extracted ($\epsilon$) and additionally symmetrized ($\epsilon_{\text{sim}}$) CDD is plotted in (j), with the minimal error indicated by the gray hairline.}
	\label{fig:extraction_method}
\end{figure*}

\Cref{fig:OA_and_AI_limits:a} schematically illustrates the hybridization effected by the terms $\hat{H} \left( \vec{R} \right)$ in \ch{NbSe2} and the resulting OA phase. The overlaid CDD of the quasi-flat band (as measured by an STM probe) reflects this strong hybridization: while the brightest $C_{3z}$-symmetric site is at the $1b$ (Se) site due to the nonzero Se weight in the $s$-like orbitals and its relative proximity to the STM probe, the second-brightest site is not at $1a$ (Nb) site, despite Nb contributing the largest orbital weight. Instead, it is located at the empty $1c$ site. This phenomenon arises precisely due to the obstructed nature of the quasi-flat \ch{NbSe2} band, whose Wannier orbital, shown in \cref{fig:OA_and_AI_limits:b}, is centered at the $1c$ position.

To better understand the OA nature of the Fermi-level \ch{NbSe2} band, we can consider a fictitious scenario in which the three $s$-like orbitals nearest to each $1a$ (empty) position strongly hybridize with one another. In such a case, depicted schematically in \cref{fig:OA_and_AI_limits:c}, a simplified Hamiltonian for the system could be written as $\hat{H}' = \sum_{\vec{R}} \hat{H}' \left( \vec{R} \right)$, with $\hat{H}' \left( \vec{R} \right)$ taking the same form as \cref{eqn:plaquete_ham}, but with $\delta \vec{R}_i = 0$ (for $i = 1, 2, 3$). The energy spectrum of this limit would be identical to that of $\hat{H}$, but the lowest flat band would now represent a UA limit. Its Wannier wave function, $\hat{\gamma'}^\dagger_{\vec{R},1} = \frac{1}{\sqrt{3}} \sum^{3}_{i=1} \hat{c}^\dagger_{\vec{R},i}$, corresponds to the original $d_{z^2}$-like effective orbital, as shown in \cref{fig:OA_and_AI_limits:d}.

The strong hybridization between $s$-orbitals within a triangular plaquette, in either the correct OA or the fictitious UA limit, is reflected in the enhancement of certain off-diagonal orbital correlators. Let $\ket{\Phi}$ and $\ket{\Phi'}$ represent the Slater determinant states where only the lowest flat band of either $\hat{H}$ or $\hat{H}'$ is fully filled. We define the orbital correlator matrix elements as $\rho^{(\prime)}_{ij} \left( \Delta \vec{R} \right) = \mel**{\Phi^{(\prime)}}{\hat{c}^\dagger_{\vec{R} + \Delta \vec{R},i} \hat{c}_{\vec{R},j}}{\Phi^{(\prime)}}$, which are real due to time-reversal symmetry. We focus on three key matrix elements (and their symmetry-related counterparts), denoted as $O_i$ (for $0 \leq i \leq 2$): $O_0 = \rho^{(\prime)}_{11} \left( \vec{0} \right)$, $O_1 = \rho^{(\prime)}_{12} \left( \vec{0} \right)$, and $O_2 = \rho^{(\prime)}_{23} \left( \vec{a}_1 \right)$. For any ground state, $O_0 = \frac{1}{3}$. However, $O_1$ and $O_2$, represented by the red and green dashed lines in \cref{fig:OA_and_AI_limits:a,fig:OA_and_AI_limits:c}, respectively, serve as effective \emph{order parameters} for the OA and UA limits. In the OA phase, $O_1 = 0$ and $O_2 = \frac{1}{3}$, while in the UA phase, $O_1 = \frac{1}{3}$ and $O_2 = 0$. In both the OA and UA phases, the inter-$s$-orbital hybridization is sufficiently strong to render the corresponding off-diagonal matrix element as large as the diagonal one. For the full \textit{ab initio}{} model of \ch{NbSe2}, we find $O_1 = 0.064$, $O_2 = 0.306$, $\rho_{11} \left( \vec{a}_1 \right) = -0.075$, $\rho_{22} \left( \vec{a}_1 \right) = 0.038$, with all other elements being no larger in absolute value than $0.014$.

\textit{Correlator extraction from STM}.~In the case of \ch{NbSe2}, the experimental identification of the OA or UA phase can be achieved by measuring the orbital correlator matrix elements. According to the Tersoff-Hamann approximation~\cite{TER83,TER85}, the tunneling current measured at a bias voltage $V$ is proportional to the local spectral function integrated from zero energy to $\abs{e}V$ (both measured relative to the Fermi energy, with $e$ as the negative electron charge). Consequently, the CDD of the \ch{NbSe2} quasi-flat band can be inferred from two tunneling current measurements at bias voltages that bracket the band in energy. Up to multiplicative factors, this enables the determination of the CDD $A \left( \vec{r} \right) = \mel**{\Phi}{\hat{\Psi}^\dagger_{} \left( \vec{r} \right) \hat{\Psi}_{} \left( \vec{r} \right)}{\Phi}$, where $\hat{\Psi}_{} \left( \vec{r} \right)$ is the annihilation fermionic field operator, at a constant height above the sample. Extracting \emph{orbital-off-diagonal} correlators from the \emph{spatially-diagonal} expectation values of the density operators $ \hat{\Psi}^\dagger_{} \left( \vec{r} \right) \hat{\Psi}_{} \left( \vec{r} \right)$ relies on the finite spatial extent of the orbital wave functions. This principle is illustrated schematically in \cref{fig:extraction_method:a,fig:extraction_method:b}, where an STM tip probing the electronic density between two hypothetical atoms can distinguish whether their corresponding $p_z$-orbitals form a bonding or antibonding molecular $\pi$-bond.

On the \ch{NbSe2} lattice, the fermionic field operator can be expressed as a convolution of the orbital fermionic operators $\hat{c}^\dagger_{\vec{R},i}$ with their Wannier wave functions $W_i \left( \vec{r}-\vec{R} \right)$, $\hat{\Psi}_{} \left( \vec{r} \right) = \sum_{\vec{R},i} W_i \left( \vec{r}-\vec{R} \right) \hat{c}_{\vec{R},i}$. As a result, the CDD directly depends on the orbital-off-diagonal correlators. By defining the Fourier transform of the quasi-flat band CDD at reciprocal lattice vectors $\vec{Q}$ as $A \left( \vec{Q}, z \right) = \frac{1}{\Omega} \int \dd[2]{r_{\parallel}} A \left( \vec{r} \right) e^{-i \vec{Q} \cdot \vec{r}}$, where $\vec{r} = \vec{r}_{\parallel} + z \hat{\vec{z}}$ and $\Omega$ is the surface area of the sample, one can show that $A \left( \vec{Q}, z \right)$ depends linearly on the orbital correlator matrix through the spatial factor $B_{ij} \left( \vec{Q}, z, \Delta \vec{R} \right)$ according to
\begin{equation}
	\label{app:sp_func_of_correlators}
	A \left( \vec{Q}, z \right) = \sum_{\Delta \vec{R}, i, j} B_{ij} \left( \vec{Q}, z, \Delta \vec{R} \right) \rho_{ij} \left( \Delta \vec{R} \right).
\end{equation}
The spatial factor $B_{ij} \left( \vec{Q}, z, \Delta \vec{R} \right)$ is computed from the overlap of Wannier wave functions corresponding to different orbitals, which are obtained from \textit{ab initio}{} calculations. The orbital correlator matrix $\rho_{ij} \left( \Delta \vec{R} \right)$ can then be determined by fitting \cref{app:sp_func_of_correlators} to the experimentally measured $A \left( \vec{Q}, z \right)$.

\Cref{fig:extraction_method:c,fig:extraction_method:d} show the experimentally measured constant-height tunneling current maps at two bias voltages $V_{\pm}$ that bracket the \ch{NbSe2} quasi-flat band. The tunneling current is proportional to the electron-like and hole-like contributions to the quasi-flat band CDD, which we denote by $A^{\pm} \left( \vec{r} \right)$. The amplitude and phases of the corresponding Fourier transforms $A^{\pm} \left( \vec{Q}, z \right)$ are shown in \crefrange{fig:extraction_method:e}{fig:extraction_method:h}. We note that the phase of the extracted $A^{\pm} \left( \vec{Q}, z \right)$ depends directly on the choice of origin for the Fourier transform, which therefore must be fixed to the convention of \cref{fig:crystal_band_structure:a}. To do so, we first note that the conventional origin corresponds to one of the three $C_{3z}$-invariant points of the \ch{NbSe2} unit cell. By inspecting the quasi-flat band CDD in \cref{fig:OA_and_AI_limits:a}, we note that the ``brightest'' $C_{3z}$-symmetric site ({\it i.e.}{} the one with the largest tunneling current) corresponds to the $1b$ (Se) position, due to the proximity of the Se atoms to the plane of the STM tip (since the tunneling current decays exponentially with the sample-tip distance). This observation was also recently confirmed experimentally by substitutional doping~\cite{HOL24}. To distinguish between the two unit cell choices in which $1b$ (Se) is the brightest site, we note that \textit{ab initio}{} simulations reveal the $1a$ (Nb) position to have the smallest tunneling current. This hierarchy of intensity is valid across any experimentally relevant sample-tip distances and can also be confirmed independently through conductance measurements at larger bias voltages, as we show later. Fixing the $1a$ (Nb) and $1b$ (Se) positions to be, respectively, the darkest and brightest $C_{3z}$-symmetric sites in the tunneling current map uniquely fixes the unit cell in the STM images, as shown in \cref{fig:extraction_method:c,fig:extraction_method:d}. After Fourier transformation, the quasi-flat band's CDD can be obtained by summing the electron- and hole-like tunneling current maps in Fourier space $A \left( \vec{Q}, z \right)  = A^+ \left( \vec{Q}, z \right) + A^- \left( \vec{Q}, z \right)$.

\begin{figure}[t]
	\centering
	\includegraphics[width=\columnwidth]{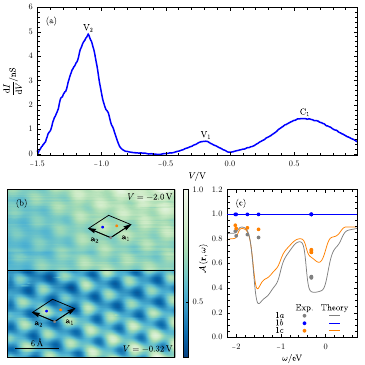}\subfloat{\label{fig:validate_pos:a}}\subfloat{\label{fig:validate_pos:b}}\subfloat{\label{fig:validate_pos:c}}\caption{Bias-dependent contrast maps in \ch{NbSe2}. (a) shows a large-bias-range conductance ($\dv{I}{V}$) curve for monolayer \ch{NbSe2}. The three prominent peaks are labeled by $\mathrm{V}_1$, $\mathrm{V}_2$, and $\mathrm{C}_1$~\cite{UGE16}. The spatially-resolved constant-height conductance maps at two bias voltages are illustrated in (b). (c) plots the spatially-averaged conductance at the three $C_{3z}$ symmetric sites (colored dots) for different bias voltages and for different regions of the sample. The conductance is compared with the \textit{ab initio}{} spectral function $\mathcal{A} \left( \vec{r}, \omega\right)$ computed at the same $C_{3z}$ symmetric positions for a tip distance of $z = \SI{4.0}{\angstrom}$. The conductance (spectral function) is normalized to one at the $1a$ site.}
	\label{fig:validate_pos}
\end{figure}

Armed with the experimentally measured quasi-flat band CDD and the spatial factor calculated using \textit{ab initio}{} methods, we can apply \cref{app:sp_func_of_correlators} to fit the orbital correlator matrix. To distinguish between the OA and UA phases, we make an approximation by including only three nonzero terms in $\rho_{ij} \left( \Delta \vec{R} \right)$, corresponding to $O_i$ (for $0 \leq i \leq 2$). This approximation is justified by \textit{ab initio}{} calculations, which show that in the correct OA limit, $O_2 \approx O_0 = \frac{1}{3}$, with other elements being much smaller. A nonzero $O_1$ is allowed as the order parameter of the UA phase. As the correlator matrix corresponds to the fully-filled quasi-flat band Slater determinant state $\ket{\Phi}$, it can be shown that $\rho_{ij} \left( \vec{R} \right)$ satisfies the following normalization conditions: $\sum_{i} \rho_{ii} \left( \vec{0} \right) = \sum_{ij} \abs{\rho_{ij} \left( \vec{R} \right)}^2 = 1$, which fix $O_0 = \frac{1}{3}$ and $O_1^2 + O_2^2 = \frac{1}{9}$. \Cref{fig:extraction_method:i} shows the fitted values of $O_i$ (for $0 \leq i \leq 2$) as a function of the \emph{assumed} tip height $z$. Since the tip height is not experimentally accessible, we perform the analysis by computing $B_{ij} \left( \vec{Q}, z, \Delta \vec{R} \right)$ at different values of $z$, effectively fitting the tip height as well. To assess the fitting error, we use two metrics: the difference between the experimentally measured $A \left( \vec{Q}, z \right)$ and the fitted one ($\epsilon$), and the difference between the experimentally measured \emph{and symmetrized} $A \left( \vec{Q}, z \right)$ and the fitted one ($\epsilon_{\text{sym}}$). As shown in \cref{fig:extraction_method:j}, the minimal error is achieved at $z = \SI{3.25}{\angstrom}$, corresponding to $O_1 = 0.030$ and $O_2 = 0.332$. This \emph{quantitatively} demonstrates that the quasi-flat band of \ch{NbSe2} forms an OA limit.

\textit{Bias-dependent contrast maps}.~The enhanced spectral weight at the empty $1c$ site is a direct consequence of the OA nature of the quasi-flat band of \ch{NbSe2}. To further validate this, we performed conductance measurements at other bias voltages corresponding to the lower-energy valence bands. \Cref{fig:validate_pos:a} shows a large-bias-range scanning tunneling spectroscopy (STS) conductance curve for \ch{NbSe2}. The peaks labeled $\mathrm{V}_1$ and $\mathrm{C}_1$, first observed in Ref.~\cite{UGE16} were assigned in Ref.~\cite{SIL16} to the band edges of the Fermi-level band in agreement with our \textit{ab initio}{} simulations shown in \cref{fig:crystal_band_structure:c}. The \textit{ab initio}{} results indicate that the uppermost mirror-even (mirror-odd) valence bands below the quasi-flat band form an atomic limit induced by Wannier orbitals located at the $1a$ ($1c$) position. Consequently, while the quasi-flat band's obstructed nature leads to a spectral weight hierarchy of $1c > 1a$, for bias voltages $V \lesssim \SI{-1.5}{\volt}$, the spectral weights for the $1a$ and $1c$ sites are expected to become comparable. This bias-dependent contrast is observed experimentally in \cref{fig:validate_pos:b}: at large negative biases ($V = \SI{-2.0}{\volt}$), the $1a$ and $1c$ sites exhibit comparable conductance, in contrast to the marked difference ($1c > 1a$) at $V = \SI{-0.32}{\volt}$. The comparison of the conductance (proportional to the spectral function $\mathcal{A} \left( \vec{r}, \omega \right)$ in the Tersoff-Hamann approximation~\cite{TER83,TER85}) with the \textit{ab initio}{} results from \cref{fig:validate_pos:c} shows excellent agreement. This provides further confirmation that the minimal spectral weight of the OA flat band indeed corresponds to the $1a$ site.

\textit{Discussion}.~We have demonstrated that the finite spatial extent of orbital wave functions allows the measurement of the spatially diagonal density operator $\hat{\Psi}^\dagger_{} \left( \vec{r} \right) \hat{\Psi}_{} \left( \vec{r} \right)$ to provide access to the orbital-off-diagonal correlators. Applying this method to monolayer \ch{NbSe2}, we have quantitatively shown that its quasi-flat band forms an OA phase. This phase is characterized by an enhancement of the orbital-off-diagonal correlators, which become comparable in magnitude to the orbital-diagonal ones. Key questions arise regarding the connection between the obstructed nature of the Fermi-level band and the CDW and superconducting phases of \ch{NbSe2}, both of which suggest an enhanced electron-phonon coupling strength. Additionally, a similar analysis can be extended to \ch{NbS2}~\cite{VAN18,KNI24}, \ch{TaSe2}~\cite{WAN23a,RYU18,HAJ13}, and \ch{TaS2}~\cite{SAN16,LIN18,YAN18,HAL19,VAN23} monolayers, which are also expected to feature an OA quasi-flat band at the Fermi level. Finally, the method of ``deconvolving'' the STM signal with the spatial factor can be applied more broadly to other platforms to infer the orbital correlator matrix. One immediate application of this approach is the determination of the CDW order parameters through STM.
\section*{Methods}

\textit{Growth of \ch{NbSe2} on BLG/SiC(0001)}.~High quality single-layer \ch{NbSe2} has been grown on epitaxial bilayer graphene (BLG) on 6H-SiC(0001) using a custom-built ultra-high vacuum (UHV) based molecular beam epitaxy (MBE) system, following similar recipes described elsewhere~\cite{DRE21, WAN22a}. In short, we first performed a slow degassing of the SiC wafer at $\SI{800}{\celsius}$, followed by high-temperature annealing at $\SI{1400}{\celsius}$ for $\SI{35}{\minute}$ to obtain bilayer graphene. The sample temperature was then maintained at $\SI{550}{\celsius}$ to grow single layers of \ch{NbSe2}. Nb and Se atoms were evaporated simultaneously from high-purity solid sources using an E-beam evaporator and a Knudsen cell, respectively, while maintaining a constant flux ratio of $1:30$ during growth. A total growth time of $\SI{35}{\minute}$ was employed, resulting in an average coverage of approximately 0.8 monolayers (ML) of \ch{NbSe2} on BLG/SiC(0001) surfaces, estimated from STM images. The entire growth process was monitored in situ \textit{via} reflection high-energy electron diffraction (RHEED) to continuously track the emergence of both BLG and \ch{NbSe2} structures on the SiC(0001) substrate. After growth, a few microns of amorphous Se were deposited as capping layers on top of \ch{NbSe2}/BLG/SiC(0001) to prevent sample degradation upon exposure to air during transport between the MBE and STM chambers. Inside the STM, the capping layer of Se can be easily removed by mild annealing (below the growth temperature) for $\SI{30}{\minute}$ without compromising the quality of the original \ch{NbSe2}/BLG/SiC(0001) samples.

\textit{STM/STS measurements and tip calibration}.~STM and STS experiments were carried out in a UHV chamber hosting a commercial STM (Unisoku, USM1300) that operates under high magnetic fields (up to \SI{11}{\tesla}, perpendicular to the sample's surface) and low temperatures (down to \SI{350}{\milli\kelvin}). In this work, all STM/STS data were recorded at \SI{77}{\kelvin} (LN$_2$ temperature), which is well above the CDW transition temperature of monolayer \ch{NbSe2} ($T_{\text{CDW}} \approx \SI{33}{\kelvin}$)~\cite{UGE16}. Before each experimental run, STM tips made of Pt/Ir were meticulously calibrated against the Shockley surface state of Cu(111) to avoid artifacts in the STM/STS data. STM images were recorded in constant height mode by disabling the feedback loop during data acquisition. To counteract the persistent thermal drift at \SI{77}{\kelvin}, the microscope was thermally stabilized for an extended period (approximately \SI{4}{\hour}) before initiating constant-height imaging. This preparation minimizes variations in the tip-sample distance during scanning. The error was further reduced by performing scans at relatively higher speeds than usual, with an average acquisition time per image of \SI{15}{\second}. Experimentally, the precise tip-sample distance is not accessible. However, the employed stabilization tunneling parameters (bias voltage $V_{s}$, and tunneling current $I_{t}$), before disconnecting the feedback loop, set directly a relative tip height with respect to the sample surface. In this work, we performed paired (or contrasting) constant height STM image measurements at different bias voltage while ensuring roughly the same tip-sample distance (considering that the residual thermal drift at 77 K is minimized). This was achieved by following a systematic approach: starting from a specific set of stabilization parameters ($V_{s}$, $I_{t}$), the feedback loop was disengaged, and the scanning area was sequentially mapped at the desired pair of bias voltage values one by one. The size of all recorded STM images was $\SI{3}{\nano\meter}\times\SI{3}{\nano\meter}$, and the specific stabilization parameters ($V_{s}$, $I_{t}$) before the feedback loop is switched off are $V_s = \SI{+0.62}{\volt}$, $I_t = \SI{4}{\nano\ampere}$ for \cref{fig:extraction_method:c,fig:extraction_method:d} and $V_s = \SI{-2}{\volt}$, $I_t = \SI{2}{\nano\ampere}$ for \cref{fig:validate_pos:b}. For the STS data, a standard lock-in detection technique was employed. During data acquisition, an AC modulated voltage ($V_{\textit{ac}}$, peak-to-peak) at $f = \SI{833}{\hertz}$ was coupled to ($V_s$). We used $V_{\textit{ac}} = \SI{6}{\milli\volt}$, with stabilization parameters set to $V_s = \SI{1}{\volt}$ and $I_t = \SI{0.5}{\nano\ampere}$ in \cref{fig:validate_pos:a}. All experimental STM/STS data were analyzed and rendered using the WSxM software~\cite{HOR07}.

\textit{Electronic structure of \ch{NbSe2} recorded at 77 K via STS}.~A representative large-bias-range $\dv{I}{V}$ curve, showing the electronic structure of monolayer \ch{NbSe2} grown on BLG/SiC(0001) and measured at \SI{77}{\kelvin}, is depicted in \cref{fig:validate_pos:a}. In the conduction band, a prominent peak labeled $\mathrm{C}_1$ is observed, centered around \SI{+0.5}{\volt}. Meanwhile, the valence band is dominated by an asymmetric peak ($\mathrm{V}_1$) located at approximately \SI{-0.2}{\volt}, followed by a region of vanishing density of states (gap) extending down to \SI{-0.8}{\volt}, where the peak labeled $\mathrm{V}_2$ begins to dominate the $\dv{I}{V}$ spectrum. These features are in excellent agreement with previous STS results obtained at \SI{4.2}{\kelvin}, indicating that the CDW transition and other temperature effects do not qualitatively change the conductance curve~\cite{UGE16}. Furthermore, this spectrum reproduces all the main characteristics of the monolayer \ch{NbSe2} band structure observed via ARPES~\cite{UGE16}, as well as the calculated one from \textit{ab initio} methods shown in \cref{fig:crystal_band_structure:c}. In comparison with the latter, the observed $\mathrm{V}_1$ and $\mathrm{C}_1$ peaks in the $\dv{I}{V}$ spectrum delineate the energy extension of the quasi-flat band shown in \cref{fig:crystal_band_structure:c}.

\textit{First-principle calculations. }
The first-principle band structures presented in this work are computed using the Vienna Ab-initio Simulation Package (VASP)~\cite{KRE93, KRE93a, KRE94, KRE96a, KRE96} with the generalized gradient approximation of the Perdew-Burke-Ernzerhof (PBE) exchange-correlation functional~\cite{PER96}. An energy cutoff of $\SI{400}{\electronvolt}$ is employed. The maximally localized Wannier functions (MLWFs) are constructed using the symmetry-adapted Wannier functions~\cite{SAK13} in Wannier90~\cite{MAR12, MAR97b, PIZ20, SOU01c}, implemented in Quantum ESPRESSO~\cite{GIA17, GIA09}, with PAW-type pseudo-potentials and PBE functionals from PSlibrary 1.0.0~\cite{DAL14}. The spin-orbit coupling (SOC) effect is not included in this work. While SOC induces small spin splitting in the band structure, it does not alter the conclusions presented here.

\begin{acknowledgments}
	We thank Yuanfeng Xu for collaboration on a related project~\cite{YU24}, as well as Emilia Moro\cb{s}an, Jonah Herzog-Arbeitman, and Ryan L. Lee for useful discussions. D.C. acknowledges support from the DOE Grant No. DE-SC0016239 and the hospitality of the Donostia International Physics Center, at which this work was carried out. D.C. also gratefully acknowledges the support provided by the Leverhulme Trust. Y.J. and H.H. were supported by the European Research Council (ERC) under the European Union’s Horizon 2020 research and innovation program (Grant Agreement No. 101020833), as well as by the IKUR Strategy under the collaboration agreement between Ikerbasque Foundation and DIPC on behalf of the Department of Education of the Basque Government. B.A.B. was supported by the Gordon and Betty Moore Foundation through Grant No. GBMF8685 towards the Princeton theory program, the Gordon and Betty Moore Foundation’s EPiQS Initiative (Grant No. GBMF11070), the Office of Naval Research (ONR Grant No. N00014-20-1-2303), the Global Collaborative Network Grant at Princeton University, the Simons Investigator Grant No. 404513, the BSF Israel US foundation No. 2018226, the NSF-MERSEC (Grant No. MERSEC DMR 2011750), the Simons Collaboration on New Frontiers in Superconductivity, and the Schmidt Foundation at the Princeton University. J. Y.'s work at Princeton University is supported by the Gordon and Betty Moore Foundation through Grant No. GBMF8685 towards the Princeton theory program. J. Y.'s work at University of Florida is supported by startup funds at University of Florida. M.M.U. acknowledges support from the European Union ERC Starting grant LINKSPM (Grant \#758558) and by the grant PID2023-153277NB-I00 funded by the Spanish Ministry of Science, Innovation and Universities. H.G. acknowledges funding from the EU NextGenerationEU/PRTR-C17.I1, as well as by the IKUR Strategy under the collaboration agreement between Ikerbasque Foundation and DIPC on behalf of the Department of Education of the Basque Government. F. J. acknowledges support from the grant PID2021-128760NB-I00 funded by the Spanish Ministry of Science, Innovation and Universities.
	
	\textit{Note added}.~In the final stages of preparing our work, we became aware of Ref.~\cite{HOL24}, which investigated the OA phase in \ch{WSe2} using a distinct approach. Their study employed substitution doping to determine atomic positions in STM images and then probed the interference of Bloch wave functions at different momenta, contrasting with the strategy used in this work. Where they overlap ({\it i.e.}{} the position of the Se atom in \ch{NbSe2} STM images), our conclusion are consistent. 
\end{acknowledgments}

\textit{Data availability}:~All data generated in this study is included in the main text and Supplementary Materials. Additional simulated data, experimental data, along with any code required for reproducing the figures, are available from the authors upon reasonable request.

\let\addcontentsline\oldaddcontentsline

\renewcommand{\thetable}{S\arabic{table}}
\renewcommand{\thefigure}{S\arabic{figure}}
\renewcommand{\theequation}{S\arabic{section}.\arabic{equation}}
\onecolumngrid
\pagebreak
\thispagestyle{empty}
\newpage
\begin{center}
	\textbf{\large Supplementary Information for ``Probing the Quantized Berry Phases in 1H-\ch{NbSe2} Using Scanning Tunneling Microscopy{}''}\\[.2cm]
\end{center}

\appendix
\renewcommand{\thesection}{\Roman{section}}
\tableofcontents
\let\oldaddcontentsline\addcontentsline
\newpage

\section{Brief theoretical review of scanning tunneling microscopy}\label{app:sec:theory_stm}

This \siSection{} provides a brief theoretical overview of scanning tunneling microscopy (STM). We begin by discussing Bardeen's formula for the tunneling current between two metals separated by a vacuum layer~\cite{BAR61}. This expression is valid when both metallic regions can be considered as hosting effectively non-interacting quasi-particles~\cite{PRA64,GOT06}. Bardeen's tunneling current formula is then applied directly to STM, where one metal represents the sample and the other serves as the STM tip. This forms the foundation of the Tersoff-Hamann approximation~\cite{TER83,TER85}, which is also reviewed here. The key result of this \siSection{} is an expression for the STM signal in terms of the system's spectral function.

\subsection{Bardeen's tunneling theory}\label{app:sec:theory_stm:bardeen}

\begin{figure}[t]
	\centering
	\includegraphics[width=0.75\textwidth]{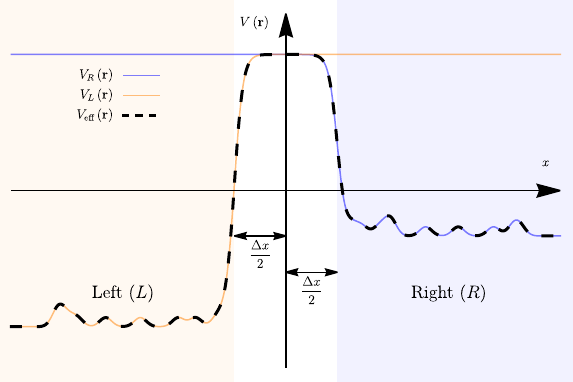}
	\caption{Tunneling of electrons between two metallic regions separated by a vacuum layer. We consider two metals separated by a vacuum of thickness $\Delta x > 0$, where the left ($L$) metal occupies the region $x < - \frac{\Delta x}{2}$ (shown in orange), and the right ($R$) metal occupies the $x > \frac{\Delta x}{2}$ region (shown in blue). The effective potential $V_{\text{eff}} \left( \vec{r} \right)$ experienced by an electron is represented by the black dashed line along a specific path parallel to the $x$ axis. The orange and blue lines represent the effective electron potential for hypothetical systems containing only the left or right metal, respectively. Within each metallic region, the potential is unaffected by the presence of the other metal and rapidly drops to its vacuum value outside the metallic regions.}
	\label{app:fig:bardeen_geometry}
\end{figure}

We begin by considering two metals, one occupying most of the positive $x$ half-plane and the other occupying most of the negative $x$ half-plane, as depicted in \cref{app:fig:bardeen_geometry}. We neglect electron-electron correlations by assuming that the electrons are well described by Landau's Fermi liquid theory within the metallic regions, and that within the barrier, the electron density is sufficiently low for Coulomb repulsion to be safely ignored. The Hamiltonian operator $\hat{H}$ for the \emph{entire} system (including both metals and the gap) can therefore be taken to include only one-body terms\footnote{In the barrier between the two metals, we consider only one-body terms since electron-electron interactions are neglected. In the metallic regions, the system is assumed to be described by quasi-particles, with their dynamics governed by a one-body Hamiltonian, where electron-electron interactions are treated at the mean-field level.}, which gives
\begin{equation}
	 \hat{H}  = \int \dd[3]{r} \ket{\vec{r}} \left( - \frac{1}{2m_e} \nabla^2 + V_{\text{eff}} \left( \vec{r} \right) \right)  \bra{\vec{r}},
\end{equation}
where $V_{\text{eff}}$ represents the effective potential describing the system, and $m_e$ is the bare electron mass. We have also employed the position eigenbasis $\ket{\vec{r}}$. For simplicity, we ignore the spin degree of freedom, assuming that both the sample and the tip feature weak spin-orbit coupling, an approximation that holds very well for the system and energy scales considered in this work.

\subsubsection{Hamiltonians for the left and right metallic regions}\label{app:sec:theory_stm:bardeen:left_right_state}

We now define the Hamiltonian $\hat{H}_{L}$ ($\hat{H}_{R}$) for a hypothetical system that contains only the left (right) metallic region. By assuming the absence of electron-electron correlations, the Hamiltonian for such a system is given by
\begin{equation}
	\hat{H}_{L/R}  = \int \dd[3]{r} \ket{\vec{r}} \left( - \frac{1}{2m_e} \nabla^2 + V_{L/R} \left( \vec{r} \right) \right)  \bra{\vec{r}},
\end{equation}
where $V_{L} \left( \vec{r} \right)$ and $V_{R} \left( \vec{r} \right)$ represent the effective potentials in the left and right metallic regions, respectively. As shown schematically in \cref{app:fig:bardeen_geometry}, the presence (or absence) of each metallic region only significantly affects the effective potential in its immediate vicinity. Specifically, we can approximate the effective potential as
\begin{equation}
    \label{app:eqn:separation_of_potential}
	V_{\text{eff}} \left( \vec{r} \right) \approx \begin{cases}
		V_{L} \left( \vec{r} \right) & \qq{for} x \lesssim \frac{\Delta x}{2} \\
		V_{R} \left( \vec{r} \right) & \qq{for} x \gtrsim -\frac{\Delta x}{2}
	\end{cases}.
\end{equation}
Let us denote the orthonormal eigenbasis of $\hat{H}_{L/R}$ by $\hat{a}^\dagger_{L/R,n} \ket{0}$, such that
\begin{equation}
	\hat{H}_{L/R} \hat{a}^\dagger_{L/R,n} \ket{0} = \epsilon_{L/R,n} \hat{a}^\dagger_{L/R,n} \ket{0},
\end{equation}
where $n > 0$ labels the eigenstates created by $\hat{a}^\dagger_{L/R,n}$, and $\epsilon_{L/R,n}$ are the corresponding energy levels, ordered by increasing energy. The associated position-space wave functions are given by
\begin{equation}
	\phi_{L/R,n} \left( \vec{r} \right) = \bra{\vec{r}} \hat{a}^\dagger_{L/R,n} \ket{0}.
\end{equation}
The low-energy states (henceforth denoted as ``bound'' states) $\phi_{L/R,n} \left( \vec{r} \right)$ are primarily localized in the respective left or right metallic regions, with only minimal probability density extending into the vacuum. Consequently, we make the further assumption that the bound state eigenfunctions of $\hat{H}_L$ and $\hat{H}_R$ are \emph{approximately} orthogonal to one another~\cite{GOT06}, leading to the condition
\begin{equation}
	\label{app:eqn:approx_orthogonality}
	\bra{0} \hat{a}^\dagger_{L,n} \hat{a}_{R,m} \ket{0} \approx 0.	
\end{equation}
We note that \cref{app:eqn:approx_orthogonality} is not true for high-energy states, which consist of plane waves both inside \emph{and outside} their corresponding metallic regions. Finally, we note that $\hat{H}$, $\hat{H}_L$, and $\hat{H}_R$ are all defined in the same region ({\it e.g.}{}, $\hat{H}_L$ is defined for the entire space, not just for the left metallic region). As such, the eigenstates of each of $\hat{H}$, $\hat{H}_L$, and $\hat{H}_R$ \emph{individually} constitute a complete basis. 

As shown schematically in \cref{app:fig:bardeen_geometry}, if $V_{L} \left( \vec{r} \right) = 0$ for $x > 0$ and $V_{R} \left( \vec{r} \right) = 0$ for $x < 0$, it follows that $V \left( \vec{r} \right) \approx V_{L} \left( \vec{r} \right) + V_{R} \left( \vec{r} \right)$, implying $\hat{H} \approx \hat{H}_L + \hat{H}_R$. However, the condition $\hat{H} \approx \hat{H}_L + \hat{H}_R$ is not strictly necessary. Our derivation remains valid even when this condition is not satisfied, such as in the presence of a potential difference between the two metals (when the vacuum potential is not constant in space) provided that \cref{app:eqn:separation_of_potential} holds true.

\subsubsection{Electron tunneling rate}\label{app:sec:theory_stm:bardeen:tunneling_rate}
Now, consider a Slater determinant of bound states from both metallic regions,
\begin{equation}
	\label{app:eqn:bardeen_initial_state}
	\ket{\Phi_0} = \hat{a}^\dagger_{L,n_1} \hat{a}^\dagger_{L,n_2} \dots \hat{a}^\dagger_{R,m_1} \hat{a}^\dagger_{R,m_2} \dots \ket{0},
\end{equation}
where $\left\lbrace n_i \right\rbrace$ and $\left\lbrace m_i \right\rbrace$ are distinct sets of quantum numbers. By ``bound states'', we refer to states of sufficiently low energy that decay exponentially outside the metallic regions (as opposed to high-energy states, which behave as plane waves in the vacuum and exhibit a macroscopic probability density beyond the metallic regions). The state in \cref{app:eqn:bardeen_initial_state} is in general \emph{not} an eigenstate of the Hamiltonian $\hat{H}$, so particles will tunnel between the two metallic regions as a result of time evolution. Our goal is to compute the tunneling rate between these regions. 

If the tunneling is completely neglected ({\it i.e.}{}, the metals are far apart), the state from \cref{app:eqn:bardeen_initial_state} is an eigenstate of $\hat{H}$ and is expected to evolve as
\begin{equation}
	\ket{\Phi_0 \left( t \right)} = \ket{\Phi_0} e^{- i \mathcal{E}_0 t}, 
\end{equation}
where
\begin{equation}
	\label{app:eqn:bardeen_gs_energy}
	\mathcal{E}_0 = \epsilon_{L,n_1} + \epsilon_{L,n_2} + \dots + \epsilon_{R,m_1} + \epsilon_{R,m_2} + \dots
\end{equation}
gives the energy of the system in the absence of tunneling. When tunneling is taken into account, electrons will move between the two metallic regions, and the time-evolved state can be expressed as a perturbation series
\begin{equation}
	\label{app:eqn:bardeen_pert_theory_ansatz}
	\ket{\Phi_0 \left( t \right)} = \ket{\Phi_0} e^{- i \mathcal{E}_0 t} + \sum_{\nu=0}^{\infty} c_{\nu} (t) \ket{\Phi_\nu}  e^{- i \mathcal{E}_{\nu} t}.
\end{equation}
In \cref{app:eqn:bardeen_pert_theory_ansatz}, $\ket{\Phi_\nu}$ for $\nu \geq 1$ represent distinct Slater determinant states (different from $\ket{\Phi_0}$) involving the same number of particles as $\ket{\Phi_0}$, but with different $\hat{a}^\dagger_{L/R,n}$ operators\footnote{As discussed in \cref{app:sec:theory_stm:bardeen:left_right_state}, the two set of states $\hat{a}^\dagger_{L,n}$ and $\hat{a}^\dagger_{R,n}$ individually furnish a complete basis for the entire space.}. For $\nu \geq 1$, the total energy $\mathcal{E}_{\nu}$ is obtained by summing the single-particle energies of the creation operators in $\ket{\Phi_\nu}$, similar to how $\mathcal{E}_0$ is defined for $\ket{\Phi_0}$ in \cref{app:eqn:bardeen_gs_energy}. The coefficients $c_{\nu} (t)$, which characterize the amplitude of each state, are assumed to be much smaller than $1$.

To find the time dependence of $c_{\nu}(t)$, we substitute \cref{app:eqn:bardeen_pert_theory_ansatz} into the Schr\"odinger equation
\begin{align}
	\left( i \dv{t} - \hat{H} \right) \ket{\Phi_0 \left( t \right)} &= 0 \nonumber \\
	\left( \mathcal{E}_{0} - \hat{H} \right) \ket{\Phi_0} e^{- i \mathcal{E}_0 t} + \sum_{\nu=0}^{\infty} i \dv{c_{\nu} (t) }{t} \ket{\Phi_\nu} e^{- i \mathcal{E}_{\nu} t} + \sum_{\nu=0}^{\infty} c_{\nu} (t) \left( \mathcal{E}_{\nu} - \hat{H} \right) \ket{\Phi_\nu} e^{- i \mathcal{E}_{\nu} t} &= 0 \label{app:eqn:bardeen_schr_1}
\end{align}
Taking the inner product of \cref{app:eqn:bardeen_schr_1} with $\bra{\Phi_\mu}$ gives
\begin{equation}
	\label{app:eqn:bardeen_schr_2}
	\sum_{\nu=0}^{\infty} i \dv{c_{\nu} (t) }{t} \bra{\Phi_\mu} \ket{\Phi_\nu} e^{-i \mathcal{E}_\nu t} = \bra{\Phi_{\mu}}\left( \hat{H} - \mathcal{E}_{0} \right) \ket{\Phi_0} e^{- i \mathcal{E}_0 t} + \sum_{\nu=0}^{\infty} c_{\nu} (t) \bra{\Phi_{\mu}} \left( \hat{H} - \mathcal{E}_{\nu} \right) \ket{\Phi_\nu} e^{- i \mathcal{E}_{\nu} t}. 
\end{equation}
Since the tunneling between the two metallic regions is weak, we assume $\abs{c_{\nu} (t)} \ll 1$ over the relevant time scales, allowing us to neglect the second term on the right-hand side\footnote{We are focused on the leading-order behavior of $c_{\nu} (t)$, which is dominated by the first term.}. Additionally, since states containing unbound electrons (whose energy is much larger than $\mathcal{E}_0$) contribute negligibly, the sum over $\nu$ in \cref{app:eqn:bardeen_schr_2} can be restricted to bound-state Slater determinants. These distinct Slater determinants, $\ket{\Phi_\nu}$ and $\ket{\Phi_\mu}$, are orthogonal due to the assumption in \cref{app:eqn:approx_orthogonality}, so
\begin{equation}
	\bra{\Phi_\mu} \ket{\Phi_\nu} \approx \delta_{\mu\nu}.
\end{equation} 
Thus, \cref{app:eqn:bardeen_schr_2} simplifies to 
\begin{equation}
	\label{app:eqn:bardeen_schr_3}
	i \dv{c_{\mu} (t) }{t} \approx \bra{\Phi_{\mu}}\left( \hat{H} - \mathcal{E}_{0} \right) \ket{\Phi_0} e^{i \left( \mathcal{E}_{\mu} - \mathcal{E}_0  \right) t} .
\end{equation}
This result is typical for time-dependent perturbation theory in quantum mechanics~\cite{SHA13}. By integrating \cref{app:eqn:bardeen_schr_3}, we find that the probability of the system being in state $\ket{\Phi_{\mu}}$ after time $t$ is
\begin{align}
	P_{\mu} (t) = \abs{c_{\mu}(t)}^2 &\approx \abs{ \bra{\Phi_{\mu}}\left( \hat{H} - \mathcal{E}_{0} \right) \ket{\Phi_0}}^2 \abs{\frac{e^{i \left( \mathcal{E}_{\mu} - \mathcal{E}_0  \right) t} - 1}{\mathcal{E}_{\mu} - \mathcal{E}_0}}^2 \nonumber \\
	&\approx \abs{ \bra{\Phi_{\mu}}\left( \hat{H} - \mathcal{E}_{0} \right) \ket{\Phi_0}}^2 \abs{\frac{\sin \left( \frac{\mathcal{E}_{\mu} - \mathcal{E}_0}{2} t \right)}{\frac{\mathcal{E}_{\mu} - \mathcal{E}_0}{2}}}^2.
\end{align} 
Taking the large-time limit (which does not necessarily imply that $\abs{c_{\nu} (t)} \ll 1$ is violated, since the magnitude of the latter is also controlled by the matrix element, which is assumed to be small)~\cite{SHA13},
\begin{equation}
	\lim_{t \to \infty} \frac{\sin^2 \left( \alpha t \right)}{\alpha^2} = \pi t \delta (\alpha) = 2 \pi t \delta (2 \alpha),
\end{equation}
we find that for long times, the probability of the system being in state $\ket{\Phi_\mu}$ is
\begin{equation}
	P_{\mu} (t) \approx 2 \pi t \abs{ \bra{\Phi_{\mu}}\left( \hat{H} - \mathcal{E}_{0} \right) \ket{\Phi_0}}^2 \delta \left( \mathcal{E}_{\mu} - \mathcal{E}_0 \right).
\end{equation}
This expression resembles Fermi's Golden Rule, where the transition rate from the state $\ket{\Phi_0}$ to the state $\ket{\Phi_{\mu}}$ is
\begin{equation}
	\label{app:eqn:bardeen_fermi_golden_rule}
	\dv{P_{\mu} (t)}{t} \approx 2 \pi \abs{ \bra{\Phi_{\mu}}\left( \hat{H} - \mathcal{E}_{0} \right) \ket{\Phi_0}}^2 \delta \left( \mathcal{E}_{\mu} - \mathcal{E}_0 \right).
\end{equation}

\subsubsection{Simplifying the tunneling many-body matrix element}\label{app:sec:theory_stm:bardeen:tunnelling_matrix_element_mb}

We now address the matrix element appearing in \cref{app:eqn:bardeen_fermi_golden_rule}. Since $\hat{H} - \mathcal{E}_{0}$ is a one-body operator and both $\ket{\Phi_0}$ and $\ket{\Phi_\nu}$ are Slater determinants, the latter can only differ from the former by a \emph{single} electron state. Focusing on processes where an electron tunnels across the barrier (which we assume, without loss of generality, is from the left to the right metallic region), we can write
\begin{equation}
	\ket{\Phi_\mu} = \hat{a}^\dagger_{R,m_0} \hat{a}_{L,n_1} \ket{\Phi_0} = \hat{a}^\dagger_{L,n_2} \hat{a}^\dagger_{L,n_3} \dots \hat{a}^\dagger_{R,m_0} \hat{a}^\dagger_{R,m_1} \hat{a}^\dagger_{R,m_2} \dots \ket{0},
\end{equation}
where $\hat{a}^\dagger_{L,n_1}$ ($\hat{a}^\dagger_{R,m_0}$) denotes a filled (empty) electron state in $\ket{\Phi_0}$. The second equality follows directly from \cref{app:eqn:bardeen_initial_state}.

To simplify the matrix element $\bra{\Phi_{\mu}} \left( \hat{H} - \mathcal{E}_{0} \right) \ket{\Phi_0}$, we make use of the fact that $\hat{H}$ is a one-body operator and apply the following lemma:

\begin{lemma*}
	Let $\hat{O}$ be a quadratic, number-conserving operator, and let $\hat{\gamma}^\dagger_{n}$ (with $n \in \mathbb{Z}$) denote fermionic creation operators, where the corresponding states are not necessarily orthogonal for different $n$. Consider two distinct sets of positive integers, $\left\lbrace n_1, n_2, \dots, n_N \right\rbrace$ and $\left\lbrace m_1, m_2, \dots, m_N \right\rbrace$. Then,
	\begin{equation}
		\label{app:eqn:many_body_overlap}
		\mel**{0}{\hat{\gamma}_{n_{N}} \dots \hat{\gamma}_{n_{2}} \hat{\gamma}_{n_{1}} \hat{O} \hat{\gamma}^\dagger_{m_{1}} \hat{\gamma}^\dagger_{m_{2}} \dots \hat{\gamma}^\dagger_{m_{N}}}{0} = \sum_{\sigma \in S_N} \epsilon(\sigma) \sum_{i=1}^{N} \mel**{0}{\hat{\gamma}_{n_{i}} \hat{O} \hat{\gamma}^\dagger_{m_{\sigma(i)}}}{0} \prod_{\substack{j = 1 \\ j \neq i}}^{N} \mel**{0}{\hat{\gamma}_{n_{j}} \hat{\gamma}^\dagger_{m_{\sigma(j)}}}{0},
	\end{equation}
	where $S_N$ is the set of all permutations of $N$ elements, and $\epsilon(\sigma)$ is the signature of the permutation $\sigma$.
\end{lemma*}
\begin{proof}
    The lemma can be proven using Wick's theorem. Let $\hat{c}^\dagger_{n}$ for $n \in \mathbb{Z}$ denote an orthonormal fermionic basis. Expanding the operator $\hat{O}$ in this basis, we have
    \begin{equation}
        \label{app:eqn:lemma_O_to_c}
        \hat{O} = \sum_{i,j} O_{ij} \hat{c}^\dagger_{i} \hat{c}_{j}, 
    \end{equation}
    where $O_{ij} = \mel**{0}{\hat{c}_{i} \hat{O} \hat{c}^\dagger_{j}}{0}$. Similarly, the operators $\hat{\gamma}^\dagger_{n}$ can be expressed as
    \begin{equation}
    	\label{app:eqn:lemma_gamma_to_c}
        \hat{\gamma}^\dagger_{n} = \sum_{m} u_{n m} \hat{c}^\dagger_{m}.
    \end{equation}
    Substituting these into \cref{app:eqn:many_body_overlap}, we can rewrite the matrix element in terms of the orthonormal basis $\hat{c}^\dagger_{n}$ as
    \begin{align}
        \mel**{0}{\hat{\gamma}_{n_{N}} \dots \hat{\gamma}_{n_{2}} \hat{\gamma}_{n_{1}} \hat{O} \hat{\gamma}^\dagger_{m_{1}} \hat{\gamma}^\dagger_{m_{2}} \dots \hat{\gamma}^\dagger_{m_{N}}}{0} =& \sum_{ \left\lbrace n'_a \right\rbrace} \sum_{ \left\lbrace m'_a \right\rbrace} O_{m'_{N+1} n'_{N+1}} \left( \prod_{a=1}^{N} u^*_{n_a,n'_a} u_{m_a,m'_a} \right) \nonumber \\
        & \times \mel**{0}{\hat{c}_{n'_{N}} \dots \hat{c}_{n'_{2}} \hat{c}_{n'_{1}} \hat{c}^\dagger_{m'_{N+1}} \hat{c}_{n'_{N+1}} \hat{c}^\dagger_{m'_{1}} \hat{c}^\dagger_{m'_{2}} \dots \hat{c}^\dagger_{m'_{N}}}{0}.
    \end{align}
    Focusing on the last matrix element, we find
    \begin{align}
        \mel**{0}{\hat{c}_{n'_{N}} \dots \hat{c}_{n'_{2}} \hat{c}_{n'_{1}} \hat{c}^\dagger_{m'_{N+1}} \hat{c}_{n'_{N+1}} \hat{c}^\dagger_{m'_{1}} \hat{c}^\dagger_{m'_{2}} \dots \hat{c}^\dagger_{m'_{N}}}{0} =& \delta_{m'_{N+1} n'_{N+1}}  \mel**{0}{\hat{c}_{n'_{N}} \dots \hat{c}_{n'_{2}} \hat{c}_{n'_{1}} \hat{c}^\dagger_{m'_{1}} \hat{c}^\dagger_{m'_{2}} \dots \hat{c}^\dagger_{m'_{N}}}{0} \nonumber \\
        & - \mel**{0}{\hat{c}_{n'_{N}} \dots \hat{c}_{n'_{2}} \hat{c}_{n'_{1}} \hat{c}_{n'_{N+1}} \hat{c}^\dagger_{m'_{N+1}} \hat{c}^\dagger_{m'_{1}} \hat{c}^\dagger_{m'_{2}} \dots \hat{c}^\dagger_{m'_{N}}}{0} \nonumber \\
        =& \sum_{\sigma \in S_{N+1}} \epsilon (\sigma) \left( \delta_{N+1, \sigma(N+1)} - 1 \right) \prod_{a=1}^{N+1} \delta_{n'_a, m'_{\sigma(a)}},
    \end{align}
    which leads to
    \begin{align}
        & \mel**{0}{\hat{\gamma}_{n_{N}} \dots \hat{\gamma}_{n_{2}} \hat{\gamma}_{n_{1}} \hat{O} \hat{\gamma}^\dagger_{m_{1}} \hat{\gamma}^\dagger_{m_{2}} \dots \hat{\gamma}^\dagger_{m_{N}}}{0} = \nonumber \\
        =& \sum_{\sigma \in S_{N+1}} \epsilon (\sigma) \left( \delta_{N+1, \sigma(N+1)} - 1 \right) \sum_{ \left\lbrace m'_a \right\rbrace} O_{m'_{N+1} m'_{\sigma(N+1)}} \left( \prod_{a=1}^{N} u^*_{n_a,m'_{\sigma(a)}} u_{m_a,m'_a} \right) \nonumber \\
        =& \sum_{\sigma \in S_{N+1}} \epsilon (\sigma) \left( \delta_{N+1, \sigma(N+1)} - 1 \right) \sum_{ \left\lbrace m'_a \right\rbrace} O_{m'_{N+1} m'_{\sigma(N+1)}} \left( \prod_{a=1}^{N} u^*_{n_a,m'_{\sigma(a)}} \right) \left( \prod_{a=1}^{N} u_{m_a,m'_{a}} \right) \nonumber \\
        =& \sum_{\sigma \in S_{N+1}} \epsilon (\sigma) \left( \delta_{N+1, \sigma(N+1)} - 1 \right) \sum_{ \left\lbrace m'_a \right\rbrace} u^{*}_{n_{\sigma^{-1}(N+1)}, m'_{N+1}} O_{m'_{N+1} m'_{\sigma(N+1)}} u_{m_{\sigma(N+1)},m'_{\sigma(N+1)}} \nonumber \\
        & \times \left( \prod_{\substack{a=1 \\ \sigma(a) \neq N+1}}^{N} u^*_{n_a,m'_{\sigma(a)}} \right) \left( \prod_{\substack{a=1 \\ \sigma(a) \neq N+1}}^{N} u_{m_{\sigma(a)},m'_{\sigma(a)}} \right) \nonumber \\
        =& \sum_{\sigma \in S_{N+1}} \epsilon (\sigma) \left( \delta_{N+1, \sigma(N+1)} - 1 \right)  \mel**{0}{\hat{\gamma}_{n_{\sigma^{-1}(N+1)}} \hat{O} \hat{\gamma}^\dagger_{m_{\sigma(N+1)}}}{0} \prod_{\substack{a=1 \\ \sigma(a) \neq N+1}}^{N} \mel**{0}{\hat{\gamma}_{n_{a}} \hat{\gamma}^\dagger_{m_{\sigma(a)}}}{0}. \label{app:eqn:inner_prod_proof_perm_N1}
    \end{align}
    In the last line of \cref{app:eqn:inner_prod_proof_perm_N1}, we have employed the fact that 
    \begin{equation}
    	\mel**{0}{\hat{\gamma}_{n_{n}} \hat{O} \hat{\gamma}^\dagger_{m_{m}}}{0} = \sum_{n',m'} O_{n' m'} u^{*}_{n n'} u_{m m'} \qq{and}
    	\mel**{0}{\hat{\gamma}_{n_{n}} \hat{\gamma}^\dagger_{m_{m}}}{0} = \sum_{i} u^{*}_{n i} u_{m i}.
    \end{equation}
  
    All that remains is to rewrite the sum in \cref{app:eqn:inner_prod_proof_perm_N1} as a sum over permutations in $S_N$ rather than $S_{N+1}$. To do so, we first observe that the terms in \cref{app:eqn:inner_prod_proof_perm_N1} are nonzero only for permutations $\sigma \in S_{N+1}$ where $\sigma(N+1) \neq N+1$. Any such permutation can be generated from a permutation in $S_N$ by multiplying it with a transposition. Denoting by $\left[i, j\right]$ the transposition that swaps $i$ and $j$, the following holds
    \begin{align}
        \sigma \in S_{N+1}, \qq{with} \sigma(N+1) \neq N+1 &\implies \sigma \left[N+1, \sigma^{-1}(N+1) \right] \in S_{N}, \nonumber \\
        \sigma' \in S_{N}, \qq{and} j \in \mathbb{Z} \qq{with} 1 \leq j \leq N &\implies \sigma = \sigma' \left[N+1, j \right]  \in S_{N+1}, \qq{with} \sigma(N+1) \neq N+1. \label{app:eqn:permutation_mapping}
    \end{align}
    This one-to-one correspondence between the sets of permutations enables us to replace the summation in \cref{app:eqn:inner_prod_proof_perm_N1} as
    \begin{equation}
        \sum_{\sigma \in S_{N+1}} \epsilon (\sigma) \left( \delta_{N+1, \sigma(N+1)} - 1 \right) \to \sum_{\sigma' \in S_{N}} \sum_{j=1}^{N} \epsilon \left( \sigma' \right),
    \end{equation}
    where $\sigma = \sigma' \left[N+1, j \right]$. Under this substitution, the following relations hold
    \begin{gather}
        \sigma^{-1} (N+1) = \left[N+1, j \right] \left(\sigma'\right)^{-1} (N+1) = j \\
        \sigma (N+1) =  \sigma' \left[N+1, j \right] (N+1) = \sigma' (j) \\
        \sigma(a) \neq N+1 \iff a \neq j \\
        \sigma(a) = \sigma'(a), \qq{for} 1 \leq a \leq N \qq{and} a \neq j.
    \end{gather}
    Using these substitutions, \cref{app:eqn:inner_prod_proof_perm_N1} simplifies to
    \begin{equation}
        \mel**{0}{\hat{\gamma}_{n_{N}} \dots \hat{\gamma}_{n_{2}} \hat{\gamma}_{n_{1}} \hat{O} \hat{\gamma}^\dagger_{m_{1}} \hat{\gamma}^\dagger_{m_{2}} \dots \hat{\gamma}^\dagger_{m_{N}}}{0} = \sum_{\sigma' \in S_{N}} \sum_{j=1}^{N} \epsilon \left( \sigma' \right) \mel**{0}{\hat{\gamma}_{n_{j}} \hat{O} \hat{\gamma}^\dagger_{m_{\sigma'(j)}}}{0} \prod_{\substack{a=1 \\ a \neq j}}^{N} \mel**{0}{\hat{\gamma}_{n_{a}} \hat{\gamma}^\dagger_{m_{\sigma'(a)}}}{0},
    \end{equation}
    which is just \cref{app:eqn:many_body_overlap}.
\end{proof}

We now have all the necessary components to evaluate the many-body matrix element $\bra{\Phi_{\mu}} \left( \hat{H} - \mathcal{E}_{0} \right) \ket{\Phi_0}$. According to \cref{app:eqn:many_body_overlap}, this matrix element can be expressed as a sum of products of one-body matrix elements and one-body overlaps. Since $\ket{\Phi_0}$ and $\ket{\Phi_{\mu}}$ differ by exactly one single-particle state, we observe that, depending on the permutation, each term in this sum will include a given number of off-diagonal Hamiltonian matrix elements and off-diagonal overlaps. This property allows us to derive an analytically tractable expression by relying on three key assumptions:
\begin{itemize}
	\item First, we assume that the bound states on the left and right metallic regions are approximately orthogonal, as per \cref{app:eqn:approx_orthogonality},
	\begin{equation}
		\abs{\bra{0} \hat{a}_{R/L,m} \hat{a}^\dagger_{L/R,n} \ket{0}} \ll 1.
	\end{equation}
	\item Second, since the bound states $\hat{a}^\dagger_{L/R,b} \ket{0}$ are mainly localized in one metallic region, we can approximate
	\begin{equation}
		\bra{0} \hat{a}_{L/R,n} \hat{H} \hat{a}^\dagger_{L/R,n} \ket{0} \approx \bra{0} \hat{a}_{L/R,n} \hat{H}_{L/R} \hat{a}^\dagger_{L/R,n} \ket{0} = \epsilon_{L/R,n}.
	\end{equation}
	\item Finally, we assume that the tunneling between the two regions is weak,
	\begin{equation}
		\abs{\bra{0} \hat{a}_{R/L,m} \hat{H} \hat{a}^\dagger_{L/R,n} \ket{0}} \sim \abs{\epsilon_{L/R,n} \bra{0} \hat{a}_{R/L,m} \hat{a}^\dagger_{L/R,n} \ket{0}} \sim \abs{\epsilon_{R/L,m} \bra{0} \hat{a}_{R/L,m} \hat{a}^\dagger_{L/R,n} \ket{0}} \ll \abs{\epsilon_{L/R,n}}, \abs{\epsilon_{R/L,m}}.
	\end{equation}
\end{itemize}
Working to leading order in the small overlap between the bound states on the left and right metallic regions, there is only one permutation that corresponds to a single off-diagonal matrix element or a single off-diagonal overlap
\begin{equation}
	\bra{\Phi_\mu} \hat{H} \ket{\Phi_0} \approx \bra{0}\hat{a}_{R,m_0} \hat{H}  \hat{a}^\dagger_{L,n_1}\ket{0} + \left( \epsilon_{L,n_2} + \epsilon_{L,n_3} + \dots + \epsilon_{R,m_1} + \epsilon_{R,m_2} + \dots \right) \bra{0}\hat{a}_{R,m_0}  \hat{a}^\dagger_{L,n_1}\ket{0}.
\end{equation}
This result allows us to conclude that
\begin{equation}
	\bra{\Phi_{\mu}}\left( \hat{H} - \mathcal{E}_{0} \right) \ket{\Phi_0} \approx \bra{0} \hat{a}_{R,m_0} \left( \hat{H} - \epsilon_{L,n_1} \right) \hat{a}^\dagger_{L,n_1} \ket{0} . \label{app:eqn:bardeen_mat_elem}
\end{equation}
Thus, the many-body matrix element $\bra{\Phi_{\mu}}\left( \hat{H} - \mathcal{E}_{0} \right) \ket{\Phi_0}$ is approximately given by the single-particle matrix element $\bra{0} \hat{a}_{R,m_0} \left( \hat{H} - \epsilon_{L,n_1} \right) \hat{a}^\dagger_{L,n_1} \ket{0}$. In other words, the tunneling rate between \emph{any} two Slater many-body states $\ket{\Phi_{0}}$ and $\ket{\Phi_{\mu}}$ -- where $\ket{\Phi_{\mu}}$ is identical to $\ket{\Phi_{0}}$ except for one electron, originally in state $\hat{a}^\dagger_{L,n_1}$, which has tunneled to state $\hat{a}^\dagger_{R,m_0}$ -- is governed by the single-particle matrix element $\bra{0} \hat{a}_{R,m_0} \left( \hat{H} - \epsilon_{L,n_1} \right) \hat{a}^\dagger_{L,n_1} \ket{0}$. Our next goal is to derive a simpler form for this \emph{single-particle} matrix element.

\subsubsection{Simplifying the tunneling single-particle matrix element}\label{app:sec:theory_stm:bardeen:tunnelling_matrix_element_sp}

We begin by expressing \cref{app:eqn:bardeen_mat_elem} in the position basis
\begin{align}
	\bra{0} \hat{a}_{R,m_0} \left( \hat{H} - \epsilon_{L,n_1} \right) \hat{a}^\dagger_{L,n_1} \ket{0} &= \int \dd[3]{r} \bra{0} \hat{a}_{R,m_0} \ket{\vec{r}} \left( - \frac{1}{2m_e} \nabla^2 + V_{\text{eff}} \left( \vec{r} \right) - \epsilon_{L,n_1} \right)  \bra{\vec{r}} \hat{a}^\dagger_{L,n_1} \ket{0} \nonumber \\
	&= \int \dd[3]{r} \phi^{*}_{R,m_0} \left( \vec{r} \right) \left( - \frac{1}{2m_e} \nabla^2 + V_{\text{eff}} \left( \vec{r} \right) - \epsilon_{L,n_1} \right) \phi_{L,n_1} \left( \vec{r} \right) \nonumber \\
	&= \int_{x>0} \dd[3]{r} \phi^{*}_{R,m_0} \left( \vec{r} \right) \left( - \frac{1}{2m_e} \nabla^2 + V_{\text{eff}} \left( \vec{r} \right) - \epsilon_{L,n_1} \right) \phi_{L,n_1} \left( \vec{r} \right) \nonumber \\
	&+ \int_{x<0} \dd[3]{r} \phi^{*}_{R,m_0} \left( \vec{r} \right) \left( - \frac{1}{2m_e} \nabla^2 + V_{\text{eff}} \left( \vec{r} \right) - \epsilon_{L,n_1} \right) \phi_{L,n_1} \left( \vec{r} \right).\label{app:eqn:mat_elem_position_basis_1}
\end{align}
For $x < 0$, where $V_{\text{eff}} \left( \vec{r} \right) \approx V_{L} \left( \vec{r} \right)$, the second term of \cref{app:eqn:mat_elem_position_basis_1} vanishes because $\phi_{L,n_1} \left( \vec{r} \right)$ satisfies Schr\"odinger's equation. In the $x > 0$ region, where $V_{\text{eff}} \left( \vec{r} \right) \approx V_{R} \left( \vec{r} \right)$, we obtain
\begin{equation}
	\label{app:eqn:mat_elem_position_basis_2}
	\bra{0} \hat{a}_{R,m_0} \left( \hat{H} - \epsilon_{L,n_1} \right) \hat{a}^\dagger_{L,n_1} \ket{0} \approx \int_{x>0} \dd[3]{r} \phi^{*}_{R,m_0} \left( \vec{r} \right) \left( - \frac{1}{2m_e} \nabla^2 + V_{R} \left( \vec{r} \right) - \epsilon_{L,n_1} \right) \phi_{L,n_1} \left( \vec{r} \right).
\end{equation}
Since $\phi^{*}_{R,m_0} \left( \vec{r} \right)$ satisfies the Schr\"odinger equation for $\hat{H}_{R}$,
\begin{equation}
	\left( - \frac{1}{2m_e} \nabla^2 + V_{R} \left( \vec{r} \right) \right)\phi^{*}_{R,m_0} \left( \vec{r} \right) = \epsilon_{R,m_0} \phi^{*}_{R,m_0} \left( \vec{r} \right),
\end{equation}
we can rewrite \cref{app:eqn:mat_elem_position_basis_2} in a more symmetric form
\begin{align}
	\bra{0} \hat{a}_{R,m_0} \left( \hat{H} - \epsilon_{L,n_1} \right) \hat{a}^\dagger_{L,n_1} \ket{0} \approx& \int_{x>0} \dd[3]{r} \left[ \phi^{*}_{R,m_0} \left( \vec{r} \right) \left( - \frac{1}{2m_e} \nabla^2 + V_{R} \left( \vec{r} \right) - \epsilon_{L,n_1} \right) \phi_{L,n_1} \left( \vec{r} \right) \right. \nonumber \\
	&- \left. \left( - \frac{1}{2m_e} \nabla^2 \phi^{*}_{R,m_0} \left( \vec{r} \right) + V_{R} \left( \vec{r} \right)  \phi^{*}_{R,m_0} \left( \vec{r} \right) - \epsilon_{R,m_0}  \phi^{*}_{R,m_0} \left( \vec{r} \right) \right) \phi_{L,n_1} \left( \vec{r} \right) \right] \nonumber \\
	\approx& \int_{x>0} \dd[3]{r} \left[ \frac{-1}{2m_e} \left( \phi^{*}_{R,m_0} \left( \vec{r} \right)  \nabla^2 \phi_{L,n_1} \left( \vec{r} \right) - \phi_{L,n_1} \left( \vec{r} \right)  \nabla^2 \phi^{*}_{R,m_0} \left( \vec{r} \right)  \right) \right. \nonumber \\
	&+ \left. \left( \epsilon_{R,m_0} - \epsilon_{L,n_1} \right) \phi^{*}_{R,m_0} \left( \vec{r} \right) \phi_{R,n_1} \left( \vec{r} \right) \right] \label{app:eqn:mat_elem_position_basis_3}
\end{align}
To simplify \cref{app:eqn:mat_elem_position_basis_3} further, we take into account the Dirac $\delta$-function from \cref{app:eqn:bardeen_fermi_golden_rule},
\begin{equation}
	\delta \left( \mathcal{E}_{\mu} - \mathcal{E}_0 \right) = \delta \left( \epsilon_{R,m_0} - \epsilon_{L,n_1} \right),
\end{equation}
which implies that the energies of the states $\hat{a}^\dagger_{L,n_1}$ and $\hat{a}^\dagger_{R,m_0}$ must be equal, {\it i.e.}{}, $\epsilon_{R,m_0} = \epsilon_{L,n_1}$. Thus, we find that
\begin{align}
	\bra{0} \hat{a}_{R,m_0} \left( \hat{H} - \epsilon_{L,n_1} \right) \hat{a}^\dagger_{L,n_1} \ket{0} \approx& \int_{x>0} \dd[3]{r} \frac{-1}{2m_e} \nabla \cdot \left( \phi^{*}_{R,m_0} \left( \vec{r} \right)  \nabla \phi_{L,n_1} \left( \vec{r} \right) - \phi_{L,n_1} \left( \vec{r} \right)  \nabla \phi^{*}_{R,m_0} \left( \vec{r} \right)  \right) \nonumber \\
	\approx& - i \int_{x=0} \left( \phi^{*}_{R,m_0} \left( \vec{r} \right)  \frac{-i \nabla}{2m_e} \phi_{L,n_1} \left( \vec{r} \right) - \phi_{L,n_1} \left( \vec{r} \right)  \frac{-i \nabla}{2m_e} \phi^{*}_{R,m_0} \left( \vec{r} \right)  \right) \cdot \dd{\vec{S}} \nonumber \\
	\approx& - i \int_{x=0} \bra{0} \hat{a}_{R,m_0} \hat{\vec{j}} \left( \vec{r} \right) \hat{a}^\dagger_{L,n_1} \ket{0}  \cdot \dd{\vec{S}}, \label{app:eqn:mat_elem_position_basis_4}
\end{align}
where $\hat{\vec{j}} \left( \vec{r} \right)$ is the current density operator, represented in the position basis as
\begin{equation}
	\label{app:eqn:current_density_operator_pos_rep}
	\hat{\vec{j}} \left( \vec{r} \right) = \int \dd[3]{r'} \delta \left( \vec{r} - \vec{r}' \right) \left( \frac{-i \nabla}{2 m_e} - \frac{-i \nabla'}{2 m_e} \right) \ket{\vec{r}'}  \bra{\vec{r}}.
\end{equation}
In \cref{app:eqn:mat_elem_position_basis_4}, the surface element $\dd{\vec{S}}$ points in the negative $x$ direction. As a summary of this section, we note that the tunneling rate between the state $\ket{\Phi_0}$ containing quasiparticles in two metallic regions and the state $\ket{\Phi_{\mu}} = \hat{a}^\dagger_{R,m_0}\hat{a}_{L,n_1} \ket{\Phi_0}$, which differs from the state $\ket{\Phi_0}$ by the transfer of one particle from the state $\hat{a}^\dagger_{L,n_1}$ to the state $\hat{a}^\dagger_{R,m_0}$ is given by~\cite{BAR61}
\begin{align}
	\dv{P_{\mu} (t)}{t} &= 2 \pi \abs{ \bra{0} \hat{a}_{R,m_0} \left( \int \dd{\vec{S}} \cdot \hat{\vec{j}} \left( \vec{r} \right) \right) \hat{a}^\dagger_{L,n_1}\ket{0} }^2 \delta \left( \epsilon_{R,m_0} - \epsilon_{L,n_1} \right), \nonumber \\
	&= 2 \pi \abs{ \bra{\Phi_\mu} \hat{J} \ket{\Phi_0} }^2 \delta \left( \mathcal{E}_{\mu} - \mathcal{E}_{0} \right),
	\label{app:eqn:bardeen_fermi_golden_rule_final}
\end{align}
where we have introduced the one-body operator corresponding to the current through the interface between the two metals,
\begin{equation}
	\label{app:eqn:def_current_through_interface}
	\hat{J} = \sum_{m,n} \bra{0} \hat{a}_{R,m} \left( \int \dd{\vec{S}} \cdot \hat{\vec{j}} \left( \vec{r} \right) \right) \hat{a}^\dagger_{L,n}\ket{0} \hat{a}^\dagger_{R,m} \hat{a}_{L,n} + \text{h.c.}.
\end{equation}
In \cref{app:eqn:def_current_through_interface}, ``$+\text{h.c.}$'' denotes the addition of the Hermitian conjugate. 

\subsection{The Tersoff-Hamann approximation}\label{app:sec:theory_stm:TH_approx}

\begin{figure}[t]
	\centering
	\includegraphics[width=0.5\textwidth]{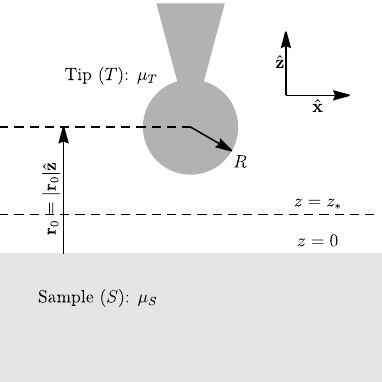}
	\caption{STM setup considered by the Tersoff-Hamann approximation. The tunneling setup is analogous to the one considered in Bardeen's formula for the tunneling current from \cref{app:fig:bardeen_geometry}, but the left and right metallic regions have been replaced by the sample and the tip, respectively. The tip consists of a metallic sphere of radius $R$ at electrochemical potential $\mu_T$. The sample is located in the $z=0$ plane and has chemical potential $\mu_S$. The center of the tip is located at $\vec{r}_0 = \abs{\vec{r}_0} \hat{\vec{z}}$. The dashed line at $z=z_*$ represents the integration surface used to compute the tunneling current between the sample and the tip, as will be done in \cref{app:TH_matrix_element_final}.}
	\label{app:fig:STM_setup}
\end{figure}

Having derived Bardeen's formula for the tunneling current, we now focus on the STM setup described in \cref{app:fig:STM_setup}. In an STM experiment, a metallic tip is brought close to the sample, and the resulting tunneling current between the tip and the sample is measured. The Tersoff-Hamann approximation~\cite{TER83,TER85} models the tip and sample as metallic regions separated by a vacuum layer. Using Bardeen's formula (with additional approximations detailed below), an analytical expression for the tunneling current in STM can be derived.

\subsubsection{Notation}\label{app:sec:theory_stm:TH_approx:notation}

As in \cref{app:sec:theory_stm:bardeen}, we assume that both the sample and the tip admit an \emph{effectively} non-interaction description. We employ a similar notation to \cref{app:sec:theory_stm:bardeen:left_right_state}, but replace the $L$ (left) and $R$ (right) indices of all the quantities with $S$ (sample) and $T$ (tip), respectively. The spectral functions for the sample and tip are defined as~\cite{COL15}
\begin{align}
	\mathcal{A}_{S,n} (\omega) &= \sum_{\lambda_S,\xi_S} P_{\lambda_S} \left( \abs{\mel**{\xi_S}{\hat{a}^\dagger_{S,n}}{\lambda_S}}^2 \delta \left(\omega - E_{\xi_S} + E_{\lambda_S} \right) + \abs{\mel**{\xi_S}{\hat{a}_{S,n}}{\lambda_S}}^2 \delta \left(\omega + E_{\xi_S} - E_{\lambda_S} \right) \right), 
	\label{app:eqn:def_spect_func_samp} \\
	\mathcal{A}_{T,n} (\omega) &= \sum_{\lambda_T,\xi_T} P_{\lambda_T} \left( \abs{\mel**{\xi_T}{\hat{a}^\dagger_{T,n}}{\lambda_T}}^2 \delta \left(\omega - E_{\xi_T} + E_{\lambda_T} \right) + \abs{\mel**{\xi_T}{\hat{a}_{T,n}}{\lambda_T}}^2 \delta \left(\omega + E_{\xi_T} - E_{\lambda_T} \right) \right), 
	\label{app:eqn:def_spect_func_tip}
\end{align}
respectively. Since both the tip and the sample are assumed to follow an effectively non-interacting description (where interactions are treated at the mean-field level), the spectral function $\mathcal{A}_{S/T,n} (\omega)$ is diagonal in the excitation index $n$. The summations in  \cref{app:eqn:def_spect_func_samp,app:eqn:def_spect_func_tip} span the exact eigenstates of the sample and the tip, respectively.

Under the non-interacting assumption, the exact states $\ket{\lambda_{S/T}}$ are Slater determinants written in terms of the $\hat{a}^\dagger_{S/T,n}$ operators. The corresponding energies are given by
\begin{equation}
	\hat{H}_{S/T} \ket{\lambda_{S/T}} = \mathcal{E}_{\lambda_{S/T}} \ket{\lambda_{S/T}}, \quad
	\hat{K}_{S/T} \ket{\lambda_{S/T}} = E_{\lambda_{S/T}} \ket{\lambda_{S/T}},
\end{equation}
where $\hat{K}_{S/T} = \hat{H}_{S/T} - \mu_{S/T} \hat{N}_{S/T}$ is the grand canonical Hamiltonian for the sample or the tip. Here, $\mu_{S/T}$ and $\hat{N}_{S/T}$ denote the chemical potential and number operator, respectively. For a state $\ket{\lambda_{S/T}}$ with $N_{\lambda_{S/T}}$ electrons ({\it i.e.}{}, $\hat{N}_{S/T} \ket{\lambda_{S/T}} = N_{\lambda_{S/T}} \ket{\lambda_{S/T}}$), the microcanonical and grand canonical energies are related by $E_{\lambda_{S/T}} = \mathcal{E}_{\lambda_{S/T}} - \mu_{S/T} N_{\lambda_{S/T}}$.

In \cref{app:eqn:def_spect_func_samp,app:eqn:def_spect_func_tip}, we introduced the Boltzmann factor
\begin{equation}
	P_{\lambda_{S/T}} = \frac{e^{- \beta E_{\lambda_{S/T}}}}{Z_{S/T}},
\end{equation} 
where the partition function $Z_{S/T}$ is defined such that the probability distribution is normalized
\begin{equation}
	\sum_{\lambda_{S/T}} P_{\lambda_{S/T}} = 1,
\end{equation}
and $\beta = \frac{1}{T}$ is the inverse temperature.

Using the expression for the Boltzmann factor, we can also rewrite \cref{app:eqn:def_spect_func_samp} as
\begin{align}
	\mathcal{A}_{S,n} (\omega) =& \frac{1}{Z_S} \sum_{\lambda_S,\xi_S} \left( e^{- \beta E_{\lambda_{S}}} \abs{\mel**{\xi_S}{\hat{a}^\dagger_{S,n}}{\lambda_S}}^2 \delta \left(\omega - E_{\xi_S} + E_{\lambda_S} \right) + e^{- \beta E_{\lambda_{S}}} \abs{\mel**{\lambda_S}{\hat{a}^\dagger_{S,n}}{\xi_S}}^2 \delta \left(\omega + E_{\xi_S} - E_{\lambda_S} \right) \right) \nonumber \\
	=& \sum_{\lambda_S,\xi_S} \frac{e^{- \beta E_{\lambda_{S}}} }{Z_S} \left[ 1 + e^{- \beta \left( E_{\xi_{S}} - E_{\lambda_{S}} \right) }  \right] \abs{\mel**{\xi_S}{\hat{a}^\dagger_{S,n}}{\lambda_S}}^2 \delta \left(\omega - E_{\xi_S} + E_{\lambda_S} \right) \nonumber \\
	=& \sum_{\lambda_S,\xi_S} P_{\lambda_{S}} \left( 1 + e^{- \beta \omega }  \right) \abs{\mel**{\xi_S}{\hat{a}^\dagger_{S,n}}{\lambda_S}}^2 \delta \left(\omega - E_{\xi_S} + E_{\lambda_S} \right) \label{app:eqn:first_form_of_simpler_spectral_function} \\
	=& \sum_{\lambda_S,\xi_S} P_{\xi_{S}} \left( 1 + e^{- \beta \omega }  \right) \abs{\mel**{\xi_S}{\hat{a}_{S,n}}{\lambda_S}}^2 \delta \left(\omega + E_{\xi_S} - E_{\lambda_S} \right) \nonumber \\
	=& \sum_{\lambda_S,\xi_S} P_{\lambda_{S}} e^{\beta \omega} \left( 1 + e^{- \beta \omega }  \right) \abs{\mel**{\xi_S}{\hat{a}_{S,n}}{\lambda_S}}^2 \delta \left(\omega + E_{\xi_S} - E_{\lambda_S} \right). \label{app:eqn:second_form_of_simpler_spectral_function}
\end{align}
\Cref{app:eqn:first_form_of_simpler_spectral_function,app:eqn:second_form_of_simpler_spectral_function} allow us to relate the thermodynamic average over the matrix elements to the spectral function through
\begin{align}
	\sum_{\lambda_S,\xi_S} P_{\lambda_{S}}  \abs{\mel**{\xi_S}{\hat{a}^\dagger_{S,n}}{\lambda_S}}^2 \delta \left(\omega - E_{\xi_S} + E_{\lambda_S} \right) &= \mathcal{A}_{S,n} ( \omega ) \left( 1 - n_F ( \omega ) \right), \label{app:eqn:mat_elem_sp_func_fermi_1}\\
	\sum_{\lambda_S,\xi_S} P_{\lambda_{S}} \abs{\mel**{\xi_S}{\hat{a}_{S,n}}{\lambda_S}}^2 \delta \left(\omega + E_{\xi_S} - E_{\lambda_S} \right) &=  \mathcal{A}_{S,n} ( \omega ) n_F ( \omega ), \label{app:eqn:mat_elem_sp_func_fermi_2}
\end{align}
and similarly for the tip spectral function. In \cref{app:eqn:mat_elem_sp_func_fermi_1,app:eqn:mat_elem_sp_func_fermi_2}, $n_F (\omega) = \frac{1}{e^{\beta \omega} + 1}$ denotes the Fermi occupation function.

The expressions in \cref{app:eqn:mat_elem_sp_func_fermi_1,app:eqn:mat_elem_sp_func_fermi_2} can be further simplified under the non-interacting assumption. Considering the matrix elements in \cref{app:eqn:mat_elem_sp_func_fermi_1} for a given $\ket{\lambda_S}$, there is only one Slater-determinant eigenstate $\ket{\xi_S}$ for which the matrix element $\mel**{\xi_S}{\hat{a}^\dagger_{S,n}}{\lambda_S}$ is non-vanishing, namely $\ket{\xi_S} = \hat{a}^\dagger_{S,n} \ket{\lambda_S}$. When $\ket{\xi_S} = \hat{a}^\dagger_{S,n} \ket{\lambda_S}$, we have $E_{\xi_S} - E_{\lambda_S} = \epsilon_{S,n} - \mu_S$. Thus, \cref{app:eqn:mat_elem_sp_func_fermi_1,app:eqn:mat_elem_sp_func_fermi_2} can be simplified to
\begin{align}
	\mathcal{A}_{S,n} ( \omega ) \left( 1 - n_F ( \omega ) \right) &= \sum_{\lambda_s} P_{\lambda_{S}} \abs{\mel**{\lambda_S}{\hat{a}_{S,n}\hat{a}^\dagger_{S,n}}{\lambda_S}}^2 \delta \left( \omega - \epsilon_{S,n} + \mu_S \right),\label{app:eqn:sp_func_non_int_simple_expr_el} \\
	\mathcal{A}_{S,n} ( \omega ) n_F ( \omega ) &= \sum_{\lambda_s} P_{\lambda_{S}} \abs{\mel**{\lambda_S}{\hat{a}^\dagger_{S,n} \hat{a}_{S,n}}{\lambda_S}}^2 \delta \left( \omega - \epsilon_{S,n} + \mu_S \right), \label{app:eqn:sp_func_non_int_simple_expr_ho}
\end{align}
and similarly for the tip.

\subsubsection{The tunneling current}\label{app:sec:theory_stm:TH_approx:tip_sample_current}

We are now ready to derive the expression for the net tunneling current between the tip and the sample. Assume that, in the absence of tunneling, the entire system is in the state $\ket{\lambda_{S} \otimes \lambda_{T}}$. If an electron in state $\hat{a}^\dagger_{S,n}$ tunnels from the sample to the tip and occupies the state $\hat{a}^\dagger_{T,m}$, the system transitions to the state $\ket{\xi_{S} \otimes \xi_{T}} = \hat{a}^\dagger_{T,m} \hat{a}_{S,n} \ket{\lambda_{S} \otimes \lambda_{T}}$. The current arising from this process is just the tunneling probability rate from \cref{app:eqn:bardeen_fermi_golden_rule_final}, which is given
\begin{align}
	\dv{P_{\ket{\lambda_{S} \otimes \lambda_{T}} \to \ket{\xi_{S} \otimes \xi_{T}}} (t)}{t} &= 2 \pi \abs{ \bra{0} \hat{a}_{T,m} \left( \int \dd{\vec{S}} \cdot \hat{\vec{j}} \left( \vec{r} \right) \right) \hat{a}^\dagger_{S,n}\ket{0} }^2 \delta \left( \epsilon_{T,m} - \epsilon_{S,n} \right)
\end{align}
To compute the total current from the sample to the tip, we perform a thermodynamic average over all possible system states. For a process where an electron tunnels from the state $\hat{a}^\dagger_{S,n}$ to the state $\hat{a}^\dagger_{T,m}$ (with tunneling amplitude $\bra{0} \hat{a}_{T,m} \left( \int \dd{\vec{S}} \cdot \hat{\vec{j}} \left( \vec{r} \right) \right) \hat{a}^\dagger_{S,n}\ket{0}$), we sum over all initial states $\ket{\lambda_{S} \otimes \lambda_{T}}$ where $\hat{a}^\dagger_{S,n}$ is occupied and $\hat{a}^\dagger_{T,m}$ is empty. The total current from the sample to the tip is then given by
\begin{align}
	I_{S \to T} =& 2 \pi e \sum_{n,m} \abs{ \bra{0} \hat{a}_{T,m} \left( \int \dd{\vec{S}} \cdot \hat{\vec{j}} \left( \vec{r} \right) \right) \hat{a}^\dagger_{S,n}\ket{0} }^2 \delta \left( \epsilon_{T,m} - \epsilon_{S,n} \right) \nonumber \\
	& \times \sum_{\lambda_S,\lambda_T} P_{\lambda_{S}} \abs{\mel**{\lambda_S}{\hat{a}^\dagger_{S,n} \hat{a}_{S,n}}{\lambda_S}}^2 P_{\lambda_{T}} \abs{\mel**{\lambda_T}{\hat{a}_{T,m} \hat{a}^\dagger_{T,m}}{\lambda_T}}^2 \nonumber \\
	=& 2 \pi e \sum_{n,m} \abs{ \bra{0} \hat{a}_{T,m} \left( \int \dd{\vec{S}} \cdot \hat{\vec{j}} \left( \vec{r} \right) \right) \hat{a}^\dagger_{S,n}\ket{0} }^2 \nonumber \\
	& \times \int \dd{\omega} \sum_{\lambda_S,\lambda_T} P_{\lambda_{S}} \abs{\mel**{\lambda_S}{\hat{a}^\dagger_{S,n} \hat{a}_{S,n}}{\lambda_S}}^2 \delta \left( \omega - \mu_{S} - \epsilon_{S,n} + \mu_{S} \right) P_{\lambda_{T}} \abs{\mel**{\lambda_T}{\hat{a}_{T,m} \hat{c}^\dagger_{T,m}}{\lambda_T}}^2 \delta \left( \omega - \mu_{T} - \epsilon_{T,m} + \mu_{T} \right) \nonumber \\
	=& 2 \pi e \sum_{n,m} \abs{ \bra{0} \hat{a}_{T,m} \left( \int \dd{\vec{S}} \cdot \hat{\vec{j}} \left( \vec{r} \right) \right) \hat{a}^\dagger_{S,n}\ket{0} }^2 \nonumber \\
	& \times \int \dd{\omega} \mathcal{A}_{S,n} \left( \omega - \mu_S \right) n_F \left( \omega - \mu_S \right) \mathcal{A}_{T,n} \left( \omega - \mu_T \right) \left( 1-n_F \left( \omega - \mu_T \right) \right), \label{app:eqn:sample_tip_current}
\end{align} 
where $e$ is the (negative) electronic charge. Similarly, the current from the tip to the sample is obtained by swapping $S \leftrightarrow T$ in \cref{app:eqn:sample_tip_current}
\begin{align}
	I_{T \to S} =& 2 \pi e \sum_{n,m} \abs{ \bra{0} \hat{a}_{S,n} \left( \int \dd{\vec{S}} \cdot \hat{\vec{j}} \left( \vec{r} \right) \right) \hat{a}^\dagger_{T,m}\ket{0} }^2 \nonumber \\
	& \times \int \dd{\omega} \mathcal{A}_{T,m} \left( \omega - \mu_T \right) n_F \left( \omega - \mu_T \right) \mathcal{A}_{S,m} \left( \omega - \mu_S \right) \left( 1-n_F \left( \omega - \mu_S \right) \right), \label{app:eqn:tip_sample_current}
\end{align}
The net current from the sample to the tip is therefore given by
\begin{align}
	I =& I_{S \to T} - I_{T \to S} \nonumber \\
	=& 2 \pi e \sum_{n,m} \abs{ \bra{0} \hat{a}_{T,m} \left( \int \dd{\vec{S}} \cdot \hat{\vec{j}} \left( \vec{r} \right) \right) \hat{a}^\dagger_{S,n}\ket{0} }^2 \nonumber \\
	& \times \int \dd{\omega} \mathcal{A}_{T,m} \left( \omega - \mu_T + \mu_S \right) \mathcal{A}_{S,n} \left( \omega \right) \left( n_F \left( \omega \right) - n_F \left( \omega - \mu_T + \mu_S \right) \right).
\end{align}
It is worth noting that no current flows in equilibrium ({\it i.e.}{}, when $\mu_T = \mu_S$): a bias voltage must be introduced to induce current flow. We assume that the tip is biased at potential $V$ relative to the sample (assumed to be at zero potential). As a result, $\mu_T - \mu_S = - e V$, which gives the net current from the sample to the tip as
\begin{equation}
	I (V) = 2 \pi e \sum_{n,m} \abs{ \bra{0} \hat{a}_{T,m} \left( \int \dd{\vec{S}} \cdot \hat{\vec{j}} \left( \vec{r} \right) \right) \hat{a}^\dagger_{S,n}\ket{0} }^2 \int \dd{\omega} \mathcal{A}_{T,m} \left( \omega - \abs{e} V \right) \mathcal{A}_{S,n} \left( \omega \right) \left( n_F \left( \omega \right) - n_F \left( \omega - \abs{e} V \right) \right). \label{app:eqn:final_tunneling_current_with_me}
\end{equation} 

So far, \cref{app:eqn:final_tunneling_current_with_me} is exact, as long as the assumptions of Bardeen's tunneling formula are satisfied (namely, weak tunneling, absence of spin-orbit coupling, and the existence of an effectively non-interacting description for both the sample and the tip). The Tersoff-Hamann approximation examines the matrix element in \cref{app:eqn:final_tunneling_current_with_me} in more detail and simplifies it by making additional assumptions about the tip. These assumptions are often justified, in the absence  of detailed information on the exact tip geometry in experiments.

\subsubsection{Approximating the tunneling current}\label{app:sec:theory_stm:TH_approx:approximation}

To compute the current matrix element from \cref{app:eqn:final_tunneling_current_with_me}, we begin by considering a plane (in between the tip and the sample) located at height $z$ above the surface of the sample. At a sufficient distance from the sample surface, an electron satisfies the Schr\"odinger equation for a free particle in vacuum. Following Refs.~\cite{TER83,TER85}, we assume that the electrostatic potential at sufficiently large distances from the sample (including this plane) quickly drops to its vacuum value (taken to be zero). The wave function of any bound state from the sample takes the form of a linear superposition of waves evanescent along the $\hat{\vec{z}}$ direction
\begin{equation}
	\label{app:eqn:wave_function_ansatz_TH}
	\phi_{S,n} \left( \vec{r}_{\parallel} + z \hat{\vec{z}} \right) = \frac{1}{(2 \pi)^2} \int \dd[2]{k_{\parallel}} \tilde{\phi}_{S,n} \left( \vec{k}_{\parallel} \right) e^{ - z \sqrt{ \kappa^2 + k_{\parallel}^2 } }e^{i \vec{k}_{\parallel} \cdot \vec{r}_{\parallel}},
\end{equation}
where 
\begin{equation}
	\kappa = 2 m_e \mathcal{W}
\end{equation} 
is related to the work function of the sample $\mathcal{W}$, and $\tilde{\phi}_{S,n} \left( \vec{k}_{\parallel} \right)$ are Fourier expansion coefficients
\begin{equation}
	\label{app:eqn:wave_function_ansatz_TH_inverse}
	\tilde{\phi}_{S,n} \left( \vec{k}_{\parallel} \right) = \frac{1}{(2 \pi)^2} e^{z \sqrt{ \kappa^2 + k_{\parallel}^2 } } \int \dd[2]{r_{\parallel}} \phi_{S,n} \left( \vec{r}_{\parallel} + z \hat{\vec{z}} \right) e^{-i \vec{k}_{\parallel} \cdot \vec{r}_{\parallel}}, \qq{for some} z>0.
\end{equation} 
The form in \cref{app:eqn:wave_function_ansatz_TH} is straightforward to understand. In \cref{app:eqn:final_tunneling_current_with_me}, for small bias voltages and temperatures much smaller than $\mathcal{W}$, only states near the Fermi level contribute to tunneling, for both the tip and the sample. Thus, the states $\hat{a}^\dagger_{S,n}$ relevant to tunneling are located close to the Fermi level of the sample and have energy $- \mathcal{W}$ relative to the vacuum. In the vacuum, such a state is a linear superposition of waves that decay along the $\hat{\vec{z}}$ direction but oscillate in the plane of the sample, as described by \cref{app:eqn:wave_function_ansatz_TH}.

So far, the only assumption we have made is that only the states near the Fermi energy contribute to tunneling for the sample. This assumption is justified by \cref{app:eqn:final_tunneling_current_with_me}, where, for small temperatures, $n_F \left( \omega \right) \approx 1 - \Theta \left( \omega \right)$, with $\Theta \left( \omega \right)$ denoting the Heaviside step function. Using this approximation
\begin{equation}
	I  (V) \approx - 2 \pi e \sum_{n,m} \abs{ \bra{0} \hat{a}_{T,m} \left( \int \dd{\vec{S}} \cdot \hat{\vec{j}} \left( \vec{r} \right) \right) \hat{a}^\dagger_{S,n}\ket{0} }^2 \int_{0}^{\abs{e} V} \dd{\omega} \mathcal{A}_{T,m} \left( \omega - \abs{e} V \right) \mathcal{A}_{S,n} \left( \omega \right).
\end{equation} 
This expression demonstrates that only states located within $\abs{eV}$ of the sample's Fermi energy contribute to tunneling.

The key simplifying assumption of the Tersoff-Hamann approximation concerns the wave function of the tip~\cite{TER83,TER85}. The tip is assumed to be composed of a spherically symmetric metal with the \emph{same} work function as the sample, $\mathcal{W}$. The radius of the tip is $R$, and its position is denoted by $\vec{r}_0$. For the tip-only system, outside the metal, Refs.~\cite{TER83,TER85} assume the following asymptotic form for the wave functions of the tip near its Fermi energy
\begin{equation}
	\label{app:eqn:tip_wave_function_ansatz}
	\phi_{T,m} \left( \vec{r} \right) = \frac{c_m}{\sqrt{\Omega_t}} \frac{ R e^{\kappa R} }{ \abs{\vec{r} - \vec{r}_0} } e^{ - \kappa \abs{\vec{r} - \vec{r}_0}},
\end{equation}
where $\Omega_t$ is the volume of the tip, and $c_m$ is an order-one complex constant. 

The asymptotic form in \cref{app:eqn:tip_wave_function_ansatz} requires some explanation. First, $\phi_{T,m} \left( \vec{r} \right)$ is the solution of the Schr\"odinger equation with energy $- \mathcal{W}$ outside the tip, exhibiting perfect spherical symmetry around the tip center $\vec{r}_0$. Specifically, we have
\begin{equation}
	- \frac{1}{2m_e} \nabla^2 \phi_{T,m} \left( \vec{r} \right) = - \mathcal{W} \phi_{T,m} \left( \vec{r} \right), \qq{for} \abs{\vec{r} - \vec{r}_0} > R.
\end{equation}
The assumed spherical symmetry is the key \emph{assumption} of the Tersoff-Hamann approximation~\cite{TER83,TER85}. In principle, there is no fundamental reason why the tip wave function should exhibit perfect spherical symmetry outside the tip. This assumption is justified \textit{a posteriori} in Ref.~\cite{TER85}, where it is argued that contributions to tunneling from tip wave functions with larger angular quantum numbers are significantly smaller than those arising from spherically symmetric wave functions (which have zero angular momentum). The best justification for assuming a spherically symmetric tip wave function stems from the lack of detailed microscopic information about the tip's shape and size. Regarding the constant $c_m$ in \cref{app:eqn:tip_wave_function_ansatz}, its order-one nature can be reasoned as follows: inside the tip, electron states near the Fermi energy are essentially plane wave states
\begin{equation}
	\label{app:eqn:wavf_inside_tip}
	\phi_{T,m} \left( \vec{r} \right) \approx \frac{1}{\sqrt{\Omega_t}} e^{i \vec{k}_m \cdot \left( \vec{r} - \vec{r}_0 \right)}, \qq{for} \abs{\vec{r} - \vec{r}_0} < R,
\end{equation} 
where $\abs{\vec{k}_m} = \kappa$, and we have assumed $\kappa R \gg 1$. Since $\phi_{T,m} \left( \vec{r} \right)$ is a bound state inside the tip, most of its spectral weight resides \emph{inside} the tip. Thus, \cref{app:eqn:wavf_inside_tip} is appropriately normalized, since
\begin{equation}
	1 = \int \dd[3]{r} \abs{\phi_{T,m} \left( \vec{r} \right)}^2 \approx \int_{\abs{r}<R} \dd[3]{r} \abs{\phi_{T,m} \left( \vec{r} \right)}^2. 
\end{equation}
The continuity of the wave function across the tip boundary ensures that $c_m$ is an order-one complex constant.

With these assumptions, computing the matrix element in \cref{app:eqn:final_tunneling_current_with_me} becomes straightforward. We start by computing the two-dimensional Fourier transform of the tip wave function, defined analogously to \cref{app:eqn:wave_function_ansatz_TH}
\begin{align}
	\int \dd[2]{\vec{r}_{\parallel}} \phi_{T,m} \left( \vec{r}_{\parallel} + z' \hat{\vec{z}} \right) e^{- i \vec{k}_{\parallel} \cdot \vec{r}_{\parallel}} 
	&= \int < \phi_{T,m} \left( \vec{r} \right) e^{- i \vec{k}_{\parallel} \cdot \vec{r}_{
	\parallel}} \delta \left( z' - z \right) \nonumber \\
	&=  \frac{1}{2 \pi} \int \dd{k_z} \int \dd[3]{r} \phi_{T,m} \left( \vec{r} \right) e^{- i \vec{k}_{\parallel} \cdot \vec{r}_{
		\parallel}} e^{- i k_z \left( z - z' \right)} \nonumber \\
	&= \frac{1}{2 \pi} \int \dd{k_z} \int \dd[3]{r} \phi_{T,m} \left( \vec{r} + \vec{r}_0 \right) e^{- i \vec{k} \cdot \vec{r}} e^{- i \vec{k} \cdot \vec{r}_0} e^{i k_z z'} \nonumber \\
	&= \frac{c_m R e^{\kappa R}}{\sqrt{\Omega_t}} \int \dd{k_z}  \frac{2}{\abs{\vec{k}_{\parallel}}^2 + k_z^2 + \kappa^2} e^{i k_z \left( z' -  \abs{\vec{r}_0} \right)} \nonumber \\
	&= \frac{c_m R e^{\kappa R}}{\sqrt{\Omega_t}}  \frac{2\pi}{ \sqrt{\abs{\vec{k}_{\parallel}}^2 + \kappa^2}} e^{ - \abs{ z' - \abs{\vec{r}_0} } \sqrt{\abs{\vec{k}_{\parallel}}^2 + \kappa^2} }. \label{app:eqn:fourier_trafo_tip_derivation}
\end{align}
where we have assumed, without loss of generality, that $\vec{r}_0 = \abs{\vec{r}_0} \hat{\vec{z}}$, as shown in \cref{app:fig:STM_setup}. From \cref{app:eqn:fourier_trafo_tip_derivation}, the tip wave function can be rewritten as 
\begin{equation}
	\label{app:eqn:tip_wavf_form}
	\phi_{T,m} \left( \vec{r}_{\parallel} + z \hat{\vec{z}} \right) = \frac{1}{2 \pi} \frac{c_m e^{\kappa R}}{\sqrt{\Omega_t}} \int \dd[2]{k_{\parallel}} \frac{1}{ \sqrt{\abs{\vec{k}_{\parallel}}^2 + \kappa^2}} e^{ - \abs{ z - \abs{\vec{r}_0} } \sqrt{\abs{\vec{k}_{\parallel}}^2 + \kappa^2} } e^{i \vec{k}_{\parallel} \cdot \vec{r}_{\parallel}}.
\end{equation}
Next, we substitute \cref{app:eqn:wave_function_ansatz_TH,app:eqn:tip_wavf_form} into the matrix element from \cref{app:eqn:final_tunneling_current_with_me}. As shown in \cref{app:fig:STM_setup}, we choose the integration surface to be a plane parallel to the sample, located at a height $z_*$ above the sample (assumed to be at $z=0$ without loss of generality). The current matrix element can be evaluated using the position basis representation of the current density operator from \cref{app:eqn:current_density_operator_pos_rep}
\begin{align}
	&\bra{0} \hat{a}_{T,m} \left( \int \dd{\vec{S}} \cdot \hat{\vec{j}} \left( \vec{r} \right) \right) \hat{a}^\dagger_{S,n}\ket{0} =  \nonumber \\  
	=&\frac{-i}{2m_e} \int \dd[2]{r_{\parallel}} \eval{ \left( \phi^{*}_{T,m} \left( \vec{r}_{\parallel} + z \hat{\vec{z}} \right) \partial_z \phi_{S,n} \left( \vec{r}_{\parallel} + z \hat{\vec{z}} \right) - \phi_{S,n} \left( \vec{r}_{\parallel} + z \hat{\vec{z}} \right) \partial_z \phi^{*}_{T,m} \left( \vec{r}_{\parallel} + z \hat{\vec{z}} \right)  \right)}_{z = z_*} \nonumber  \\ 
	=&\frac{i }{2m_e} \frac{1}{2 \pi} \frac{2 c_m R e^{\kappa R}}{\sqrt{\Omega_t}} \int \dd[2]{k_{\parallel}} \tilde{\phi}_{S,n} \left( \vec{k}_{\parallel} \right) e^{ - \abs{\vec{r}_0}  \sqrt{\abs{\vec{k}_{\parallel}}^2 + \kappa^2} } \nonumber \\ 
	=&\frac{i }{2m_e} \frac{4 \pi c_m R e^{\kappa R}}{\sqrt{\Omega_t}} \phi_{S,n} \left( \vec{r}_0 \right), \label{app:TH_matrix_element_final}
\end{align}
where the final expression shows the dependence of the current matrix element on the wave function of the sample, evaluated at the tip center $\vec{r}_0$. We highlight two key observations based on the final expression for the current matrix element:
\begin{itemize}
	\item First, as expected, the result is independent of $z_*$ (the height of the integration surface, as shown in \cref{app:fig:STM_setup}) confirming that the chosen plane does not influence the tunneling matrix element.
	\item Second, the matrix element is directly proportional to the amplitude of the sample wave function at the tip center, located at $\vec{r}_0$ (a height $\abs{\vec{r}_0}$ above the sample surface, rather than exactly at the surface). Because the sample wave function decays evanescently in the direction perpendicular to the surface, tunneling to and from states with large in-plane momenta $\abs{\vec{k}_{\parallel}} \gg \kappa$ is exponentially suppressed. This suppression arises from \cref{app:eqn:wave_function_ansatz_TH}, where the contribution of high-momentum components decreases rapidly with increasing $z$. 
\end{itemize}

It is important to note that the simplification of the current matrix elements in \cref{app:TH_matrix_element_final} relies on two critical assumptions: the spherical symmetry of the tip wave function and the assumption that the tip and the sample have identical work functions. Both these assumption result in the convenient cancelation of the square-root denominator in \cref{app:eqn:tip_wavf_form}. Substituting \cref{app:TH_matrix_element_final} into \cref{app:eqn:final_tunneling_current_with_me}, we arrive at
\begin{align}
	I (V) &= 2 \pi e \left( \frac{4 \pi R e^{\kappa R}}{2 m_e \sqrt{\Omega_t}} \right)^2 \int \dd{\omega} \left( \sum_{m} \abs{c_m}^2 \mathcal{A}_{T,m} \left( \omega - \abs{e} V \right) \right) \left( \sum_{n} \mathcal{A}_{S,n} \left( \omega \right) \abs{\phi_{S,n} \left( \vec{r}_0 \right)}^2 \right) \left( n_F \left( \omega \right) - n_F \left( \omega - \abs{e} V \right) \right) \nonumber \\
	&\approx 2 \pi e \left( \frac{4 \pi R e^{\kappa R}}{2 m_e \sqrt{\Omega_t}} \right)^2 \int \dd{\omega} g_T \left( \omega - \abs{e} V \right) \mathcal{A}_{S} \left( \vec{r}_0, \omega \right) \left( n_F \left( \omega \right) - n_F \left( \omega - \abs{e} V \right) \right)  \nonumber \\ 
	&\approx - 2 \pi e \left( \frac{4 \pi R e^{\kappa R}}{2 m_e \sqrt{\Omega_t}} \right)^2 g_T \left( 0 \right) \int_{0}^{\abs{e} V} \dd{\omega} \mathcal{A}_{S} \left( \vec{r}_0, \omega \right). \label{app:eqn:final_TH_expression_current}
\end{align}
In \cref{app:eqn:final_TH_expression_current}, we introduced the real-space spectral function of the sample
\begin{equation}
	\label{app:eqn:def_spect_func_rs}
	\mathcal{A}_S \left( \vec{r}, \omega \right) = \sum_{\lambda_S,\xi_S} P_{\lambda_S} \left( \abs{\mel**{\xi_S}{ \hat{\Psi}^\dagger_{} \left( \vec{r} \right)}{\lambda_S}}^2 \delta \left(\omega - E_{\xi_S} + E_{\lambda_S} \right) + \abs{\mel**{\xi_S}{ \hat{\Psi}_{} \left( \vec{r} \right) }{\lambda_S}}^2 \delta \left(\omega + E_{\xi_S} - E_{\lambda_S} \right) \right), 
\end{equation}
where $\hat{\Psi}^\dagger_{} \left( \vec{r} \right)$ denotes the creation operator for the fermionic field. Additionally, we defined the density of states of the tip
\begin{equation}
	g_T \left( \omega \right) = \sum_{m} \mathcal{A}_{T,m} \left( \omega \right).
\end{equation}
To derive the second line of \cref{app:eqn:final_TH_expression_current}, we approximated the order-one constant $c_m \approx 1$, assuming it to be state-independent and approximately equal to one. This allowed us to substitute the expression for the tip density of states in the first parenthesis of the integrand in the first line of \cref{app:eqn:final_TH_expression_current}. 

Additionally, by decomposing the fermionic field operator into the complete basis of sample states
\begin{equation}
	\hat{\Psi}^\dagger_{} \left( \vec{r} \right) = \sum_{n} \phi^{*}_{S,n} \left( \vec{r} \right) \hat{a}^\dagger_{S,n},
\end{equation}
and substituting the resulting expression into \cref{app:eqn:def_spect_func_rs}, we find
\begin{align}
	\mathcal{A}_S \left( \vec{r}, \omega \right) =& \sum_{\lambda_S,\xi_S} P_{\lambda_S}  \sum_{m,n} \ \left( \mel**{\lambda_S}{ \hat{a}_{S,m} }{\xi_S} \mel**{\xi_S}{ \hat{a}^\dagger_{S,n} }{\lambda_S} \phi_{m,S} \left( \vec{r} \right) \phi^{*}_{n,S} \left( \vec{r} \right) \delta \left(\omega - E_{\xi_S} + E_{\lambda_S} \right) \right. \nonumber \\
	&+\left. \sum_{m,n} \mel**{\lambda_S}{ \hat{a}^\dagger_{S,m} }{\xi_S} \mel**{\xi_S}{ \hat{a}_{S,n} }{\lambda_S} \phi^{*}_{m,S} \left( \vec{r} \right) \phi_{S,n} \left( \vec{r} \right) \delta \left(\omega + E_{\xi_S} - E_{\lambda_S} \right) \right) \nonumber \\
	=& \sum_{\lambda_S,\xi_S} P_{\lambda_S}  \sum_{n} \ \left( \abs{ \mel**{\xi_S}{ \hat{a}^\dagger_{S,n} }{\lambda_S}}^2 \abs{\phi_{S,n} \left( \vec{r} \right)}^2 \delta \left(\omega - E_{\xi_S} + E_{\lambda_S} \right) \right. \nonumber \\
	&+\left. \sum_{n} \abs{\mel**{\xi_S}{ \hat{a}_{S,n} }{\lambda_S}} \abs{\phi_{S,n} \left( \vec{r} \right)}^2 \delta \left(\omega + E_{\xi_S} - E_{\lambda_S} \right) \right) \nonumber \\
	=& \sum_{n} \mathcal{A}_{S,n} \left( \omega \right) \abs{\phi_{S,n} \left( \vec{r} \right)}^2,\label{app:eqn:relating_the_real_and_orb_SF}
\end{align}
where, in the second equality of \cref{app:eqn:relating_the_real_and_orb_SF}, we employed the non-interacting assumption for the sample. \Cref{app:eqn:relating_the_real_and_orb_SF} was similarly applied to \cref{app:eqn:def_spect_func_rs} to simplify the second parenthesis in the integrand of the first line.

Finally, in the last line of \cref{app:eqn:final_TH_expression_current}, we assumed that temperatures are small compared to the energy scales of the sample, approximating $n_F \left( \omega \right) \approx 1 - \Theta \left( \omega \right)$. Furthermore, we treated the density of states of the tip as constant around its Fermi energy, approximating $g_T \left( \omega \right) \approx g_T \left( 0 \right)$ for $0 \leq \abs{\omega} \leq \abs{eV}$. These approximations significantly simplify the final expression for the tunneling current.

\Cref{app:eqn:final_TH_expression_current} shows that, under a series of approximations discussed in this section, the STM tunneling current measures the spectral function integrated from zero ({\it i.e.}{}, the Fermi energy) to the bias voltage. Notice also that the current is positive for a positive bias voltage, as expected. Importantly, within these same approximations, the conductance at a given bias voltage is directly proportional to the spectral function of the sample at that bias voltage, as expressed by
\begin{equation}
	\dv{I (V)}{V} \approx - 2 \pi e \left( \frac{4 \pi R e^{\kappa R}}{2 m_e \sqrt{\Omega_t}} \right)^2 g_T \left( 0 \right) \mathcal{A}_{S} \left( \vec{r}_0, \abs{e} V \right).
\end{equation}

\subsubsection{The Tersoff-Hamann approximation for an entire isolated band}\label{app:sec:theory_stm:TH_approx:approximation:isolated_band}

Consider a system with a single (spinful) isolated band gapped from the rest of the spectrum\footnote{The discussion can be straightforwardly generalized to a set of bands.}. We assume the energy of the top and bottom of the band, relative to the Fermi energy, is given by $E_t$ and $E_b$, respectively (such that $E_{t/b} > 0$ if the top/bottom of the band is above the Fermi energy and $E_{t/b} < 0$ otherwise). No 
additional constraints are imposed on the band.

Suppose the tunneling current of the system is measured at a point $\vec{r}$ for bias voltages $V_t$ and $V_b$, where $\abs{e} V_{t/b} = E_{t/b}$. Subtracting the two tunneling currents and applying \cref{app:eqn:final_TH_expression_current}, we find 
\begin{equation}
	\label{app:eqn:current_for_band}
	I \left( V_t \right) - I \left( V_b \right) \approx - 2 \pi e \left( \frac{4 \pi R e^{\kappa R}}{2 m_e \sqrt{\Omega_t}} \right)^2 g_T \left( 0 \right) \int_{\abs{e} V_b}^{\abs{e} V_t} \dd{\omega} \mathcal{A}_{S} \left( \vec{r}, \omega \right).
\end{equation} 
We now examine the integral in \cref{app:eqn:current_for_band} in more detail and simplify its form under the non-interacting limit by using \cref{app:eqn:sp_func_non_int_simple_expr_el,app:eqn:sp_func_non_int_simple_expr_ho,app:eqn:relating_the_real_and_orb_SF}
\begin{align}
	A \left( \vec{r} \right) &= \int_{\abs{e} V_b}^{\abs{e} V_t} \dd{\omega} \mathcal{A}_{S} \left( \vec{r}, \omega \right) \nonumber \\
	&= \int_{E_b}^{E_t} \dd{\omega} \sum_{n} \abs{\phi_{S,n} \left( \vec{r} \right)}^2 \mathcal{A}_{S,n} \left( \omega \right) \nonumber \\
	&= \sum_{\lambda_s} \sum_{n} \abs{\phi_{S,n} \left( \vec{r} \right)}^2 P_{\lambda_{S}} \left( \abs{\mel**{\lambda_S}{\hat{a}_{S,n}\hat{a}^\dagger_{S,n}}{\lambda_S}}^2 + \abs{\mel**{\lambda_S}{\hat{a}^\dagger_{S,n} \hat{a}_{S,n}}{\lambda_S}}^2 \right) \int_{E_b}^{E_t} \dd{\omega} \delta \left( \omega - \epsilon_{S,n} + \mu_S \right) \nonumber \\
	&= \sum_{n} \abs{\phi_{S,n} \left( \vec{r} \right)}^2 \left( \Theta \left( E_t - \epsilon_{S,n} + \mu_S \right) - \Theta \left( E_b - \epsilon_{S,n} + \mu_S \right) \right). \label{app:eqn:relating_band_CDD_to_measurment}
\end{align}
The result of \cref{app:eqn:relating_band_CDD_to_measurment} is nothing but the charge density distribution (CDD) of the band of interest (which is located between $E_{b}$ and $E_{t}$). This can be easily understood by noting that $\abs{\phi_{S,n} \left( \vec{r} \right)}^2$ is the CDD of state $\hat{a}^\dagger_{S,n}$ and \cref{app:eqn:relating_band_CDD_to_measurment} sums the CDD corresponding to all the states within the flat band. Within the Tersoff-Hamann approximation, the current difference in \cref{app:eqn:current_for_band} is equal to the CDD of the band of interest, up to proportionality factors that depend on the STM setup.

An equivalent and more convenient expression for $A \left( \vec{r} \right)$, which will be used extensively in subsequent discussions, is given by the expectation value
\begin{equation}
	\label{app:eqn:ldos_for_single_band}
	A \left( \vec{r} \right) = \mel**{\Phi}{\hat{\Psi}^\dagger_{} \left( \vec{r} \right) \hat{\Psi}_{} \left( \vec{r} \right)}{\Phi},
\end{equation}
taken in the product state $\ket{\Phi}$, where only the band of interest is fully populated
\begin{equation}
	\label{app:eqn:initial_definition_fully_filled_band}
	\ket{\Phi} = \left( \prod_{\substack{n \\ E_b \leq \epsilon_{S,n} \leq E_t}} \hat{a}^\dagger_{S,n} \right) \ket{0}.
\end{equation}
In the following discussion, which focuses on the CDD of an entire isolated band, we will use the expression in \cref{app:eqn:ldos_for_single_band}, noting that the CDD can be extracted by performing tunneling current measurements at constant height for two bias voltages that bracket the energy dispersion of the band of interest.
 
\section{Extracting correlation functions from STM}\label{app:sec:correlator_general}

This \siSection{} builds on the results of \cref{app:sec:theory_stm} and explains how measurements of a diagonal operator, {\it i.e.}{}, the local fermionic density operator $\hat{\Psi}^\dagger_{} \left( \vec{r} \right) \hat{\Psi}_{} \left( \vec{r} \right)$, can be used to infer the \emph{off-diagonal} expectation values of Wannier orbital operators. The analysis presented here is \emph{completely} general and applicable to any crystalline system.

We begin by introducing a generic notation for crystalline systems. Next, we examine the CDD of an isolated band and show that it can be expressed as the tensor product between a spatial factor and the inter-orbital correlation functions computed for the isolated band, which we term the correlators of the band. 
We then analyze the effects of crystalline symmetries on the CDD, the spatial factor, and the correlation functions. Finally, we outline how the correlation functions of the band can be extracted from experimentally measured CDD in the general case.

\subsection{Notation}\label{app:sec:correlator_general:notation}
We consider a two-dimensional\footnote{The results here can be straightforwardly generalized to the surface of a three-dimensional crystal.} crystalline material and let $\hat{c}^\dagger_{\vec{R},i}$ denote the corresponding fermionic creation operators\footnote{In this \siSection{}, we focus solely on the sample and use $\hat{c}^\dagger_{\vec{R},i}$ to represent its electronic states. These operators, like $\hat{a}^\dagger_{S,n}$ in \cref{app:sec:theory_stm:TH_approx}, span the complete basis of sample states. However, they are expressed in the Wannier basis rather than the energy band basis.}: $\hat{c}^\dagger_{\vec{R},i}$ creates an electron located in the $i$-th Wannier orbital of the unit cell at $\vec{R}$. The Wannier orbital associated with $\hat{c}^\dagger_{\vec{R},i}$ is displaced by $\vec{r}_i$ from the origin of the unit cell (located at $\vec{R}$), and its real-space wave function is denoted by 
\begin{equation}
	W_{i} \left( \vec{r} - \vec{R} \right) = \mel**{0}{\hat{\Psi}_{} \left( \vec{r} \right)\hat{c}^\dagger_{\vec{R},i}}{0}	.
\end{equation}

We also introduce the momentum-space fermionic operators
\begin{align}
    \hat{c}^\dagger_{\vec{R},i} &= \frac{1}{\sqrt{N}} \sum_{\vec{k}} \hat{c}^\dagger_{\vec{k},i} e^{-i \vec{k} \cdot \left( \vec{R} + \vec{r}_i \right)}, \label{app:eqn:ft_def_c_to_r} \\
    \hat{c}^\dagger_{\vec{k},i} &= \frac{1}{\sqrt{N}} \sum_{\vec{R}} \hat{c}^\dagger_{\vec{R},i} e^{i \vec{k} \cdot \left( \vec{R} + \vec{r}_i \right)}, \label{app:eqn:ft_def_c_to_k}
\end{align}
as well as the Fourier transform of the Wannier orbital wave function
\begin{align}
    W_{i} \left( \vec{r} \right) &= \frac{1}{\sqrt{N \Omega_0}} \sum_{\vec{k},\vec{Q}} W_{i} \left( \vec{k} + \vec{Q}, z \right) e^{i \left( \vec{k} + \vec{Q} \right) \cdot \vec{r}} \label{app:eqn:ft_def_w_to_r} \\
    W_{i} \left( \vec{k} + \vec{Q}, z \right) &= \frac{1}{\sqrt{N \Omega_0}} \int \dd[2]{r_{\parallel}} W_{i} \left( \vec{r} \right) e^{-i \left( \vec{k} + \vec{Q} \right) \cdot \vec{r}}. \label{app:eqn:ft_def_w_to_k} 
\end{align}
In \crefrange{app:eqn:ft_def_c_to_r}{app:eqn:ft_def_w_to_k}, $N$ denotes the number of unit cells, while $\Omega_0$ represents their surface area. We use the convention where momenta in the first Brillouin zone (BZ) are represented by bold lowercase letters, while reciprocal lattice momenta are represented by bold uppercase letters. Consequently, the summation over $\vec{k}$ runs over the first BZ, while the summation over $\vec{Q}$ spans the infinite reciprocal lattice.
 
It is important to note that the lattice vectors $\vec{R}$, as well as the crystalline momenta $\vec{k}$ and reciprocal lattice vectors $\vec{Q}$, lie within the plane of the sample. However, orbital displacements from the unit cell origin $\vec{r}_i$, as well as the position $\vec{r}$, can extend out of the plane. For concreteness, we take the sample plane to be perpendicular to the $\hat{\vec{z}}$ direction and decompose any three-dimensional vector $\vec{r}$ as $\vec{r} = \vec{r}_{\parallel} + z \hat{\vec{z}}$. Because the momenta are strictly two-dimensional, we have $\vec{k} \cdot \vec{r} = \vec{k} \cdot \vec{r}_{\parallel}$, $\vec{Q} \cdot \vec{r} = \vec{Q} \cdot \vec{r}_{\parallel}$, {\it etc.}{}

For future convenience, we also define the following partial Fourier transformations of the Wannier orbitals
\begin{align}
	W_{i} \left(\vec{k}, \vec{r} \right) &= \frac{1}{\sqrt{N}} \sum_{\vec{R}} W_{i} \left( \vec{r} - \vec{R} \right) e^{i \vec{k} \cdot \vec{R}} = \frac{1}{\sqrt{\Omega_0}} \sum_{\vec{Q}} W_{i} \left( \vec{k} + \vec{Q}, z \right) e^{i \left( \vec{k} + \vec{Q} \right) \cdot \vec{r}} \label{app:eqn:ft_def_w_to_k_partial}  \\ 
	W_{i} \left( \vec{r} - \vec{R} \right) &= \frac{1}{\sqrt{N}} \sum_{\vec{k}} W_{i} \left(\vec{k}, \vec{r} \right) e^{-i \vec{k} \cdot \vec{R}}. \label{app:eqn:ft_def_w_to_r_partial}
\end{align}
The real-space fermionic field operator can then be expanded into the Wannier orbital basis as
\begin{align}
	\hat{\Psi}^\dagger_{} \left( \vec{r} \right) &= \sum_{\vec{R},i} W^*_i \left( \vec{r}-\vec{R} \right) \hat{c}^\dagger_{\vec{R},i} + \dots, \label{app:eqn:field_in_terms_wan_real} \\
	\hat{\Psi}^\dagger_{} \left( \vec{r} \right) &= \frac{1}{\sqrt{\Omega_0}} \sum_{\vec{k},\vec{Q},i} W^*_i \left( \vec{k} + \vec{Q}, z \right) 
	e^{-i \left( \vec{k}+\vec{Q} \right)\cdot \vec{r} } e^{i \vec{k} \cdot \vec{r}_i }\hat{c}^\dagger_{\vec{k},i} + \dots, \label{app:eqn:field_in_terms_wan_mom}\\
	\hat{\Psi}^\dagger_{} \left( \vec{r} \right) &= \sum_{\vec{k},i} W^{*}_{i} \left(\vec{k}, \vec{r} \right) 
	e^{-i \vec{k} \cdot \vec{r}_i }\hat{c}^\dagger_{\vec{k},i} + \dots. \label{app:eqn:field_in_terms_wan_hyb}
\end{align}
In \crefrange{app:eqn:field_in_terms_wan_real}{app:eqn:field_in_terms_wan_hyb}, the dots ``$\dots$'' denote other states of higher energy, orthogonal to the Wannier orbitals created by $\hat{c}^\dagger_{\vec{R},i}$, which are irrelevant for the present discussion and will be ignored in what follows.

\subsubsection{Obtaining the Wannier orbitals and tight-binding Hamiltonian}\label{app:sec:correlator_general:notation:obtain_wan}

For an isolated set of bands, the Wannier orbital wave functions $W_{i} \left( \vec{r} - \vec{R} \right)$ can generally be obtained using Wannier90~\cite{MAR12, MAR97b, PIZ20, SOU01c} from the \textit{ab initio}{} Bloch wave functions. The Bloch wave function of the $\alpha$-th band at crystalline momentum $\vec{k}$ is denoted by $\psi_{\alpha} \left( \vec{k}, \vec{r} \right)$ and is expanded into plane wave coefficients as
\begin{equation}
	\psi_{\alpha} \left( \vec{k}, \vec{r} \right)= \frac{1}{\sqrt{\Omega}} \sum_{\vec{G}, G_z} C_{\alpha,\vec{k}}^{\vec{G}, G_z} e^{i \left[ \left( \vec{k} + \vec{G} \right) \cdot \vec{r}_{\parallel} + G_z z \right]},
\end{equation}
where $\Omega$ denotes the volume of the system considered in the \textit{ab initio}{} simulations. As mentioned below \cref{app:eqn:ft_def_w_to_k}, momenta in the first BZ are represented by bold lowercase letters, while reciprocal lattice momenta are represented by bold uppercase ones, hence $\vec{k}$ belongs to the first BZ. 

Numerically, for two-dimensional ({\it i.e.}{}, layered) systems, periodic boundary conditions are also assumed along the direction perpendicular to the sample (in this case, the $\hat{\vec{z}}$ direction), necessitating a plane wave expansion in the $\hat{\vec{z}}$ direction as well. However, the system size along $\hat{\vec{z}}$ is chosen to be sufficiently large to ensure no interlayer effects arise. The Bloch wave functions are normalized according to
\begin{equation}
	\label{app:eqn:normalization_bloch_condition}
	\sum_{\vec{G}, G_z} \abs{C_{\alpha,\vec{k}}^{\vec{G}, G_z}}^2 = \Omega.
\end{equation} 

On a strictly technical note, in \textit{ab initio}{} simulations, one typically employs a supercell made up of $N_{k_1} \times N_{k_2} \times N_{k_3}$ unit cell. Each unit cell is further discretized on an $N_{G_1} \times N_{G_2} \times N_{G_3}$ grid. Periodic boundary conditions are applied in all three directions. Denoting the direct lattice vectors by $\vec{a}_i$ and the reciprocal lattice vectors by $\vec{b}_i$, the allowed values for the \emph{three-dimensional} crystalline and reciprocal momenta are
\begin{align}
	\vec{k} &= \sum_{i=1}^{3} \frac{n_i}{N_{k_i}} \vec{b}_i, \qq{for} 0 \leq n_i <N_{k_i}, \\
	\vec{G} &= \sum_{i=1}^{3} n_i \vec{b}_i, \qq{for} 0 \leq n_i < N_{G_i}.
\end{align}  
For layered materials, we take $\abs{\vec{a}_3} \gg \abs{\vec{a}_{1,2}}$, $N_{k_3} = 1$, and $N_{G_3} \gg N_{G_{1,2}}$, which makes $\vec{k}$ strictly two-dimensional, as assumed throughout this work.

Using the Bloch wave functions, Wannier90 determines a gauge-smoothing transformation $U_{i \alpha} \left(\vec{k} \right)$, which provides the Wannier orbital wave functions through
\begin{equation}
	\label{app:eqn:wannier_wavf_plane_wave}
	W_{i} \left(\vec{k}, \vec{r} \right) = \sum_{\alpha} U_{i \alpha} \left(\vec{k} \right) \psi_{\alpha} \left( \vec{k}, \vec{r} \right)  = \frac{1}{\sqrt{\Omega}} \sum_{\vec{G}, G_z} \sum_{\alpha} U_{i\alpha} \left( \vec{k} \right) C_{\alpha,\vec{k}}^{\vec{G}, G_z} e^{i \left[ \left( \vec{k} + \vec{G} \right) \cdot \vec{r}_{\parallel} + G_z z \right]}.
\end{equation}
If the number of bands equals the number of Wannier wave functions, $U_{i \alpha} \left(\vec{k} \right)$ is a unitary matrix designed to yield maximally localized Wannier functions (MLWFs). When the number of bands exceeds the number of Wannier wave functions, $U_{i \alpha} \left(\vec{k} \right)$ includes both a projector for disentanglement and a unitary transformation for obtaining the MLWFs. In general, $U_{i \alpha} \left(\vec{k} \right)$ is a full-rank matrix with unit singular values. By comparing \cref{app:eqn:ft_def_w_to_k_partial,app:eqn:wannier_wavf_plane_wave}, we can directly determine the Fourier coefficients of the Wannier wave functions
\begin{equation}
	W_{i} \left( \vec{k} + \vec{G}, z \right) = \sqrt{\frac{\Omega_0}{\Omega}} \sum_{G_z} \sum_{\alpha} U_{i\alpha} \left( \vec{k} \right) C_{\alpha,\vec{k}}^{\vec{G}, G_z} e^{i G_z z}.
\end{equation}

Finally, the Kohn-Sham tight-binding quadratic Hamiltonian of the system is defined as 
\begin{equation}
	\hat{H} = \sum_{\Delta\vec{R},i,j} H_{ij} \left( \Delta \vec{R} \right) \hat{c}^\dagger_{\vec{R},i} \hat{c}_{\vec{R} + \Delta \vec{R},j} = \sum_{\vec{k},i,j} H_{ij} \left( \vec{k} \right) \hat{c}^\dagger_{\vec{k},i} \hat{c}_{\vec{k},j},
\end{equation}
where $H_{ij} \left( \Delta \vec{R} \right)$ and $H_{ij} \left( \vec{k} \right)$ are, respectively, the real-space and momentum-space Hamiltonian matrices. These matrices are related through the following Fourier transformation
\begin{equation}
	H_{ij} \left( \vec{k} \right) = \sum_{\Delta \vec{R}} H_{ij} \left( \Delta \vec{R} \right) e^{i \vec{k} \left( \cdot \Delta \vec{R} + \vec{r}_j - \vec{r}_i \right)},
\end{equation}
with the real-space Hamiltonian matrix being also outputted by Wannier90~\cite{MAR12, MAR97b, PIZ20, SOU01c}. We denote the eigenvector and eigenvalue of the $\alpha$-th band of the tight binding Hamiltonian by $\epsilon_{\alpha} \left( \vec{k} \right)$ and $u_{i \alpha} \left( \vec{k} \right)$, respectively, such that
\begin{equation}
	\sum_{j} H_{ij} \left( \vec{k} \right) u_{j \alpha} \left( \vec{k} \right) = \epsilon_{\alpha} \left( \vec{k} \right) u_{i \alpha} \left( \vec{k} \right).
\end{equation}

\subsection{The CDD of an isolated band}\label{app:sec:correlator_general:band_ldos}

Consider a single (isolated) band of the system, chosen without loss of generality to be the $\alpha$-th band. For convenience, we denote its wave function by
\begin{equation}
	u_{i} \left( \vec{k} \right) = u_{i \alpha} \left( \vec{k} \right).
\end{equation}
Here, we suppress the band index for simplicity, as our discussion exclusively concerns this band. We aim to compute the CDD of a state where the $\alpha$-th band is fully filled, defined analogously to \cref{app:eqn:initial_definition_fully_filled_band} as
\begin{equation}
	\ket{\Phi} = \prod_{\vec{k}} \hat{\gamma}^\dagger_{\vec{k}} \ket{0},
\end{equation}
where the band operators are defined by 
\begin{equation}
	\hat{\gamma}^\dagger_{\vec{k}} = \sum_{i} u_{i} \left( \vec{k} \right) \hat{c}^\dagger_{\vec{k},i}.
\end{equation}

\subsubsection{Decomposition into the spatial factor and the correlation functions}\label{app:sec:correlator_general:band_ldos:decomposition}

Using \cref{app:eqn:field_in_terms_wan_real}, the CDD of the system can be expressed as
\begin{equation}
	\label{app:eqn:ldos_to_wannier}
	A \left( \vec{r} \right) =  \mel**{\Phi}{ \hat{\Psi}^\dagger_{} \left( \vec{r} \right) \hat{\Psi}_{} \left( \vec{r} \right) }{\Phi} = \sum_{\substack{\vec{R}_1, \vec{R}_2 \\ i,j}} W^{*}_{i} \left( \vec{r} - \vec{R}_1 \right) W_{j} \left( \vec{r} - \vec{R}_2 \right)  \mel**{\Phi}{ \hat{c}^\dagger_{\vec{R}_1,i} \hat{c}_{\vec{R}_2,j} }{\Phi}.
\end{equation}
\Cref{app:eqn:ldos_to_wannier} forms the basis of our method: it relates an operator diagonal in space, $\hat{\Psi}^\dagger_{} \left( \vec{r} \right) \hat{\Psi}_{} \left( \vec{r} \right)$, to the potentially \emph{off-diagonal} correlation function $\mel**{\Phi}{ \hat{c}^\dagger_{\vec{R}_1,i} \hat{c}_{\vec{R}_2,j} }{\Phi}$. The expression can then be employed to derive the matrix elements $\mel**{\Phi}{ \hat{c}^\dagger_{\vec{R}_1,i} \hat{c}_{\vec{R}_2,j} }{\Phi}$ in terms of the experimentally measured CDD. This method relies on the Wannier orbitals not being highly localized. Indeed, for Wannier orbitals $W^{*}_{i} \left( \vec{r} - \vec{R} \right)$ tightly (delta-function) localized around the Wannier centers $\vec{R} + \vec{r}_i$
\begin{equation}
	W^{*}_{i} \left( \vec{r} - \vec{R}_1 \right) W_{j} \left( \vec{r} - \vec{R}_2 \right) \approx W^{*}_{i} \left( \vec{r} - \vec{R}_1 \right) W_{j} \left( \vec{r} - \vec{R}_1 \right) \delta_{\vec{R}_1+\vec{r}_i,\vec{R}_2+\vec{r}_j},
\end{equation}
making the CDD approximately independent of the off-diagonal matrix elements $\mel**{\Phi}{ \hat{c}^\dagger_{\vec{R}_1,i} \hat{c}_{\vec{R}_2,j} }{\Phi}$, for which $\vec{R}_1 + \vec{r}_i \neq \vec{R}_2 + \vec{r}_j$.

By substituting the fermionic field operators from \cref{app:eqn:field_in_terms_wan_hyb} into \cref{app:eqn:ldos_to_wannier}, we obtain
\begin{align}
	A \left( \vec{r} \right) &= \sum_{\substack{\vec{k}_1, \vec{k}_2 \\ i,j}} W^{*}_{i} \left( \vec{k}_1, \vec{r} \right) W_{j} \left( \vec{k}_2, \vec{r} \right) e^{- i \vec{k}_1 \cdot \vec{r}_i} e^{i \vec{k}_2 \cdot \vec{r}_j}  \mel**{\Phi}{ \hat{c}^\dagger_{\vec{k}_1,i} \hat{c}_{\vec{k}_2,j} }{\Phi} \nonumber \\
	&= \sum_{\vec{k}, i,j} W^{*}_{i} \left( \vec{k}, \vec{r} \right) W_{j} \left( \vec{k}, \vec{r} \right) e^{-i \vec{k} \cdot \vec{r}_i} e^{i \vec{k} \cdot \vec{r}_j}  \mel**{\Phi}{ \hat{c}^\dagger_{\vec{k},i} \hat{c}_{\vec{k},j} }{\Phi} \nonumber \\
	&= \sum_{\vec{k}, i,j} B_{ij} \left(\vec{r}, \vec{k} \right) \rho_{ij} \left( \vec{k} \right), \label{app:eqn:ldos_with_sp_factor_mom_space}
\end{align}
where the second equality follows from $\ket{\Phi}$ being translation-invariant. In \cref{app:eqn:ldos_with_sp_factor_mom_space}, we have also defined
\begin{align}
	B_{ij} \left(\vec{r}, \vec{k} \right) &= W^{*}_{i} \left( \vec{k}, \vec{r} \right) W_{j} \left( \vec{k}, \vec{r} \right), \label{app:eqn_def_spatial_factor_interm} \\ 
	\rho_{ij} \left( \vec{k} \right) &= e^{-i \vec{k} \cdot \vec{r}_i} e^{i \vec{k} \cdot \vec{r}_j}  \mel**{\Phi}{ \hat{c}^\dagger_{\vec{k},i} \hat{c}_{\vec{k},j} }{\Phi}. \label{app:eqn:def_order_parameter_matrix_k_space}
\end{align}
The two quantities defined above are termed the spatial factor and the correlation function of the band. The correlation function $\rho_{ij} \left( \vec{k} \right)$ encodes the correlators between the Wannier orbitals on which the band is supported, while $B_{ij} \left(\vec{r}, \vec{k} \right)$ relates these correlators to the band's CDD. In real space, the correlators characterizing the band are
\begin{align}
	\rho_{ij} \left( \Delta \vec{R} \right) =& \mel**{\Phi}{ \hat{c}^\dagger_{\vec{R},i} \hat{c}_{\vec{R} + \Delta \vec{R},j} }{\Phi}, \qq{for any lattice vector} \vec{R} \nonumber \\
	=& \frac{1}{N} \sum_{\vec{R}} \mel**{\Phi}{ \hat{c}^\dagger_{\vec{R},i} \hat{c}_{\vec{R} + \Delta \vec{R},j} }{\Phi} \nonumber \\
	=& \frac{1}{N} \sum_{\vec{R}, \vec{k}_1, \vec{k}_2} e^{- i\vec{k}_1 \cdot \left( \vec{R} + \vec{r}_i \right)} e^{i \vec{k}_2 \cdot \left( \vec{R} + \Delta \vec{R} + \vec{r}_j \right)}\mel**{\Phi}{ \hat{c}^\dagger_{\vec{k}_1,i} \hat{c}_{\vec{k}_2,j} }{\Phi} \nonumber \\
	=& \sum_{\vec{k}} e^{- i\vec{k} \cdot \vec{r}_i} e^{i \vec{k} \cdot \left( \Delta \vec{R} + \vec{r}_j \right)} \mel**{\Phi}{ \hat{c}^\dagger_{\vec{k},i} \hat{c}_{\vec{k},j} }{\Phi} \nonumber \\
	=& \sum_{\vec{k}} \rho_{ij} \left( \vec{k} \right) e^{i \vec{k} \cdot \Delta \vec{R}}, \label{app:eqn:def_order_parameter_matrix_r_space}\\
	\rho_{ij} \left( \vec{k} \right) =& \frac{1}{N} \sum_{\vec{R}} \rho_{ij} \left( \Delta \vec{R} \right) e^{-i \vec{k} \cdot \Delta \vec{R}}
\end{align}

In practice, it is more convenient to work in momentum space. Since the state with one filled band is translation-invariant, $\hat{T}_{\vec{R}} \ket{\Phi} = \ket{\Phi}$, we must have 
\begin{equation}
	A \left( \vec{r} + \vec{R} \right) =  \mel**{\Phi}{ \hat{T}_{\vec{R}} \hat{\Psi}^\dagger_{} \left( \vec{r} \right) \hat{\Psi}_{} \left( \vec{r} \right) \hat{T}^{-1}_{\vec{R}} }{\Phi} =  \mel**{\Phi}{ \hat{\Psi}^\dagger_{} \left( \vec{r} \right) \hat{\Psi}_{} \left( \vec{r} \right) }{\Phi} = A \left( \vec{r} \right).
\end{equation}
As a result, the CDD admits the following Fourier decomposition
\begin{align}
	A \left( \vec{r} \right) &= \sum_{\vec{Q}} A \left( \vec{Q}, z\right) e^{i \vec{Q} \cdot \vec{r}}, \label{app:eqn:ft_sp_func_to_real} \\
	A \left( \vec{Q}, z \right) &= \frac{1}{N \Omega_0} \int \dd[2]{r_{\parallel}} A \left( \vec{r} \right) e^{-i \vec{Q} \cdot \vec{r}}.  \label{app:eqn:ft_sp_func_to_mom}
\end{align} 
Writing the Fourier transform of the CDD $A \left( \vec{Q}, z \right)$ in terms of the real-space correlation functions $\rho_{ij} \left( \Delta \vec{R} \right)$, we obtain
\begin{equation}
	\label{app:eqn:sp_func_of_correlators}
	A \left( \vec{Q}, z \right) = \sum_{\vec{R}, i, j} B_{ij} \left( \vec{Q}, z, \Delta \vec{R} \right) \rho_{ij} \left( \Delta \vec{R} \right),
\end{equation}
where the Fourier-transformed spatial factor is defined as
\begin{align}
	B_{ij} \left( \vec{Q}, z, \Delta \vec{R} \right) &= \frac{1}{N} \sum_{\vec{k}} \frac{1}{N \Omega_0} \int \dd[2]{r_{\parallel}}B_{ij} \left(\vec{r}, \vec{k} \right) e^{- i \vec{k} \cdot \Delta \vec{R}} e^{- i \vec{Q} \cdot \vec{r}} \nonumber \\
	&= \frac{1}{N} \sum_{\vec{k}} \frac{1}{N \Omega_0} \int \dd[2]{r_{\parallel}} W^{*}_{i} \left( \vec{k}, \vec{r} \right) W_{j} \left( \vec{k}, \vec{r} \right)  e^{- i \vec{k} \cdot \Delta \vec{R}} e^{- i \vec{Q} \cdot \vec{r}} \nonumber \\
	&= \frac{1}{N} \sum_{\vec{k}} \frac{1}{N \Omega_0} \int \dd[2]{r_{\parallel}} \frac{1}{\Omega_0} \sum_{\vec{Q}_1, \vec{Q}_2} W^{*}_{i} \left( \vec{k} + \vec{Q}_1, z \right) W_{j} \left( \vec{k} + \vec{Q}_2, z \right) e^{i \left( \vec{Q}_2 - \vec{Q}_1 \right) \cdot \vec{r}}  e^{- i \vec{k} \cdot \Delta \vec{R}} e^{- i \vec{Q} \cdot \vec{r}} \nonumber \\
	&= \frac{1}{N \Omega_0} \sum_{\vec{k},\vec{Q}'} W^{*}_{i} \left( \vec{k} + \vec{Q}', z \right) W_{j} \left( \vec{k} + \vec{Q}' + \vec{Q}, z \right)  e^{- i \vec{k} \cdot \Delta \vec{R}}. \label{app:eqn:def_spatial_factor_final}
\end{align}

\Cref{app:eqn:sp_func_of_correlators} will be used to extract the real-space correlation functions $\rho_{ij} \left( \Delta \vec{R} \right)$ from the experimentally measured CDD $A \left( \vec{Q} \right)$, as will be explained in \cref{app:sec:correlator_general:fitting}. Regarding the spatial factor, \cref{app:eqn:def_spatial_factor_final} illustrates multiple ways to compute it from the Wannier wave functions. In practice, the most efficient method is as follows
\begin{itemize}
	\item First, determine the Wannier wave functions $ W_{i} \left( \vec{k}, \vec{r} \right)$ from the \textit{ab initio}{} Bloch functions using Wannier90, as shown in \cref{app:eqn:wannier_wavf_plane_wave}. Alternatively, start from the Fourier-transformed Wannier wave functions $ W_{i} \left( \vec{k} + \vec{Q}, z \right)$ and perform a Fourier transformation over $\vec{Q}$, as in the last equality of \cref{app:eqn:ft_def_w_to_k_partial}.
	
	\item Next, compute the spatial factor $B_{ij} \left(\vec{r}, \vec{k} \right)$ using \cref{app:eqn_def_spatial_factor_interm}.
	
	\item Finally, perform a Fourier transformation of $B_{ij} \left(\vec{r}, \vec{k} \right)$ over both $\vec{r}$ and $\vec{k}$, as detailed in the first line of \cref{app:eqn:def_spatial_factor_final}.
\end{itemize}
By following this procedure, the numerically expensive convolution in \cref{app:eqn:def_spatial_factor_final} can be avoided\footnote{Numerically the direct convolution needs $\mathcal{O} \left( N_{G_1}^2 N_{G_2}^2 \right)$ operations, while the more efficient method using Fourier transformations only requires $\mathcal{O} \left( N_{G_1} N_{G_2} \log \left( N_{G_1} N_{G_2} \right) \right)$ operations.}.

\subsubsection{Properties of real-space correlation functions}\label{app:sec:correlator_general:band_ldos:properties}

We now examine the real-space correlation functions in more detail. In momentum space, the correlation function takes the form of a projector
\begin{equation}
	\rho_{ij} \left( \vec{k} \right) = e^{-i \vec{k} \cdot \vec{r}_i} e^{i \vec{k} \cdot \vec{r}_j}  \mel**{\Phi}{ \hat{c}^\dagger_{\vec{k},i} \hat{c}_{\vec{k},j} }{\Phi} =  e^{-i \vec{k} \cdot \vec{r}_i} u^*_i \left( \vec{k} \right) e^{i \vec{k} \cdot \vec{r}_j} u_j \left( \vec{k} \right).
\end{equation}
As a result, $\rho_{ij} \left( \vec{k} \right)$ has trace one and is idempotent
\begin{align}
	\sum_{i} \rho_{ii} \left( \vec{k} \right) &= \sum_{i} e^{-i \vec{k} \cdot \vec{r}_i} u^*_i \left( \vec{k} \right) e^{i \vec{k} \cdot \vec{r}_i} u_i \left( \vec{k} \right) = 1, \label{app:eqn:tr_one_mom_space} \\
	\sum_{j} \rho_{ij} \left( \vec{k} \right) \rho_{jk} \left( \vec{k} \right) &= \sum_{j} e^{-i \vec{k} \cdot \vec{r}_i} u^*_i \left( \vec{k} \right) e^{i \vec{k} \cdot \vec{r}_j} u_j \left( \vec{k} \right) e^{-i \vec{k} \cdot \vec{r}_j} u^*_j \left( \vec{k} \right) e^{i \vec{k} \cdot \vec{r}_k} u_k \left( \vec{k} \right) \nonumber \\
	&= e^{-i \vec{k} \cdot \vec{r}_i} u^*_i \left( \vec{k} \right) e^{i \vec{k} \cdot \vec{r}_k} u_k \left( \vec{k} \right) \nonumber \\
	&= \rho_{ik} \left( \vec{k} \right). \label{app:eqn:idem_pot_mom_space}
\end{align}
By Fourier transforming \cref{app:eqn:tr_one_mom_space,app:eqn:idem_pot_mom_space} to real space, we obtain two properties of the real-space correlation function
\begin{align}
	1 =& \frac{1}{N} \sum_{\vec{k},i} \rho_{ii} \left( \vec{k} \right) = \sum_{i} \eval{ \rho_{ii} \left( \Delta \vec{R} \right)}_{\Delta \vec{R} = \vec{0}} \label{app:eqn:normalization_rho_linear} \\
	1 =& \frac{1}{N} \sum_{\vec{k},i} \rho_{ii} \left( \vec{k} \right) = \frac{1}{N} \sum_{\vec{k},i,j} \rho_{ij} \left( \vec{k} \right) \rho_{ji} \left( \vec{k} \right) = \frac{1}{N} \sum_{\vec{k},i,j} \rho_{ij} \left( \vec{k} \right) \rho^{*}_{ij} \left( \vec{k} \right) \nonumber \\
	=&  \frac{1}{N^2} \sum_{\vec{k}_1, \vec{k}_2} \sum_{\Delta \vec{R}, i, j} \rho_{ij} \left( \vec{k}_1 \right) e^{i \vec{k}_1 \cdot \Delta \vec{R}} \rho^{*}_{ij} \left( \vec{k}_2 \right) e^{-i \vec{k}_2 \cdot \Delta \vec{R}}  \nonumber \\
	=&  \sum_{\Delta \vec{R}, i, j} \abs{\rho_{ij} \left( \Delta \vec{R} \right)}^2. \label{app:eqn:normalization_rho_square}
\end{align}
We term these the normalization conditions
\begin{equation}
	\sum_{\Delta \vec{R}, i, j} \abs{\rho_{ij} \left( \Delta \vec{R} \right)}^2 = \sum_{i} \eval{ \rho_{ii} \left( \Delta \vec{R} \right)}_{\Delta \vec{R} = \vec{0}} = 1.
\end{equation}

Finally, for computing the CDD of a band (as opposed to extracting the correlation function from the CDD) using \textit{ab initio}{} methods, the following expression is the most computationally efficient, as it avoids direct convolutions. Starting from the second row of \cref{app:eqn:ldos_with_sp_factor_mom_space}, we obtain
\begin{align}
	A \left( \vec{r} \right) &= \sum_{\vec{k},i,j} W^{*}_{i} \left( \vec{k}, \vec{r} \right) W_{j} \left( \vec{k}, \vec{r} \right) e^{- i \vec{k} \cdot \vec{r}_i} e^{i \vec{k} \cdot \vec{r}_j} u^*_i \left( \vec{k} \right) u_j \left( \vec{k} \right) \nonumber \\
	&= \sum_{\vec{k},i}  \abs{ W_{i} \left( \vec{k}, \vec{r} \right) u_i \left( \vec{k} \right) e^{i \vec{k} \cdot \vec{r}_i}}^2.
\end{align}

\subsection{Symmetry properties}\label{app:sec:correlator_general:sym_prop}

In this section, we consider the constraints imposed by the crystalline symmetries of the system on the CDD, correlation function, and spatial factor. We take $g = \left\lbrace R | \vec{p} \right\rbrace$ to be a symmorphic ($\vec{p} = \vec{0}$) or non-symmorphic ($\vec{p} \neq \vec{0}$) crystalline symmetry of the system, consisting of a rotation $R$ followed by a translation by a non-lattice vector $\vec{p}$. As we focus on two-dimensional samples probed by STM, $R$ ($\vec{p}$) denotes only in-plane rotations and translations, as well as out-of-plane reflections. In full generality, $g$ can be either unitary or antiunitary.

The action of $g$ on any real-space vector is defined as
\begin{equation}
	g \vec{r} = R \vec{r} + \vec{p},
\end{equation}
while the action on a momentum $\vec{k}$ depends on whether $g$ is unitary or not
\begin{equation}
	g \vec{k} = \begin{cases}
		R \vec{k} & \quad \text{if $g$ is unitary} \\
		-R \vec{k} & \quad \text{if $g$ is antiunitary} \\
	\end{cases}.
\end{equation}

The action of $g$ on the electron operators is given by
\begin{equation}
	\label{app:eqn:tb_general_notation_symmetries}
	g \hat{c}^\dagger_{\vec{R},i} g^{-1} = \sum_{j} D_{ij} (g) \hat{c}^\dagger_{g \left( \vec{R} + \vec{r}_i \right) - \vec{r}_j,j},
\end{equation}
where $D_{ij} (g)$ represents the unitary matrix corresponding to the symmetry $g$. Importantly, $D_{ij} (g)$ is nonzero only if $g \left( \vec{R} + \vec{r}_i \right) - \vec{r}_j$ is a lattice vector.

To obtain the constraints imposed by $g$ on the Wannier wave functions, we use the fact that 
\begin{equation}
	\label{app:eqn:trafo_of_field}
	g \hat{\Psi}_{} \left( \vec{r} \right) g^{-1} = \hat{\Psi}_{} \left( g \vec{r} \right).
\end{equation}
Using the expansion in \cref{app:eqn:field_in_terms_wan_real}, \cref{app:eqn:trafo_of_field} leads to
\begin{alignat}{4}
	&& \sum_{\vec{R},i,j} W^{(*)}_i \left( \vec{r}-\vec{R} \right) D^{*}_{ij} (g) \hat{c}_{g \left( \vec{R} + \vec{r}_i \right) - \vec{r}_j,j} &= \sum_{\vec{R},i} W_i \left( g\vec{r}-\vec{R} \right) \hat{c}_{\vec{R},i} && \iff \nonumber  \\
	\iff &&  \sum_{\vec{R},i,j} W^{(*)}_i \left( \vec{r}- g^{-1} \left( \vec{R} + \vec{r}_j \right) + \vec{r}_i  \right) D^{*}_{ij} (g) \hat{c}_{\vec{R},j} &= \sum_{\vec{R},j} W_j \left( g\vec{r}-\vec{R} \right) \hat{c}_{\vec{R},j}, && \label{app:eqn:action_symmetry_wan_func_interm}
\end{alignat}
with ${}^{(*)}$ implying complex conjugation whenever $g$ is antiunitary. Since the operators $\hat{c}_{\vec{R},j}$ are linearly independent, we find that (for any $\vec{r}$ and $\vec{R}$)
\begin{alignat}{4}
	&& W_j \left( R\vec{r} + \vec{p} -\vec{R} \right) &= \sum_{i} W^{(*)}_i \left( \vec{r}- R^{-1} \left( \vec{R} + \vec{r}_j - \vec{p} \right) + \vec{r}_i  \right) D^{*}_{ij} (g) &&\iff \nonumber \\
	\iff && W_j \left( \vec{r} \right) & = \sum_{i} W^{(*)}_i \left( R^{-1} \left( \vec{r} - \vec{p} \right)- R^{-1} \left( \vec{r}_j - \vec{p} \right) + \vec{r}_i  \right) D^{*}_{ij} (g) &&\iff \nonumber \\
	\iff && W_j \left( \vec{r} + \vec{r}_j \right) & = \sum_{i} W^{(*)}_i \left( R^{-1} \vec{r}  + \vec{r}_i  \right) D^{*}_{ij} (g). && \label{app:eqn:symmetry_real_wannier}
\end{alignat}
Fourier transforming \cref{app:eqn:symmetry_real_wannier} according to \cref{app:eqn:ft_def_w_to_r}, we directly obtain
\begin{equation}
	W_j \left( \vec{k} + \vec{Q}, z \right) = \sum_{i} W^{(*)}_i \left( g^{-1} \left( \vec{k} + \vec{Q} \right), z \right) e^{i \left( \vec{k} + \vec{Q} \right) \cdot \left( R \vec{r}_i - \vec{r}_j \right)} D^{*}_{ij} (g). \label{app:eqn:symmetry_momentum_wannier}
\end{equation}
From \cref{app:eqn:def_spatial_factor_final,app:eqn:symmetry_momentum_wannier}, the constraints imposed by the symmetry $g$ on the spatial factor are
\begin{align}
	B_{ij} \left( \vec{Q}, z, \Delta \vec{R} \right) =& \frac{1}{N \Omega_0} \sum_{\substack{\vec{k},\vec{Q}' \\ i' j'}} \left[ W^{*}_{i'} \left(g^{-1} \left( \vec{k} + \vec{Q}' \right), z \right) W_{j'} \left( g^{-1} \left( \vec{k} + \vec{Q}' + \vec{Q} \right), z \right) \right]^{(*)} \nonumber \\
	& \times e^{-i \left( \vec{k} + \vec{Q}' \right) \cdot \left( R \vec{r}_{i'} - \vec{r}_i \right)} e^{i \left( \vec{k} + \vec{Q}' + \vec{Q} \right) \cdot \left( R \vec{r}_{j'} - \vec{r}_j \right)}  e^{- i \vec{k} \cdot \Delta \vec{R}} D_{i'i} (g) D^{*}_{j'j} (g) \nonumber \\
	=& \frac{1}{N \Omega_0} \sum_{\substack{\vec{k},\vec{Q}' \\ i' j'}} \left[ W^{*}_{i'} \left(g^{-1} \left( \vec{k} + \vec{Q}' \right), z \right) W_{j'} \left( g^{-1} \left( \vec{k} + \vec{Q}' + \vec{Q} \right), z \right) \right]^{(*)} \nonumber \\
	& \times e^{-i \left( \vec{k} + \vec{Q}' \right) \cdot \left( \Delta \vec{R} +  R \vec{r}_{i'} - \vec{r}_i - R \vec{r}_{j'} + \vec{r}_j \right)} e^{i \vec{Q} \cdot \left( R \vec{r}_{j'} - \vec{r}_j \right)} D_{i'i} (g) D^{*}_{j'j} (g)\nonumber \\
	=& \frac{1}{N \Omega_0} \sum_{\substack{\vec{k},\vec{Q}' \\ i' j'}} \left\lbrace W^{*}_{i'} \left(\vec{k} + \vec{Q}', z \right) W_{j'} \left( \vec{k} + \vec{Q}' + g^{-1} \vec{Q}, z \right) e^{-i \left[ R \left( \vec{k} + \vec{Q}' \right) \right] \cdot \left( \Delta \vec{R} +  R \vec{r}_{i'} - \vec{r}_i - R \vec{r}_{j'} + \vec{r}_j \right)}\right\rbrace^{(*)} \nonumber \\
	& \times e^{i \vec{Q} \cdot \left( g \vec{r}_{j'} - \vec{p} - \vec{r}_j \right)} D_{i'i} (g) D^{*}_{j'j} (g) \nonumber \\
	=& \frac{1}{N \Omega_0} \sum_{\substack{\vec{k},\vec{Q}' \\ i' j'}} \left\lbrace W^{*}_{i'} \left(\vec{k} + \vec{Q}', z \right) W_{j'} \left( \vec{k} + \vec{Q}' + g^{-1} \vec{Q}, z \right) e^{-i  \left( \vec{k} + \vec{Q}' \right) \cdot \left[ g^{-1} \left( \Delta \vec{R}  + \vec{r}_j\right) - g^{-1} \vec{r}_i - \vec{r}_{j'} + \vec{r}_{i'}\right]}\right\rbrace^{(*)} \nonumber \\
	& \times e^{- i \vec{Q} \cdot \vec{p}} D_{i'i} (g) D^{*}_{j'j} (g) \nonumber \\
	=& \sum_{i',j'} B^{(*)}_{i'j'} \left( g^{-1} \vec{Q}, z, g^{-1} \left( \Delta \vec{R}  + \vec{r}_j\right) - g^{-1} \vec{r}_i - \vec{r}_{j'} + \vec{r}_{i'} \right) e^{- i \vec{Q} \cdot \vec{p}} D_{i'i} (g) D^{*}_{j'j} (g). \label{app:eqn:trafo_spatial_factor_interm}
\end{align}
For future use in \cref{app:eqn:sym_constraint_FFT_CDD_confirmation}, we take $g = \left\lbrace R | \vec{p} \right\rbrace \to g^{-1} = \left\lbrace R^{-1} | - R^{-1}\vec{p} \right\rbrace$ in \cref{app:eqn:trafo_spatial_factor_interm} to obtain
\begin{equation}
	\label{app:eqn:trafo_spatial_factor}
	B_{ij} \left( \vec{Q}, z, \Delta \vec{R} \right) =\sum_{i',j'} B^{(*)}_{i'j'} \left( g \vec{Q}, z, R \left( \Delta \vec{R}  + \vec{r}_j -  \vec{r}_i \right) - \vec{r}_{j'} + \vec{r}_{i'} \right) e^{i \left( R \vec{Q} \right) \cdot \vec{p}} D^{*}_{ii'} (g) D_{jj'} (g). 
\end{equation}

We now turn to the real-space correlation function. Using the fact that the state $\ket{\Phi}$ is symmetric under the symmetry $g$ ({\it i.e.}{}, $g \ket{\Phi} = \ket{\Phi}$), we obtain
\begin{align}
	\rho_{ij} \left( \Delta \vec{R} \right) =& \mel**{\Phi}{g \hat{c}^\dagger_{\vec{0},i} \hat{c}_{\Delta \vec{R},j}  g^{-1}}{\Phi}^{(*)} \nonumber \\
	=& \sum_{i',j'} D_{ii'} (g) D^{*}_{jj'} (g) \mel**{\Phi}{ \hat{c}^\dagger_{g \vec{r}_i  - \vec{r}_{i'} ,i'} \hat{c}_{g \left( \Delta \vec{R} + \vec{r}_j \right) - \vec{r}_{j'},j'} }{\Phi}^{(*)} \nonumber \\
	=& \sum_{i',j'} D_{ii'} (g) D^{*}_{jj'} (g)	\rho^{(*)}_{i'j'} \left( g \left( \Delta \vec{R} + \vec{r}_j \right) - g \vec{r}_i  - \vec{r}_{j'} + \vec{r}_{i'} \right) \nonumber \\
	=& \sum_{i',j'} D_{ii'} (g) D^{*}_{jj'} (g)	\rho^{(*)}_{i'j'} \left( R \left( \Delta \vec{R} + \vec{r}_j - \vec{r}_i \right)  - \vec{r}_{j'} + \vec{r}_{i'} \right). \label{app:eqn:trafo_ord_param}
\end{align}

Finally, we consider the CDD of the band. Again using the fact that the state $\ket{\Phi}$ is symmetric under the symmetry $g$, as well as the transformation property of the fermionic field from \cref{app:eqn:trafo_of_field}, we trivially find
\begin{equation}
	A \left( g \vec{r} \right) = \mel**{\Phi}{g \hat{\Psi}^\dagger_{} \left( \vec{r} \right)  \hat{\Psi}_{} \left( \vec{r} \right) g^{-1} }{\Phi} = A \left( \vec{r} \right).
\end{equation} 
This implies, through \cref{app:eqn:ft_sp_func_to_mom}, that the Fourier transformation of $A \left( \vec{r} \right)$ obeys
\begin{alignat}{4}
	&& A \left( \vec{Q}, z \right) &= \frac{1}{N \Omega_0} \int \dd[2]{r_{\parallel}} A \left( g \vec{r} \right) e^{- i \vec{Q} \cdot \vec{r}} && \iff \nonumber \\
	\iff && A \left( \vec{Q}, z \right) &= \frac{1}{N \Omega_0} \int \dd[2]{r_{\parallel}} A \left( \vec{r} \right) e^{- i \vec{Q} \cdot \left[ R^{-1} \left( \vec{r} - \vec{p} \right) \right]} && \iff \nonumber \\
	\iff && A \left( \vec{Q}, z \right) &= \frac{1}{N \Omega_0} \int \dd[2]{r_{\parallel}} A \left( \vec{r} \right) e^{- i \left( R \vec{Q} \right) \cdot \vec{r}} e^{R \vec{Q} \cdot \vec{p}} && \iff \nonumber \\
	\iff && A \left( \vec{Q}, z \right) &= A \left( R \vec{Q}, z \right) e^{i \left( R \vec{Q} \right) \cdot \vec{p}}, && \qq{for any transformation} g = \left\lbrace R | \vec{p} \right\rbrace. 
	\label{app:eqn:sym_constraint_FFT_CDD}
\end{alignat}
Importantly, note the absence of complex conjugation when $g$ is antiunitary, as the real-space CDD is a real (and positive) quantity. In particular, the presence of time-reversal symmetry (which is antiunitary and for which $R$ is the identity and $\vec{p} = \vec{0}$) imposes no \emph{direct} constraints on the CDD. Consequently, the breaking of time-reversal symmetry cannot be \emph{directly} inferred by analyzing $A \left( \vec{r} \right)$. Although not a symmetry property, we also note that because $A \left( \vec{r} \right)$ is real and positive, we must have
\begin{equation}
	\label{app:eqn:hermiticty_prop_FFT_CDD}
	A \left( \vec{Q}, z \right) = A^{*} \left( -\vec{Q}, z \right).
\end{equation}
In addition, it is worth mentioning that \cref{app:eqn:sym_constraint_FFT_CDD} can also be obtained from \cref{app:eqn:trafo_spatial_factor,app:eqn:trafo_ord_param,app:eqn:hermiticty_prop_FFT_CDD}
\begin{align}
	A \left( \vec{Q}, z \right) &= \sum_{\vec{R}, i, j} B^{(*)}_{ij}  \left( g \vec{Q}, z, R \left( \Delta \vec{R}  + \vec{r}_j -  \vec{r}_i \right) - \vec{r}_{j'} + \vec{r}_{i'} \right) e^{i \left( R \vec{Q} \right) \cdot \vec{p}} \rho^{(*)}_{i'j'} \left( R \left( \Delta \vec{R} + \vec{r}_j - \vec{r}_i \right)  - \vec{r}_{j'} + \vec{r}_{i'} \right) \nonumber \\
	&= A^{*} \left( g \vec{Q}, z \right) e^{i \left( R \vec{Q} \right) \cdot \vec{p}} \nonumber \\
	&= A \left( R \vec{Q}, z \right) e^{i \left( R \vec{Q} \right) \cdot \vec{p}}. \label{app:eqn:sym_constraint_FFT_CDD_confirmation}
\end{align}

\subsection{Fitting the correlation function}\label{app:sec:correlator_general:fitting}
We now explain how the correlation function can be extracted from experimental measurements of $A \left( \vec{r} \right)$. The general method can be outlined as follows:
\begin{enumerate}
	\item Begin by writing down the most general form of the real-space correlation function that obeys the symmetries of the system. These symmetries constrain $\rho_{ij} \left( \Delta \vec{R} \right)$ according to \cref{app:eqn:trafo_ord_param}. A cutoff distance $R_{\text{max}}$ is typically introduced, beyond which all the components of the correlation function are approximated as zero
	\begin{equation}
		\label{app:eqn:cutoff_correlation_function}
		\rho_{ij} \left( \Delta \vec{R} \right) \approx 0, \qq{for} \abs{\Delta \vec{R} + \vec{r}_j - \vec{r}_i} > R_{\text{max}}.
	\end{equation}
	Because of the symmetry constraints from \cref{app:eqn:trafo_ord_param}, not all nonzero components of $\rho_{ij} \left( \Delta \vec{R} \right)$ are independent. In what follows, we will let $N_{\rho}$ denote the number of \emph{independent} real components of the correlation function $\rho_{ij} \left( \Delta \vec{R} \right)$.  
	
	The cutoff in \cref{app:eqn:cutoff_correlation_function} can be justified as follows:
	\begin{itemize}
		\item For bands with zero Chern number, the band projector is an analytic function of $\vec{k}$, making its real-space Fourier transform exponentially decaying in $\vec{R}$~\cite{BRO07}. Thus, large-distance components of the correlation function are suppressed and can be ignored.
		\item As shown in \cref{app:eqn:ldos_to_wannier}, Wannier orbitals farther apart than their typical spread contribute minimally to the CDD due to small overlap. The corresponding correlators cannot be determined and must be ignored.
		\item Experimental data has finite resolution and cannot reliably fit a large number of independent components.
	\end{itemize}

	\item Compute the spatial factor $B_{ij} \left( \vec{Q}, z, \Delta \vec{R} \right)$ using the method outlined in \cref{app:sec:correlator_general:band_ldos:decomposition} with Wannier wave functions determined from \textit{ab initio}{} simulations.

	\item Fourier transform the experimentally measured CDD to obtain $A \left( \vec{Q}, z \right)$. Due to finite experimental resolution, only a limited number of Fourier components can be resolved ({\it i.e.}{}, $A \left( \vec{Q}, z \right) \approx 0$ for large $\abs{\vec{Q}}$). Additionally, system symmetries reduce the number of \emph{independent} real components of $A \left( \vec{Q}, z \right)$. We denote this number by $N_{A}$.

	\item Determine the correlation function by solving \cref{app:eqn:sp_func_of_correlators} for $\rho_{ij} \left( \Delta \vec{R} \right)$ in the least-squares sense. We begin by writing \cref{app:eqn:sp_func_of_correlators} as a matrix-vector product
	\begin{equation}
		\label{app:eqn:fitting_op_mat_form}
		A_{\vec{Q}} = \sum_{I} B_{\vec{Q},I} \rho_{I},
	\end{equation}
	where $I = \left(i, j, \Delta \vec{R} \right)$ is a composite index and
	\begin{align}
		\rho_{I} \equiv & \rho_{ij} \left( \Delta \vec{R} \right), \\
		B_{\vec{Q},I} \equiv & B_{ij} \left( \vec{Q}, z, \Delta \vec{R} \right), \\
		A_{\vec{Q}} \equiv & A \left( \vec{Q}, z \right).
	\end{align}
	Since not all components of $A_{\vec{Q}}$ or $\rho_{I}$ are independent, \cref{app:eqn:fitting_op_mat_form} contains some redundancy ({\it e.g.}{}, equations corresponding to components of $A_{\vec{Q}}$ related by symmetry are duplicated, and there appear to be more unknowns than independent components of $\rho_{I}$). To eliminate this redundancy, we express both $A_{\vec{Q}}$ and $\rho_{I}$ in terms of their independent components,
	\begin{equation}
		A_{\vec{Q}} = \sum_{n=1}^{N_A} \mathcal{A}_{\vec{Q},n} \alpha_n \qq{and} \rho_{I} = \sum_{n=1}^{N_{\rho}} \mathcal{R}_{I,n} \beta_n,
	\end{equation}
	where $\alpha_n, \beta_n \in \mathbb{R}$, and the bases $\mathcal{A}_{\vec{Q},n}$ and $\mathcal{R}_{I,n}$ are chosen to be orthonormal without loss of generality
	\begin{equation}
		\sum_{\vec{Q}} \mathcal{A}^{*}_{\vec{Q},n} \mathcal{A}_{\vec{Q},n'} = \delta_{nn'} \qq{and}
		\sum_{I} \mathcal{R}^{*}_{I,n} \mathcal{R}_{I,n'} = \delta_{nn'}.
	\end{equation} 
	In this nonredundant basis, \cref{app:eqn:fitting_op_mat_form} becomes
	\begin{equation}
		\label{app:eqn:fitting_op_mat_form_non_redundant}
		\sum_{n=1}^{N_A} \mathcal{A}_{\vec{Q},n} \alpha_n = \sum_{n=1}^{N_{\rho}} \mathcal{B}_{\vec{Q},n} \beta_n, \qq{where} \mathcal{B}_{\vec{Q},n} \equiv \sum_{I} B_{\vec{Q},I} \mathcal{R}_{I,n}.
	\end{equation}
	\Cref{app:eqn:fitting_op_mat_form_non_redundant} is a linear system for the independent components of $\rho_{I}$ ({\it i.e.}{}, $\beta_n$, for $1 \leq n \leq N_{\rho}$). The coefficient $\alpha_n$, for $1 \leq n \leq N_{A}$, are known from the experimentally determined $A_{\vec{Q}}$. Depending on the rank $r$ of the matrix $\mathcal{B}_{\vec{Q},n}$, the number of independent components of the correlation function $N_{\rho}$, and the number of independent Fourier components of the CDD $N_{A}$, different scenarios emerge. The rank of $\mathcal{B}_{\vec{Q},n}$ is bounded by its second dimension, so $r \leq N_{\rho}$. Moreover, $r$ is, by definition, equal to the dimension of the column space of $\mathcal{B}_{\vec{Q},n}$. Since for any $\beta_n$, $\sum_{n=1}^{N_{\rho}} \mathcal{B}_{\vec{Q},n} \beta_n$ is a \textit{bona fide} symmetry-preserving CDD Fourier transformation, the column space of $\mathcal{B}_{\vec{Q},n}$ must be included in the space spanned by $\mathcal{A}_{\vec{Q},n}$, implying that $r \leq N_{A}$.
	
	In general, the upper bound derived above, $r \leq N_{A}, N_{\rho}$, should be saturated at generic heights $z$ (as all symmetries are explicitly accounted for by the nonredundant parameterizations for $A_{\vec{Q}}$ and $\rho_I$). If $N_{A} > N_{\rho}$, the system is overdetermined and can be solved in the least-squares sense by computing the Moore-Penrose inverse of $\mathcal{B}_{\vec{Q},n}$. If $N_{A} = N_{\rho}$, the system can be solved exactly by inverting $\mathcal{B}_{\vec{Q},n}$. If $N_{A} < N_{\rho}$, the system is underdetermined, necessitating a reduction in the number of nonzero independent components of the correlation function.
	
	Experiments determine $A \left( \vec{Q}, z \right)$ only up to an unspecified multiplicative constant, as described by the Tersoff-Hamann approximation in \cref{app:eqn:current_for_band}. Consequently, the correlators are determined only up to this same multiplicative factor. To ensure proper normalization, the elements should be rescaled by a constant so that they satisfy the normalization condition from \cref{app:eqn:normalization_rho_linear}.
\end{enumerate}

The advantage of the method outlined above is that solving \cref{app:eqn:fitting_op_mat_form} in the least-squares sense is computationally inexpensive. However, a drawback is that the information contained in the normalization condition from \cref{app:eqn:normalization_rho_square} is not explicitly utilized. This condition can be employed to reduce the number of independent parameters effectively. The trade-off for using this constraint is that the problem transitions into a nonlinear optimization problem.

To be specific, we first assume that the experiment measures $A' \left( \vec{Q}, z \right) = \gamma A \left( \vec{Q}, z \right)$, where $\gamma$ is an unknown positive constant. If both normalization conditions from \cref{app:eqn:normalization_rho_linear,app:eqn:normalization_rho_square} are to be satisfied by the fitted correlation function $\rho_{ij} \left( \Delta \vec{R} \right)$, $\gamma$ must be treated as a variational parameter. This leads to the following constrained minimization problem
\begin{equation}
	\begin{gathered}
		\min_{\rho_{ij} \left( \Delta \vec{R} \right), \gamma} \sum_{\vec{Q}} \abs{ A'\left( \vec{Q}, z \right) - \gamma \sum_{\Delta \vec{R}, i, j} B_{ij} \left( \vec{Q}, z, \Delta \vec{R} \right) \rho_{ij} \left( \Delta \vec{R} \right)}^2 \qq{subject to} \nonumber \\
		\sum_{i} \rho_{ii} \left( \vec{0} \right) = 1 \qq{and} \sum_{\vec{R},i,j} \abs{\rho_{ij} \left( \Delta \vec{R} \right)}^2 = 1.
	\end{gathered}
\end{equation}
However, the number of constraints and variational parameters can be reduced by one each if we instead fit $\tilde{\rho}_{ij} \left( \Delta \vec{R} \right) = \gamma \rho_{ij} \left( \Delta \vec{R} \right)$. As a result of \cref{app:eqn:normalization_rho_linear,app:eqn:normalization_rho_square}, $\tilde{\rho}_{ij} \left( \Delta \vec{R} \right)$ obeys
\begin{equation}
	\label{app:eqn:general_constraint_interm}
	\sum_{i} \tilde{\rho}_{ii} \left( \vec{0} \right) = \gamma \qq{and} \sum_{\vec{R},i,j} \abs{\rho_{ij} \left( \Delta \vec{R} \right)}^2 = \gamma^2.
\end{equation}
By eliminating $\gamma$ in \cref{app:eqn:general_constraint_interm}, we derive a single constraint for the scaled correlation function $\tilde{\rho}_{ij} \left( \Delta \vec{R} \right)$
\begin{equation}
	\label{app:eqn:general_constraint}
	\sum_{\vec{R},i,j} \abs{\tilde{\rho}_{ij} \left( \Delta \vec{R} \right)}^2 = \left( \sum_{i} \tilde{\rho}_{ii} \left( \vec{0} \right) \right)^2.
\end{equation}
This reformulation makes extracting the scaled correlation function $\tilde{\rho}_{ij} \left( \Delta \vec{R} \right)$ equivalent to solving the following constrained minimization problem
\begin{equation}
	\label{app:eqn:general_constraint_minimization}
	\min_{\tilde{\rho}_{ij} \left( \Delta \vec{R} \right)} \sum_{\vec{Q}} \abs{ A' \left( \vec{Q}, z \right) - \sum_{\Delta \vec{R}, i, j} B_{ij} \left( \vec{Q}, z, \Delta \vec{R} \right) \tilde{\rho}_{ij} \left( \Delta \vec{R} \right)}^2 \qq{subject to} \sum_{\vec{R},i,j} \abs{\tilde{\rho}_{ij} \left( \Delta \vec{R} \right)}^2 - \left( \sum_{i} \tilde{\rho}_{ii} \left( \vec{0} \right) \right)^2 = 0.
\end{equation}
After minimization, the correlation function $\tilde{\rho}_{ij} \left( \Delta \vec{R} \right)$ is normalized according to \cref{app:eqn:normalization_rho_linear} to obtain $\rho_{ij} \left( \Delta \vec{R} \right)$.

The choice between a least-squares or constrained-minimization approach depends on the problem at hand. In this work, we utilize a constrained minimization method in which the correlation function is efficiently parameterized, as detailed in \cref{app:sec:experimental:fit}.

\subsubsection{The independent components of \texorpdfstring{$A \left( \vec{Q}, z \right)$}{A(Q,z)} for triangular lattices}

\begin{figure}[!t]
	\centering
	\includegraphics[width=0.4\textwidth]{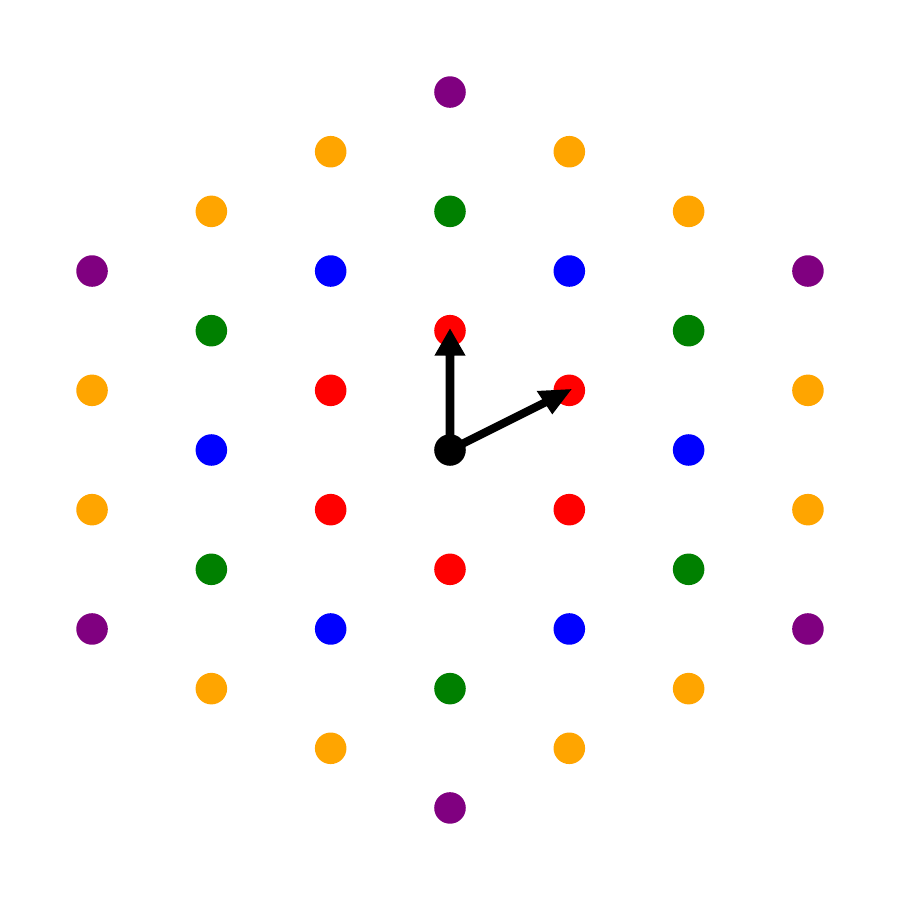}
	\caption{Harmonic decomposition of the Fourier-transformed CDD. The two black arrows represent the reciprocal lattice vectors $\vec{b}_{1}$ and $\vec{b}_{2}$. Dots of the same color correspond to reciprocal lattice vectors $\vec{Q} = n_1 \vec{b}_1 + n_2 \vec{b}_2$ (with $n_1, n_2 \in \mathbb{Z}$) that share the same length $\abs{\vec{Q}}$, {\it i.e.}{}, they represent harmonics of $A\left( \vec{Q}, z \right)$ of the same order.}
	\label{app:fig:BZ_circles}
\end{figure}

Finally, we consider the case of a triangular lattice and discuss the number of independent components of $A \left( \vec{Q}, z \right)$ under different symmetries. The primitive reciprocal lattice vectors of the system are given by
\begin{align}
	\vec{b}_1=\frac{2\pi}{a} \left( 1, \frac{1}{\sqrt{3}} \right), \quad 
	\vec{b}_2=\frac{2\pi}{a} \left( 0, \frac{2}{\sqrt{3}} \right),
\end{align}
where $a$ is the unit cell size. Labeling each reciprocal lattice vector $\vec{Q} = n_1 \vec{b}_1 + n_2 \vec{b}_2$ (with $n_1, n_2 \in \mathbb{Z}$) as $[n_1,n_2]$, we classify the reciprocal lattice vectors $\vec{Q}$ based on their distance from the origin:
\begin{itemize}
    \item zeroth harmonic:$[0,0]$.
    \item first harmonic: $[1, 0], [0, 1], [-1, 1], [-1, 0], [0, -1], [1, -1]$.
    \item second harmonic: $[1, 1], [-1, 2], [-2, 1], [-1, -1], [1, -2], [2, -1]$.
    \item third harmonic: $[2, 0], [0, 2], [-2, 2], [-2, 0], [0, -2], [2, -2]$.
    \item fourth harmonic: $[-3, 1], [-3, 2], [-2, -1], [-2, 3], [-1, -2], [-1, 3], [1, -3], [1, 2], [2, -3], [2, 1], [3, -2], [3, -1]$.
    \item fifth harmonic: $[-3, 0], [-3, 3], [0, -3], [0, 3], [3, -3], [3, 0]$.
\end{itemize}
The classification by $\abs{\vec{Q}}$ arises because $A \left( \vec{Q}, z \right)$ decreases in magnitude with increasing $\abs{\vec{Q}}$. Experimentally, only a small number of harmonics are resolved.  \Cref{app:fig:BZ_circles} illustrates the positions of $\vec{Q}$ for each order of harmonics. 

We now enumerate the independent \emph{real} components of $A \left( \vec{Q}, z \right)$ up to the fifth harmonic, depending on the symmetries of the system:
\begin{itemize}
	\item With $C_{6z}$ symmetry: Each set of harmonics with six $\vec{Q}$ vectors has only one symmetry-independent component. The set with 12 harmonics provides two symmetry-independent components if the system lacks $m_x$ or $m_y$ symmetry; otherwise, it also yields only one symmetry-independent component.
    \item With only $C_{3z}$ symmetry: Each set of harmonics with six $\vec{Q}$ vectors provides two symmetry-independent components, while the set with 12 harmonics furnishes four symmetry-independent components. If $m_x$ symmetry is additionally present, the first, third, fourth, and fifth harmonics each contribute two independent components, whereas the second harmonic contributes only one. Conversely, if $m_y$ symmetry is additionally present, the first, third, and fifth harmonics each have one independent component, while the second and fourth harmonics provide two independent components.
\end{itemize}
In all cases, the zeroth harmonic contributes one additional independent component, though it is often ignored due to low-wave vector noise, which can dominate this harmonic.

\section{First-principle results for \ch{NbSe2}}\label{app:sec:NbSe2_ab_initio}

In this \siSection{}, we present \textit{ab initio}{} results for the single-layer hexagonal phase of \ch{NbSe2}. The bulk phase, 2H-\ch{NbSe2}, comprises two layers per unit cell and is symmetric under the group $P6_3/mmc1'$ [Shubnikov Space Group (SSG) 194.264]. The single-layer phase, 1H-\ch{NbSe2}, is symmetric under the $p\bar{6}m2$ layer group or equivalently under the $P\bar{6}m21'$ space group (SSG 187.210), without the $\hat{\vec{z}}$-directional translation.

The bulk 2H-\ch{NbSe2} undergoes a charge density wave (CDW) transition at $T = \SI{33}{\kelvin}$, forming a quasi-commensurate $3 \times 3 \times 1$ CDW order~\cite{VAL04,WEB11,RAH12,SOU13,ARG14,ARG15}, and exhibits superconductivity below $T_c = \SI{7.2}{\kelvin}$~\cite{REV65,SOT07}. Similarly, the single-layer phase, 1H-\ch{NbSe2}, develops a $3 \times 3$ CDW order at ($T_{\text{CDW}} \approx \SI{33}{\kelvin}$)~\cite{WIL75,CHA15,XI15,UGE16,LIA18,LIN20,DRE21} and superconductivity ($T_{c} \approx \SI{2}{\kelvin}$)~\cite{CAO15,UGE16,XI16,LIA18,ZHA19a,DRE21,WAN22a}. Notably, the primary CDW displacements are associated with the Nb atoms.

In this work, we focus exclusively on the single-layer 1H-\ch{NbSe2} in its pristine ({\it i.e.}{}, non-CDW) phase, which we refer to as \ch{NbSe2} for brevity. First, we describe its crystal structure and symmetries. Next, we analyze its \textit{ab initio}{} electronic spectra, emphasizing the topology of its bands and their connection to the system's CDD. Subsequently, we construct Wannier tight-binding models that accurately reproduce the quasi-flat band near the Fermi energy in \ch{NbSe2}. We also review a remarkably simple and faithful tight-binding model for \ch{NbSe2} proposed in Ref.~\cite{YU24} and employ it to demonstrate the obstructed atomic nature of the flat band.

By computing the correlators $\rho_{ij} \left( \Delta \vec{R} \right)$ of the quasi-flat band both analytically and through \textit{ab initio}{} calculations, and comparing them with those of a fictitious non-obstructed atomic limit, we identify the key distinguishing features of the obstructed atomic phase. Finally, in the following \cref{app:sec:experimental}, we outline the experimental determination of the correlation functions, which provides unambiguous evidence for the obstructed atomic nature of the quasi-flat band in \ch{NbSe2}.

\subsection{Crystal structure}\label{app:sec:NbSe2_ab_initio:crystal_struct}
\begin{figure}[t]
	\centering
	\includegraphics[width=0.7\textwidth]{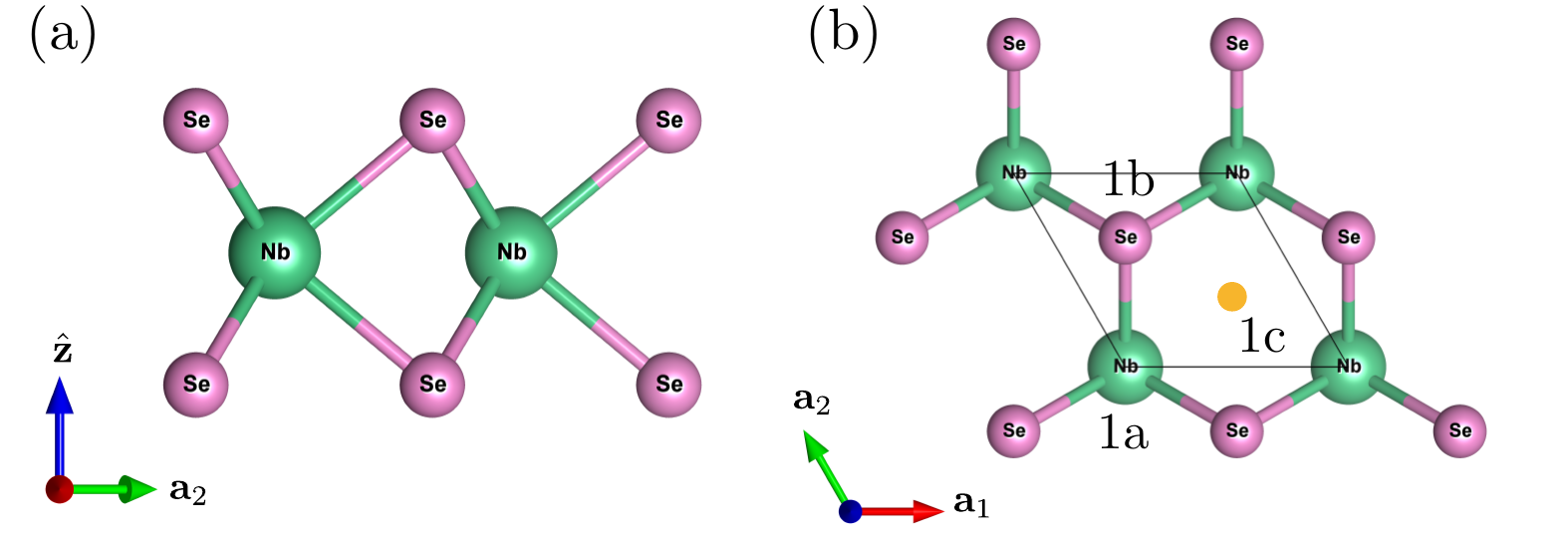}\subfloat{\label{app:fig:crystal_structure:a}}\subfloat{\label{app:fig:crystal_structure:b}}\caption{The crystalline structure of the single-layer hexagonal phase of \ch{NbSe2}. (a) is the side view and (b) is the top view. In (b), we mark the three $C_{3z}$-symmetric sites in the $xy$ plane by $1a$ (the Nb site), $1b$ (the Se site), and $1c$ (the empty site).}\label{app:fig:crystal_structure}\end{figure}

The crystal structure of \ch{NbSe2} is shown in \cref{app:fig:crystal_structure}, with the corresponding lattice vectors given in Cartesian coordinates by
\begin{equation}
	\label{app:eqn:unitcell_basis}
	\vec{a}_1=a \left( \frac{1}{2}, -\frac{\sqrt{3}}{2}, 0 \right), \quad \vec{a}_2=a \left( \frac{1}{2}, \frac{\sqrt{3}}{2}, 0 \right),
\end{equation}
where $a= \SI{3.47}{\angstrom}$ is the lattice constant. In the unit cell, the Nb atom occupies the $1a$ Wyckoff position located at $\left( 0,0,0 \right)$ with site symmetry group $\bar{6}m2$, while the two Se atoms occupy the $2f$ Wyckoff position at $\left( \frac{2}{3}, \frac{1}{3}, \pm z_{\text{Se}} \right)$, with $z_{\text{Se}}=\SI{1.67}{\angstrom}$, having $3m$ symmetry, as shown in \cref{app:fig:crystal_structure}. The first two coordinates of the Wyckoff positions are given in units of the lattice vectors, while the third coordinate is given in the $\hat{\vec{z}}$ direction.

Restricting attention to the $xy$ plane, there are three $C_{3z}$-symmetric Wyckoff positions in a unit cell, denoted by
\begin{equation}
	\label{app:eqn:c3z_pos_nbse2}
    1a: \left( 0,0,0 \right), \quad 1b: \left( \frac{1}{3}, \frac{2}{3}, 0 \right), \quad 1c: \left( \frac{2}{3}, \frac{1}{3}, 0 \right),
\end{equation}
with the first two coordinates expressed in terms of the lattice vectors from \cref{app:eqn:unitcell_basis}. The Nb atoms occupy the $1a$ position, with the two Se atoms in each unit cell positioned above and below the $1b$ site, while the $1c$ position remains vacant. In the next section, we show that the $1c$ position is, in fact, occupied by the Wannier center associated with the quasi-flat band near the Fermi level of \ch{NbSe2}. For clarity, as our focus will remain exclusively on the in-plane crystalline structure of \ch{NbSe2}, we will refer to the Se atoms as occupying the $1c$ position, even though they are technically positioned above and below the $1c$ site at the $2f$ Wyckoff positions.

\subsection{Electronic spectra of \ch{NbSe2}}\label{app:sec:NbSe2_ab_initio:dft_simulations}

\begin{figure}[t]
	\centering
	\includegraphics[width=\textwidth]{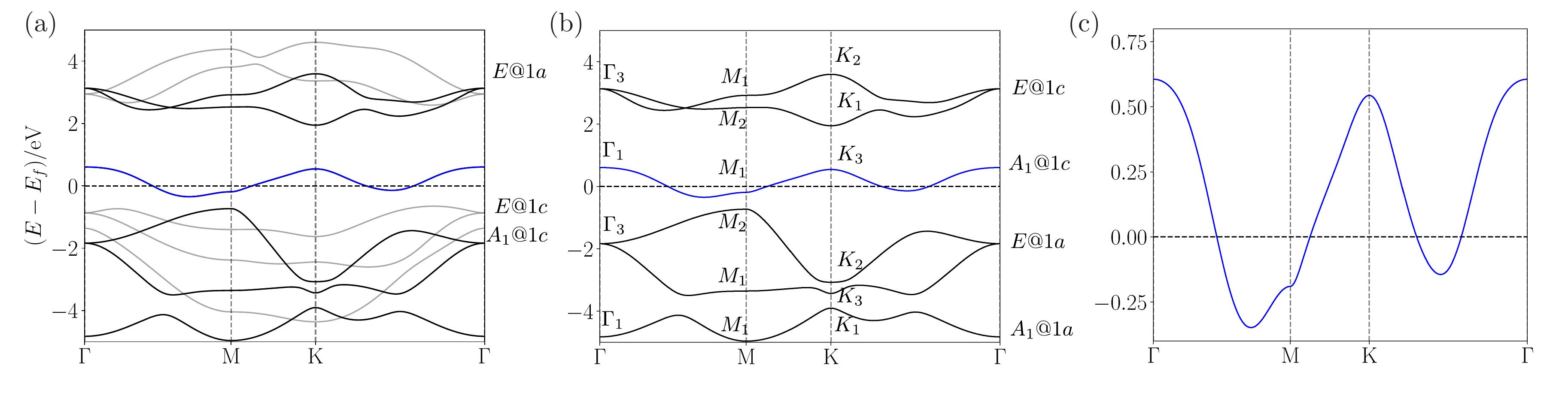}\subfloat{\label{app:fig:DFTbands:a}}\subfloat{\label{app:fig:DFTbands:b}}\subfloat{\label{app:fig:DFTbands:c}}\caption{The \textit{ab initio}{} band structure of monolayer \ch{NbSe2}. (a) illustrates the full band structure, with the bands of the mirror-even (mirror-odd) sectors shown in black and blue (gray). A focused view of the mirror-even bands is presented separately in (b) for clarity. (c) shows enlarged view of the quasi-flat band at the Fermi level. The band representations of the mirror-odd bands are indicated on the right-hand side of (a), while those of the mirror-even bands are explicitly shown in (b), along with the irreducible representations of each band. The quasi-flat band near the Fermi level corresponds to the band representation induced from $A_1 @ 1c$, representing an effective $s$ orbital located at the empty site $1c$.}
	\label{app:fig:DFTbands}
\end{figure}

\begin{figure}[t]
	\centering
	\includegraphics[width=\textwidth]{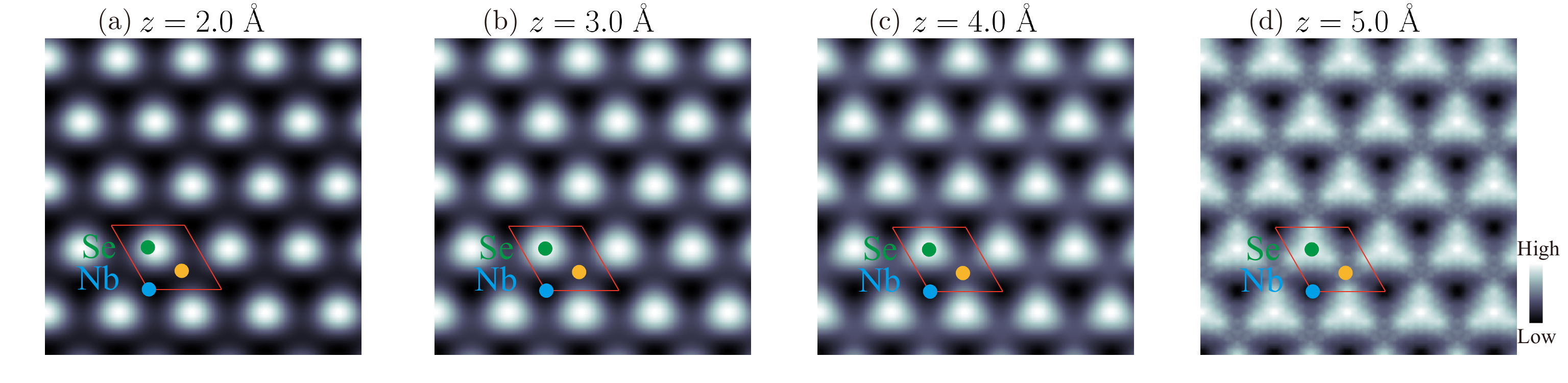}\subfloat{\label{app:fig:DFT-CDD-diffH:a}}\subfloat{\label{app:fig:DFT-CDD-diffH:b}}\subfloat{\label{app:fig:DFT-CDD-diffH:c}}\subfloat{\label{app:fig:DFT-CDD-diffH:d}}\caption{The \textit{ab initio}{} CDD of the quasi-flat band near the Fermi level in \ch{NbSe2} at various tip heights $z$, as indicated above each panel. The distance $z$, measured from the topmost layer of Se atoms, is specified in \si{\angstrom}. The unit cell (red diamond) and the three $C_{3z}$-symmetric sites (colored dots) are explicitly shown in each panel. Among these sites, the Se site $1b$ (green dots) is consistently the brightest (due to the Se atoms' proximity to the tip), followed by the empty site $1c$ (yellow dots), while the Nb site $1a$ (blue dots) appears the darkest.}
	\label{app:fig:DFT-CDD-diffH}
\end{figure}

\begin{figure}[t]
	\centering
	\includegraphics[width=0.8\textwidth]{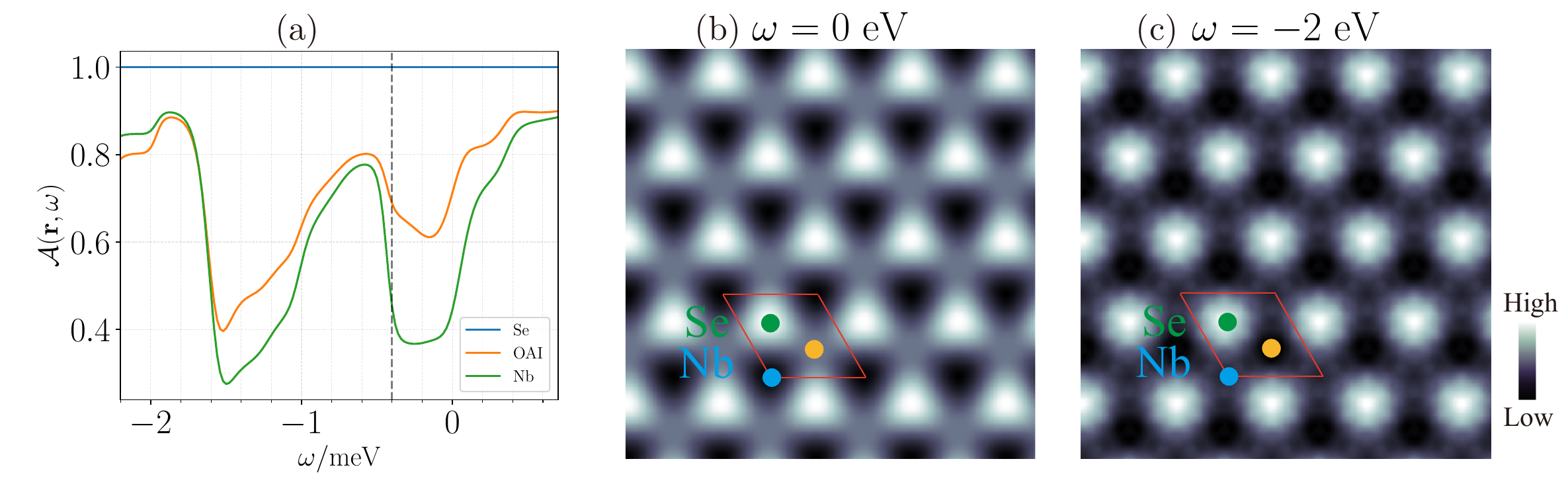}\subfloat{\label{app:fig:DFT-CDD-diffE:a}}\subfloat{\label{app:fig:DFT-CDD-diffE:b}}\subfloat{\label{app:fig:DFT-CDD-diffE:c}}\caption{The \textit{ab initio}{} spectral function $\mathcal{A} \left( \vec{r}, \omega \right)$ of \ch{NbSe2}. (a) presents the spectral function at various energies $\omega$ for the $1a$ (Nb), $1b$ (Se), and $1c$ (empty) sites. The dashed line at $\omega = \SI{-0.4}{\electronvolt}$ marks the gap between the quasi-flat band at the Fermi level and the lower bands. For a given energy $\omega$, the spectral function $\mathcal{A} \left( \vec{r}, \omega \right)$ is normalized such that its maximum value is 1 ({\it i.e.}{}, at the $1b$ site). Panels (b) and (c) show $\mathcal{A} \left( \vec{r}, \omega \right)$ at $\omega = \SI{0}{\electronvolt}$ and $\omega = \SI{-2}{\electronvolt}$, respectively. In (b), the intensity hierarchy of the spectral function is $1b \text{ (Se)} > 1c \text{ (empty)} > 1a \text{ (Nb)}$ (marked by the green, yellow, and blue circles, respectively), with a pronounced difference between the empty and Nb sites. In (c), the hierarchy of intensities is $1b \text{ (Se)} > 1a \text{ (Nb)} \gtrsim 1c \text{ (empty)}$, with only a slight difference between the empty and Nb site intensities. The tip distance is fixed at $z = \SI{4}{\angstrom}$ throughout the figure.}
	\label{app:fig:DFT-CDD-diffE}
\end{figure}

In this section, we analyze the \textit{ab initio} electronic band structures and simulated CDD of \ch{NbSe2}. The band structures presented in this work are computed using the Vienna Ab-initio Simulation Package (VASP)~\cite{KRE93, KRE93a, KRE94, KRE96a, KRE96} with the generalized gradient approximation of the Perdew-Burke-Ernzerhof (PBE) exchange-correlation functional~\cite{PER96}. An energy cutoff of $\SI{400}{\electronvolt}$ is employed. The maximally localized Wannier functions (MLWFs) are constructed using the symmetry-adapted Wannier functions~\cite{SAK13} in Wannier90~\cite{MAR12, MAR97b, PIZ20, SOU01c}, implemented in Quantum ESPRESSO~\cite{GIA17, GIA09}, with PAW-type pseudo-potentials and PBE functionals from PSlibrary 1.0.0~\cite{DAL14}. The spin-orbit coupling (SOC) effect is not included in this work. While SOC induces small spin splitting in the band structure, it does not alter the conclusions presented here.

The band structure of \ch{NbSe2} is shown in \cref{app:fig:DFTbands:a}. The 11 bands near the Fermi level originate from the five $d$ orbitals of Nb and the three $p$ orbitals of the two Se atoms. These bands are divided into two independent mirror $m_z$ sectors. Our analysis focuses on the mirror-even sector, shown in \cref{app:fig:DFTbands:b}, with the irreducible representations (IRREPs) at the high-symmetry momentum points being indicated in the plot.

A quasi-flat band near the Fermi level, with a bandwidth of approximately $\SI{1}{\electronvolt}$, is highlighted in \cref{app:fig:DFTbands:c}. Using topological quantum chemistry (TQC) and band representation (BR) theory~\cite{BRA17,CAN18a,ELC21, ZAK82}, we identify that this quasi-flat band originates from an effective $s$ orbital located at the unoccupied $1c$ Wyckoff position. It corresponds to the BR induced from $A_1@1c$, with IRREPs $\Gamma_1$, $M_1$, and $K_3$, and constitutes an obstructed atomic (OA) phase. For simplicity, the IRREPs of the mirror-even sector are calculated under the $P3m11'$ group (SSG 156.50), which governs the symmetry of this mirror sector. The BRs of the other bands are also marked in \cref{app:fig:DFTbands:a,app:fig:DFTbands:b}, where the three mirror-even occupied bands have Wannier centers at the Nb $1a$ position, while the three mirror-odd occupied bands are centered at the empty $1c$ position. Finally, it is worth mentioning that the abundance of obstructed bands in the system has led to the identification of many higher-order topological phases in transition metal dichalcogenides (TMDs) with crystalline structures similar to \ch{NbSe2}~\cite{ZEN21,JUN22,QIA22,SOD22}.

Next, we analyze the CDD of the quasi-flat band from \cref{app:fig:DFTbands:c}. In \cref{app:fig:DFT-CDD-diffH}, we present the CDD at tip distances $z$ in the range $2 \leq z/\si{\angstrom} \leq 5$, with $z$ denoting the distance between the ``tip'' ({\it i.e.}{}, the location at which the CDD of the quasi-flat band is probed) and the topmost Se layer. Across this range of heights, the Se site at $1b$ consistently exhibits the largest intensity, followed by the empty site at $1c$, while the Nb site at $1a$ is the least bright. At larger tip distances, the CDD becomes too weak and noisy for reliable analysis, and those results are excluded.

In \cref{app:fig:DFT-CDD-diffE}, we present the spectral function $\mathcal{A} \left( \vec{r}, \omega \right)$ of \ch{NbSe2} simulated at various energies $\omega$. By inspecting the Tersoff-Hamann expression for the tunneling current in \cref{app:eqn:final_TH_expression_current}, we note that $\mathcal{A} \left( \vec{r}, \omega \right)$ can be directly obtained from differential conductance ({\it i.e.}{}, $\dv{I}{V}$) measurements in STM experiments. A detailed comparison of the spectral function at the three $C_{3z}$-symmetric sites, $1a$ (Nb), $1b$ (Se), and $1c$ (empty), reveals the following trends as a function of tunneling energy $\omega$:
\begin{enumerate}
	\item The Se site consistently exhibits the largest differential conductance at all energies. This is expected since the Se atoms are closest to the STM tip, as shown in the crystal structure in \cref{app:fig:crystal_structure:a}.
	\item Within the energy range of the quasi-flat band, $-0.4 \leq \omega/ \si{\electronvolt} \leq 0.6$, the $1c$ (empty) site shows significantly higher differential conductance than the Nb site. This is consistent with the BR analysis discussed earlier, which places the Wannier center of the quasi-flat band at the empty $1c$ position.
	\item Near $\omega = \SI{-2}{\electronvolt}$, the simulated differential conductance at the $1a$ (Nb) and $1c$ (empty) sites is comparable, with the intensity at the $1b$ (Se) site being slightly larger. This can be attributed to the bands in this energy range being induced from Wannier orbitals with Wannier centers located at both the Nb and the empty sites.
	\item For $-1.5 \leq \omega/ \si{\electronvolt} \leq -1.0$, the $1c$ (empty) site surpasses the $1a$ (Nb) site in intensity. This is due to contributions from the mirror-odd sector, whose Wannier centers are located at the $1c$ site, as indicated in \cref{app:fig:DFTbands:a}.
\end{enumerate}

We will use the hierarchy of intensities among the three $C_{3z}$-symmetric sites, as described in \cref{app:sec:experimental:extract_CDD:fix_unit_cell} to identify the unit cell origin in the experimental data.

\subsection{Wannier tight-binding models and the correlation function}\label{app:sec:NbSe2_ab_initio:wanniers}

With the \textit{ab initio}{} band structure of \ch{NbSe2} at hand, we construct tight-binding models for the system using Wannier90. The models presented here are similar to those reported in the literature for other TMD monolayers that share the same crystalline structure as \ch{NbSe2}~\cite{LIU13a,CAP13,FAN15,YU24}.  Starting from different MLWF bases, we can obtain the following five tight-binding models:
\begin{itemize}
    \item An 11-band TB model that reproduces all 11 bands around the Fermi level in \cref{app:fig:DFTbands:a} can be constructed based on the MLWFs obtained from the five $d$ orbitals of Nb and the three $p$ orbitals of the two Se atoms.
	
    \item From the 11-band model, we can extract a six-band model by restricting to the mirror-even sector. The MLWFs of this model are obtained from the three mirror-symmetric $d$ orbitals of Nb and the mirror-even combination of Se $p$ orbitals
	\begin{equation}
		\left( 
			d_{z^2}, d_{x^2-y^2}, d_{xy}, 
			\frac{1}{\sqrt{2}} \left( p_z^{\text{Se}_1} - p_z^{\text{Se}_2} \right),
			\frac{1}{\sqrt{2}} \left( p_x^{\text{Se}_1} + p_x^{\text{Se}_2} \right),
			\frac{1}{\sqrt{2}} \left( p_y^{\text{Se}_1} + p_y^{\text{Se}_2} \right) 
		\right).
	\end{equation}
    This model reproduces the bands from \cref{app:fig:DFTbands:b}. 
	
	\item A simplified five-band tight-binding model can be constructed starting from the five $d$ orbitals of Nb, which can reproduce the top five bands shown in \cref{app:fig:DFTbands:a}. The MLWFs of this model also have additional contributions from the $p$ orbitals of Se.

	\item A further simplified three-band model is obtained by restricting the five-orbital model to the mirror-even sector. This model reproduces the top three bands in \cref{app:fig:DFTbands:b}. The MLWFs of this model are built starting from the $\left( d_{z^2}, d_{x^2-y^2}, d_{xy} \right)$ orbitals of Nb. After optimizing the wave functions in Wannier90, the final MLWFs also contain weights from the $p$ orbitals of Se.
		
    \item Finally, a one-band model that reproduces the quasi-flat band near the Fermi energy from \cref{app:fig:DFTbands:c} can be constructed. The MLWFs of this model are obtained in Wannier90 by starting from an $s$ orbital located at the $1c$ $C_{3z}$-symmetric site.
\end{itemize}

\subsubsection{The three-orbital model of \ch{NbSe2}}\label{app:sec:NbSe2_ab_initio:wanniers:three-orb}

\begin{figure}[t]
	\centering
	\includegraphics[width=0.5\textwidth]{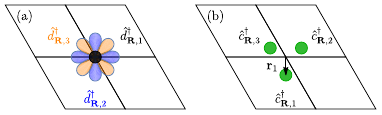}\subfloat{\label{app:fig:orb_trafo:a}}\subfloat{\label{app:fig:orb_trafo:b}}\caption{Different bases for the three-orbital model of \ch{NbSe2}. The three-orbital model of \ch{NbSe2} can be described in terms of $d$-like orbitals located at the $1a$ position, as shown in (a). The three operators $\hat{d}^\dagger_{\vec{R},1}$, $\hat{d}^\dagger_{\vec{R},2}$, and $\hat{d}^\dagger_{\vec{R},3}$ correspond respectively to the MLWFs obtained from the $d_{z^2}$, $d_{x^2-y^2}$, and $d_{xy}$ orbitals of Nb. These $d$-like orbitals can be recombined into $s$-like orbitals located \emph{around} the $1a$ position, as depicted in (b). The Wannier centers of these $s$-like orbitals are slightly displaced from the unit cell origin (the distance of their Wannier center from the $1a$ position is exaggerated for clarity).}
	\label{app:fig:orb_trafo}
\end{figure}

Throughout this work, we focus exclusively on the three-orbital model of \ch{NbSe2}. Specifically, the MLWFs of this model are obtained starting from the mirror-symmetric $d$-orbitals of Nb, which, after optimization, also contain weights from the $p$ orbitals of Se. The three $d$-orbital MLWFs are located at the origin of the unit cell ($\vec{r}^{d}_i = \vec{0}$ for $i = 1, 2, 3$). In a notation similar to \cref{app:sec:correlator_general:notation}, the creation operators for the $d_{z^2}$, $d_{x^2-y^2}$, and $d_{xy}$ orbitals in the unit cell $\vec{R}$ are denoted, respectively, by $\hat{d}^\dagger_{\vec{R},1}$, $\hat{d}^\dagger_{\vec{R},2}$, and $\hat{d}^\dagger_{\vec{R},3}$. 

In the notation of \cref{app:eqn:tb_general_notation_symmetries}, the matrix representations of the $C_{3z}$ and $m_x$ symmetries of the system are given by
\begin{equation}
	D_{3d} \left( C_{3z} \right) = \begin{pmatrix}
		1 & 0 & 0 \\
		0 & -\frac{1}{2} & -\frac{\sqrt{3}}{2} \\
		0 & \frac{\sqrt{3}}{2} & -\frac{1}{2}
	\end{pmatrix}, \quad
	D_{3d} \left( m_x \right) = \begin{pmatrix}
		1 & 0 & 0 \\
		0 & 1 & 0 \\
		0 & 0 & -1
	\end{pmatrix}.
\end{equation}

A more intuitive, but equivalent, basis~\cite{YU24} can be obtained by recombining the three $d$-like orbitals of the model into effective $s$-like orbitals also located at the origin of the unit cell. To achieve this, we introduce the following unitary transformation
\begin{equation}
	S = \begin{pmatrix}
		\frac{1}{\sqrt{3}} & \sqrt{\frac{2}{3}} & 0 \\
		\frac{1}{\sqrt{3}} & -\frac{1}{\sqrt{6}} & -\frac{1}{\sqrt{2}} \\
		\frac{1}{\sqrt{3}} & -\frac{1}{\sqrt{6}} & \frac{1}{\sqrt{2}} \\
	\end{pmatrix}.
\end{equation}
The three $s$-like orbitals, denoted by $\hat{c}^\dagger_{\vec{R},i}$ for $i=1,2,3$, are defined as
\begin{equation}
	\begin{pmatrix}
		\hat{c}^\dagger_{\vec{R},1} \\ \hat{c}^\dagger_{\vec{R},2} \\ \hat{c}^\dagger_{\vec{R},3}
	\end{pmatrix} = 
	S \begin{pmatrix}
		\hat{d}^\dagger_{\vec{R},1} \\ \hat{d}^\dagger_{\vec{R},2} \\ \hat{d}^\dagger_{\vec{R},3}
	\end{pmatrix}.
\end{equation}
The symmetry representation matrices in this $s$-like basis are given by
\begin{equation}
	\label{app:eqn:sym_rep_3s_basis}
	D_{3s} \left( C_{3z} \right) = S D_{3d} \left( C_{3z} \right) S^{-1} = \begin{pmatrix}
		0 & 1 & 0 \\
		0 & 0 & 1 \\
		1 & 0 & 0
	\end{pmatrix}, \quad
	D_{3s} \left( m_x \right) = S D_{3d} \left( m_x \right) S^{-1} =  \begin{pmatrix}
		1 & 0 & 0 \\
		0 & 0 & 1 \\
		0 & 1 & 0
	\end{pmatrix}.
\end{equation}
The three $d$-like and $s$-like orbitals of the model are shown schematically in \cref{app:fig:orb_trafo}. The $d$-like orbitals are located exactly at the $1a$ Wyckoff position, while the $s$-like orbitals are arranged in a $C_{3z}$-symmetric configuration \emph{around} the $1a$ Wyckoff position. In particular, $\hat{c}^\dagger_{\vec{R},1}$ remains invariant under $m_x$ symmetry, while $\hat{c}^\dagger_{\vec{R},2}$ and $\hat{c}^\dagger_{\vec{R},3}$ are exchanged by this operation, as dictated by the representation matrix in \cref{app:eqn:sym_rep_3s_basis}. Additionally, the $s$-like orbitals are slightly displaced \emph{away} from the $1a$ Wyckoff position. This displacement is small ($\abs{\vec{r}_i} \ll \abs{\vec{a}_1}$) and unimportant for the current discussion ({\it i.e.}{}, it is \emph{not} fixed by symmetry). 

In the $s$-orbital basis, the obstructed nature of the quasi-flat band of \ch{NbSe2} becomes immediately apparent~\cite{YU24}. In particular, the \textit{ab initio}{} tight-binding Hamiltonian can be expressed as~\cite{YU24}
\begin{equation}
	\label{app:eqn:NbSe2_simple_hamiltonian}
	\hat{H} = \sum_{\vec{R}} \hat{H}_{\text{Loc}} \left( \vec{R} \right) + \dots, 
\end{equation}
where the ``local'' Hamiltonian term is given by
\begin{equation}
	\label{app:eqn:NbSe2_simple_loc_hamiltonian}
	\hat{H}_{\text{Loc}} \left( \vec{R} \right) = \begin{pmatrix}
		\hat{c}^\dagger_{\vec{R} + \vec{a}_1 + \vec{a}_2,1} &
		\hat{c}^\dagger_{\vec{R},2} &
		\hat{c}^\dagger_{\vec{R} + \vec{a}_1,3} \\
	\end{pmatrix} 
	\begin{pmatrix}
		E_0 & t & t \\
		t & E_0 & t \\
		t & t & E_0 \\
	\end{pmatrix}
	\begin{pmatrix}
		\hat{c}_{\vec{R} + \vec{a}_1 + \vec{a}_2,1} \\
		\hat{c}_{\vec{R},2} \\
		\hat{c}_{\vec{R} + \vec{a}_1,3}
	\end{pmatrix} .
\end{equation}
The hopping amplitude and onsite terms are given by $t = \SI{-0.7840}{\electronvolt}$ and $E_0 = \SI{1.733}{\electronvolt}$, and the dots ``$\dots$'' in \cref{app:eqn:NbSe2_simple_hamiltonian} denote other terms with hopping amplitudes no larger than $\SI{0.3}{\electronvolt}$~\cite{YU24}, which will be ignored for now. Because the different ``local'' Hamiltonian terms commute, $\commutator{H_{\text{Loc}} \left( \vec{R} \right)}{H_{\text{Loc}} \left( \vec{R}' \right)} = 0$ (for $\vec{R} \neq \vec{R}'$), the eigenstates of $\hat{H}$ can be straightforwardly determined in real space (when the terms represented by the dots are ignored). Specifically, the system features three flat bands with energies $\epsilon_1 = E_0 + 2t$, $\epsilon_2 = E_0 - t$, and $\epsilon_3 = E_0 - t$, whose real-space operators exhibit \emph{compact support} and are explicitly given by
\begin{align}
	\hat{\gamma}^\dagger_{\vec{R},1} &= \frac{1}{\sqrt{3}} \left( \hat{c}^\dagger_{\vec{R} + \vec{a}_1 + \vec{a}_2,1} + \hat{c}^\dagger_{\vec{R},2} + \hat{c}^\dagger_{\vec{R} + \vec{a}_1,3} \right), \label{app:eqn:OAI_wave_function_simple}\\
	\hat{\gamma}^\dagger_{\vec{R},2} &= \frac{1}{\sqrt{2}} \left( - \hat{c}^\dagger_{\vec{R} + \vec{a}_1 + \vec{a}_2,1} + \hat{c}^\dagger_{\vec{R} + \vec{a}_1,3} \right), \\
	\hat{\gamma}^\dagger_{\vec{R},3} &= \frac{1}{\sqrt{6}} \left( - \hat{c}^\dagger_{\vec{R} + \vec{a}_1 + \vec{a}_2,1} + 2 \hat{c}^\dagger_{\vec{R},2} - \hat{c}^\dagger_{\vec{R} + \vec{a}_1,3} \right). 
\end{align}
In the full model, $\hat{\gamma}^\dagger_{\vec{R},1}$ corresponds to the quasi-flat band near the Fermi energy, while $\hat{\gamma}^\dagger_{\vec{R},2}$ and $\hat{\gamma}^\dagger_{\vec{R},3}$ represent the topmost mirror-even valence bands of \ch{NbSe2}. The dispersion of these bands in the full \textit{ab initio}{} model arises from the additional hopping terms in the tight-binding Hamiltonian from \cref{app:eqn:NbSe2_simple_hamiltonian}, which are ignored in this simplified model.

The obstructed atomic nature of $\hat{\gamma}^\dagger_{\vec{R},1}$ can be immediately inferred by noting that its center of charge is located at $\vec{r}^{\gamma}_1 = \frac{1}{3} \left( \vec{R} +  \vec{a}_1 + \vec{a}_2 + \vec{R} + \vec{R} + \vec{a}_1  \right) - \vec{R} = \left(\frac{2}{3} \vec{a}_1 + \frac{1}{3} \vec{a}_2 \right)$, {\it i.e.}{}, at the $1c$ Wyckoff position. Moreover, under the symmetries of the problem, $\hat{\gamma}^\dagger_{\vec{R},1}$ transforms as an $s$ orbital located at the $1c$ Wyckoff position. Using the representation matrices in \cref{app:eqn:sym_rep_3s_basis}, it is straightforward to show that
\begin{equation}
	C_{3z} \hat{\gamma}^\dagger_{\vec{R},1} C^{-1}_{3z} = \hat{\gamma}^\dagger_{C_{3z} \left( \vec{R} + \vec{r}^{\gamma}_1 \right) - \vec{r}^{\gamma}_1,1}, \quad
	m_x \hat{\gamma}^\dagger_{\vec{R},1} m^{-1}_x = \hat{\gamma}^\dagger_{\vec{R},1}.
\end{equation}
This confirms that $\hat{\gamma}^\dagger_{\vec{R},1}$ indeed corresponds to an obstructed phase, as its Wannier center is located at the empty site. Finally, reintroducing the neglected hopping terms in \cref{app:eqn:NbSe2_simple_hamiltonian} does not close the gap between the upper flat band and the other two bands, indicating that the flat band in the \textit{ab initio}{} model is also obstructed, consistent with the simplified model. Remarkably, the overlap between the simplified wave function in \cref{app:eqn:OAI_wave_function_simple} and the \textit{ab initio}{} wave functions is as high as $94\%$~\cite{YU24}.

\subsubsection{The correlation function of the quasi-flat band of \ch{NbSe2}}\label{app:sec:NbSe2_ab_initio:wanniers:order_parameters}

\begin{table}[t]
	\centering
	\begin{tabular}{|c|c|c|c|}
		\hline
		Matrix elements & Correlator & Value & Multiplicity \\\hline
		$\rho_{11}([0,0]), \rho_{22}([0,0]), \rho_{33}([0,0])$ & $O_0$ & 0.333 & 3 \\\hline
		\makecell{$\rho_{12}([0,0]), \rho_{13}([0,0]), \rho_{23}([0,0])$, \\ 
			$\rho_{21}([0,0]), \rho_{31}([0,0]), \rho_{32}([0,0])$} 
		& $O_1$ & 0.064 & 6 \\\hline
		\makecell{$\rho_{23}([1,0]), \rho_{21}([1,1]), \rho_{13}([0,-1])$, \\
			$\rho_{32}([-1,0]), \rho_{12}([-1,-1]), \rho_{31}([0,1])$}
		& $O_2$ & 0.306 & 6 \\\hline
		$\rho_{11}([1,0]),\dots$ & $O_3$ & -0.075 & 6 \\\hline
		$\rho_{22}([1,0]),\dots$ & $O_4$ & 0.038 & 12 \\\hline
		$\rho_{12}([1,0]),\dots$ & $O_5$ & -0.011 & 12 \\\hline
		$\rho_{13}([1,0]),\dots$ & $O_6$ & -0.014 & 12 \\\hline
		$\rho_{32}([1,0]),\dots$ & $O_7$ & -0.013 & 6 \\
		\hline
	\end{tabular}
	\caption{The indepdent correlators of the \ch{NbSe2} quasi-flat band within the three-orbital model with $s$-like orbital. For each independent correlator, the first column lists the matrix elements that are equal to it, while the third column list its value, as obtained in \textit{ab initio}{} compoutations. The last column lists the multiplicity of the correlator, which is defined as the number of distinct matrix elements equal to the corresponding correlator. For denoting the matrix elements, we use a shorthand notation in which $\rho_{ij} ([n,m]) \equiv \rho_{ij} \left(n \vec{a}_1 + m \vec{a}_2 \right)$. For $O_{1-3}$ we list all matrix elements equal to them, while for $O_{3-5}$, we only list a representation matrix element $\rho_{ij} \left( \Delta \vec{R} \right)$ equal to the corresponding correlator. It can be seen that the onsite correlator $O_0$ and inter-cell NN correlator $O_2$ have dominant values about $\frac{1}{3}$, while all other correlators [including the longer distance matrix elements of $\rho_{ij} \left( \Delta \vec{R} \right)$ which are not shown here] take very small values. }
	\label{app:tab:OPs-3band-model}
\end{table}

The fact that the quasi-flat band of \ch{NbSe2} constitutes an OA phase has significant implications for the correlation function $\rho \left( \Delta \vec{R} \right)$ associated with the band, which can also be experimentally extracted through STM measurements. By leveraging the symmetries of the system (namely time-reversal, $C_{3z}$, and $m_x$ symmetries), we constrain the correlation function using \cref{app:eqn:trafo_ord_param} and subsequently parameterize it into independent components. The parameterization of the correlation function in the $s$-orbital model is summarized in \cref{app:tab:OPs-3band-model}, along with the values extracted from \textit{ab initio}{} simulations.

Each independent component of $\rho \left( \Delta \vec{R} \right)$ is referred to as a correlator. As shown in \cref{app:tab:OPs-3band-model}, $O_0$ corresponds to the onsite correlator, $O_1$ represents the intra-cell nearest-neighbor (NN) correlator, and $O_2$ denotes the inter-cell NN correlator that corresponds to the orbitals that are strongly hybridized in the obstructed phase. The remaining correlators, $O_{3-5}$, correspond to other inter-cell NN correlators. It can be observed that the onsite correlator $O_0$ and inter-cell NN correlator $O_2$ have dominant values of approximately $\frac{1}{3}$, while all other correlators exhibit small values. In what follows, we provide an analytical explanation of these results using the simplified wave function of the quasi-flat band in \ch{NbSe2}.

\subsubsection{The atomic and obstructed atomic limits of \ch{NbSe2}}\label{app:sec:NbSe2_ab_initio:wanniers:limits}

We now use the simplified model of the quasi-flat band in \ch{NbSe2} to provide an intuitive explanation of the correlator values presented in \cref{app:tab:OPs-3band-model} and their connection to the OA nature of the band. Due to the compact support of the flat band eigenstates, the correlation function in real space can be directly and exactly computed from its definition in \cref{app:eqn:def_order_parameter_matrix_r_space}, taking $\ket{\Phi} = \left( \prod_{\vec{R}} \hat{\gamma}^\dagger_{\vec{R},1} \right) \ket{0}$. Specifically, the correlation function is given by
\begin{equation}
	\label{app:eqn:computing_OP_for_NbSe2_OAI}
	\rho_{ij} \left( \Delta \vec{R} \right) = \mel**{\Phi}{\hat{c}^\dagger_{\vec{0},i} \hat{c}_{\Delta \vec{R},j}}{\Phi} = \sum_{\vec{R}} \bra{0} \hat{\gamma}_{\vec{R},1} \hat{c}^\dagger_{\vec{0},i} \hat{c}_{\Delta \vec{R},j} \hat{\gamma}^\dagger_{\vec{R},1} \ket{0} = \sum_{\vec{R} \in \left\lbrace \vec{0}, \vec{a}_1, \vec{a}_1 + \vec{a}_2 \right\rbrace} \bra{0} \hat{\gamma}_{-\vec{R},1} \hat{c}^\dagger_{\vec{0},i} \hat{c}_{\Delta \vec{R},j} \hat{\gamma}^\dagger_{-\vec{R},1} \ket{0}.
\end{equation}
It is straightforward to see that only two components of the correlation function are nonzero. Using the notation of \cref{app:tab:OPs-3band-model}, we find that $O_0 = O_2 = \frac{1}{3}$, while all other correlators vanish. This immediately accounts for the correlator values obtained from \textit{ab initio}{} calculations. The small, yet nonzero, values of $O_i$ for $i \neq 0,2$ in \cref{app:tab:OPs-3band-model} arise from the small hopping terms neglected in \cref{app:eqn:NbSe2_simple_hamiltonian}.

To better understand the relation between these correlator values and the OA phase, we consider a \emph{fictitious} unobstructed atomic (UA) phase of \ch{NbSe2}. Analogous to \cref{app:eqn:NbSe2_simple_hamiltonian,app:eqn:NbSe2_simple_loc_hamiltonian}, we define a Hamiltonian for this fictitious UA limit as
\begin{equation}
	\label{app:eqn:NbSe2_simple_AI_hamiltonian}
	\hat{H}' = \sum_{\vec{R}} \hat{H}'_{\text{Loc}} \left( \vec{R} \right), 
\end{equation}
where the ``local'' Hamiltonian term is given by
\begin{equation}
	\label{app:eqn:NbSe2_simple_AI_loc_hamiltonian}
	\hat{H}'_{\text{Loc}} \left( \vec{R} \right) = \begin{pmatrix}
		\hat{c}^\dagger_{\vec{R},3} &
		\hat{c}^\dagger_{\vec{R},2} &
		\hat{c}^\dagger_{\vec{R},1} \\
	\end{pmatrix} 
	\begin{pmatrix}
		E_0 & t & t \\
		t & E_0 & t \\
		t & t & E_0 \\
	\end{pmatrix}
	\begin{pmatrix}
		\hat{c}_{\vec{R},1} \\
		\hat{c}_{\vec{R},2} \\
		\hat{c}_{\vec{R},3}
	\end{pmatrix}. 
\end{equation}
This Hamiltonian is similar to the one in \cref{app:eqn:NbSe2_simple_hamiltonian,app:eqn:NbSe2_simple_loc_hamiltonian}, but instead of hybridizing $s$-like orbitals across unit cells, it only hybridizes them within each unit cell. The energy spectrum of $\hat{H}'$ is identical to that of $\hat{H}$; however, the eigenstates are now given by
\begin{align}
	\hat{\gamma'}^\dagger_{\vec{R},1} &= \frac{1}{\sqrt{3}} \left( \hat{c}^\dagger_{\vec{R},1} + \hat{c}^\dagger_{\vec{R},2} + \hat{c}^\dagger_{\vec{R},3} \right), \label{app:eqn:OAI_wave_function_simple_fake}\\
	\hat{\gamma'}^\dagger_{\vec{R},2} &= \frac{1}{\sqrt{2}} \left( - \hat{c}^\dagger_{\vec{R},1} + \hat{c}^\dagger_{\vec{R},3} \right), \\
	\hat{\gamma'}^\dagger_{\vec{R},3} &= \frac{1}{\sqrt{6}} \left( - \hat{c}^\dagger_{\vec{R},1} + 2 \hat{c}^\dagger_{\vec{R},2} - \hat{c}^\dagger_{\vec{R},3} \right). 
\end{align}
In this fictitious UA phase, the Wannier center of $\hat{\gamma'}^\dagger_{\vec{R},1}$ corresponds to the $1a$ Wyckoff position. This is evident from the symmetry transformations
\begin{equation}
	C_{3z} \hat{\gamma'}^\dagger_{\vec{R},1} C^{-1}_{3z} = \hat{\gamma'}^\dagger_{\vec{R},1}, \quad
	m_x \hat{\gamma'}^\dagger_{\vec{R},1} m^{-1}_x = \hat{\gamma'}^\dagger_{\vec{R},1}.
\end{equation}
In fact, $\hat{\gamma'}^\dagger_{\vec{R},1}$ is nothing but the $\hat{d}^\dagger_{\vec{R},1}$ fermion, which corresponds to the $d_{z^2}$ orbital of Nb. In this fictitious UA limit, we can trivially compute the correlators of the topmost band
\begin{equation}
	\label{app:eqn:computing_OP_for_NbSe2_AI}
	\rho'_{ij} \left( \Delta \vec{R} \right) = \mel**{\Phi'}{\hat{c}^\dagger_{\vec{0},i} \hat{c}_{\Delta \vec{R},j}}{\Phi'} = \sum_{\vec{R}} \bra{0} \hat{\gamma'}_{\vec{R},1} \hat{c}^\dagger_{\vec{0},i} \hat{c}_{\Delta \vec{R},j} \hat{\gamma'}^\dagger_{\vec{R},1} \ket{0} = \bra{0} \hat{\gamma'}_{\vec{0},1} \hat{c}^\dagger_{\vec{0},i} \hat{c}_{\Delta \vec{R},j} \hat{\gamma'}^\dagger_{\vec{0},1} \ket{0},
\end{equation}
where $\ket{\Phi'} = \left( \prod_{\vec{R}} \hat{\gamma'}^\dagger_{\vec{R},1} \right) \ket{0}$. Unlike \cref{app:eqn:computing_OP_for_NbSe2_OAI}, we now find $O_0 = O_1 = \frac{1}{3}$, while all other correlators from \cref{app:tab:OPs-3band-model} are zero.

\begin{figure}[t]
	\centering
	\includegraphics[width=0.7\textwidth]{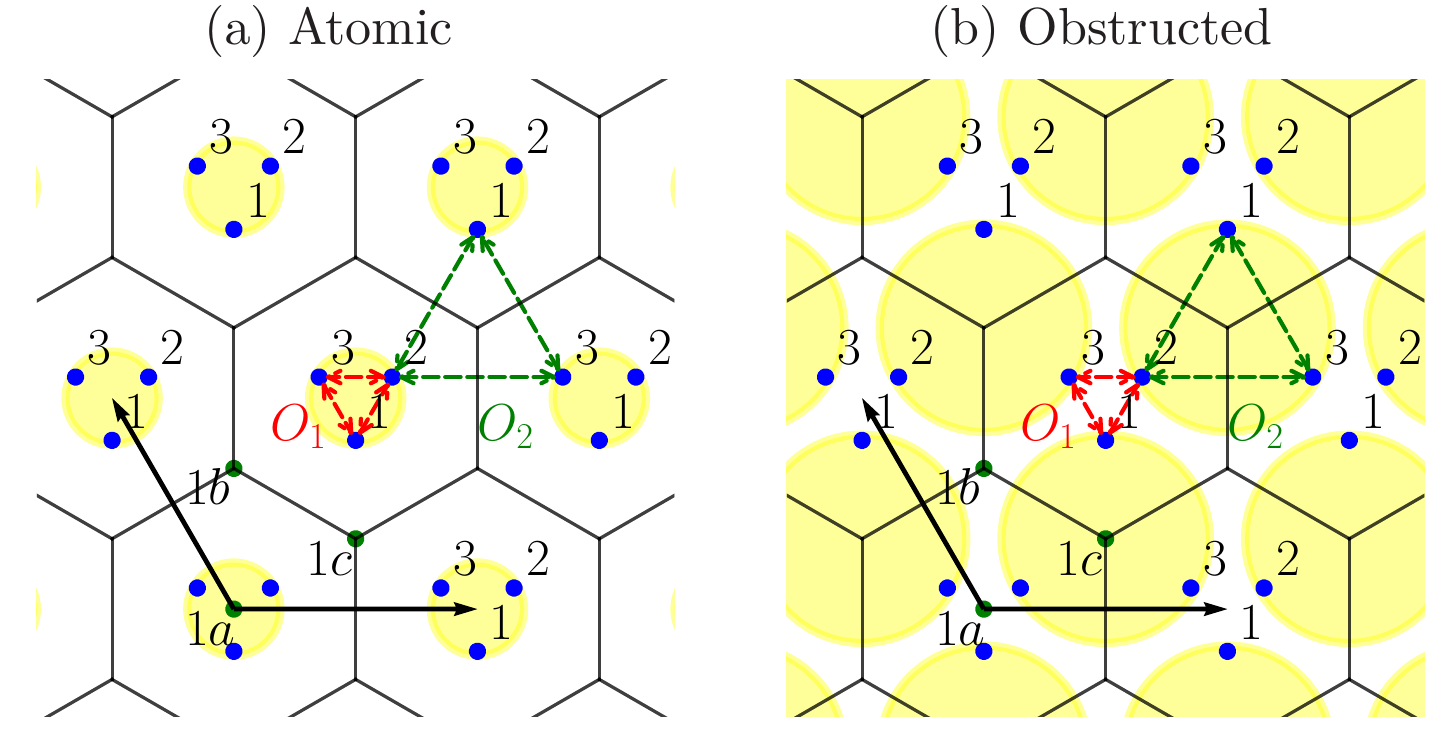}\subfloat{\label{app:fig:OAI_vs_AI:a}}\subfloat{\label{app:fig:OAI_vs_AI:b}}\caption{Illustrations of the unobstructed atomic (UA) and obstructed atomic (OA) phases and their dominant correlators. In each unit cell, the three $s$-like orbitals are represented by numbered blue circles, corresponding to the $\hat{c}^\dagger_{\vec{R},i}$ operators (for $1 \leq i \leq 3$). The $O_1$ and $O_2$ represent the correlators between neighboring Wannier orbitals, as indicated by the respective red arrows. Panel (a) depicts the fictitious UA phase of \ch{NbSe2}, where the quasi-flat band near the Fermi energy arises from the strong hybridization of the $s$-like orbitals surrounding the origin. This phase features $O_1 = \frac{1}{3}$ and $O_2 = 0$. In contrast, panel (b) illustrates the correct OA limit, where the $s$-like orbitals surrounding the $1c$ Wyckoff position strongly hybridize, resulting in $O_1 = 0$ and $O_2 = \frac{1}{3}$. In both cases, the strongly hybridized groups of orbitals are denoted by yellow circles. In the idealized OA and UA limits, the equal superposition of orbitals within the yellow circles corresponds to the $\hat{\gamma}^\dagger_{\vec{R},1}$ and $\hat{\gamma'}^\dagger_{\vec{R},1}$ Wannier orbitals of the quasi-flat bands, as defined, respectively, in \cref{app:eqn:OAI_wave_function_simple,app:eqn:OAI_wave_function_simple_fake}.}
	\label{app:fig:OAI_vs_AI}
\end{figure}
The difference between the fictitious UA and true OA phases of \ch{NbSe2} is illustrated schematically in \cref{app:fig:OAI_vs_AI}. In the atomic limit shown in \cref{app:fig:OAI_vs_AI:a}, the $s$-like orbitals around the unit cell origin strongly hybridize, resulting in a large value for the $O_1$ correlator. In contrast, in the OA phase depicted in \cref{app:fig:OAI_vs_AI:b}, the $s$-like orbitals around the $1c$ Wyckoff position strongly hybridize, leading to a large value of the $O_2$ correlator. In both cases (and more generally), the $O_0$ correlator is fixed to $O_0 = \frac{1}{3}$ by the normalization condition in \cref{app:eqn:normalization_rho_linear}.

\section{Extracting the correlation functions in \ch{NbSe2}}\label{app:sec:experimental}

In this \siSection{}, we demonstrate how the methodology outlined in \cref{app:sec:correlator_general} can be applied to extract orbital correlation functions in \ch{NbSe2}, thereby establishing the OA nature of its Fermi-level band. To achieve this, we perform constant-height tunneling current measurements at positive and negative bias voltages, $V_{\pm}$, which encompass the quasi-flat band at the Fermi level in \ch{NbSe2}. As explained in \cref{app:sec:theory_stm}, this approach enables direct measurement of the band’s CDD. We first describe the process for extracting the CDD Fourier transformation $A \left( \vec{Q}, z \right)$ from the experimental data. Next, we explain how this Fourier transformation is fitted to obtain the corresponding orbital correlation functions. Finally, we present experimental results from various regions and with different STM tips, all of which unequivocally confirm the OA nature of the quasi-flat band in \ch{NbSe2}.

\subsection{Extracting the CDD Fourier transformation $A \left( \vec{Q}, z \right)$ from experimental data}\label{app:sec:experimental:extract_CDD}

As described in the main text, we measure the tunneling current of \ch{NbSe2} at constant height. According to the Tersoff-Hamann approximation discussed in \cref{app:sec:theory_stm:TH_approx:approximation:isolated_band}, this tunneling current is proportional to the CDD of the sample arising from states with energies between the Fermi level and the bias voltage. To determine the CDD of the entire \ch{NbSe2} quasi-flat band, which disperses across the Fermi energy as shown in \cref{app:fig:DFTbands}, we conduct two measurements in close succession at bias voltages bracketing the band. Specifically, we measure the CDD $A^{\pm}_{\text{meas}} \left( \vec{r} \right)$ at bias voltages $V_{\pm}$, where $\abs{e}V_+$ and $\abs{e}V_-$ correspond to the top and bottom of the quasi-flat band, respectively. Since our interest lies in the distribution of the charge density rather than its absolute value, we disregard the proportionality constant between the tunneling current and the integrated spectral function, as described in the Tersoff-Hamann approximation in \cref{app:eqn:final_TH_expression_current}. However, because the two measurements are taken in rapid succession, we assume this proportionality constant (and the sample-tip distance) remains consistent, enabling us to sum the two measurements and obtain the CDD for the \emph{entire} quasi-flat band.

\subsubsection{Fourier transforming the experimental data}\label{app:sec:experimental:extract_CDD:fft}

For both positive and negative biases, the experiment measures the CDD of the sample over a finite square region of approximately $L = \SI{3}{\nano\meter}$. Ignoring other potential sources of error, the measured CDD can be assumed to correspond to the true CDD, subject to an in-plane affine transformation and multiplication by a window function
\begin{equation}
	\label{app:eqn:measured_ldos}
	A^{\pm}_{\text{meas}} \left( \vec{r} \right) = A^{\pm} \left( T_{\pm} \vec{r}_{\parallel} + \vec{r}^{\pm}_{\parallel,0} + z \hat{\vec{z}} \right) w \left(  \frac{\vec{r}_{\parallel}}{L} \right),
\end{equation}
where $T_{\pm}$ represents a linear transformation close to a rotation, $\vec{r}^{\pm}_{\parallel,0}$ denotes the (unspecified) in-plane origin of the experimental data, and $w \left( \frac{\vec{r}_{\parallel}}{L} \right)$ is a window function that accounts for the finite spatial extent of the measurement
\begin{equation}
	w \left( \frac{\vec{r}_{\parallel}}{L} \right) = \begin{cases}
		1 & \qq{if}   \frac{\abs{x}}{L}, \frac{\abs{y}}{L} \leq \frac{1}{2} \\ 
		0 & \qq{otherwise}
	\end{cases}.
\end{equation}
We remind the reader that the vector $\vec{r}$ is decomposed into in-plane and out-of-plane components as $\vec{r} = \vec{r}_{\parallel} + z \hat{\vec{z}} = x \hat{\vec{x}} + y \hat{\vec{y}} + z \hat{\vec{z}}$. The CDD at positive and negative biases are measured in succession, so the affine transformation associated with each measurement is slightly different (due to sample drift), hence the $\pm$ subscript of $T_{\pm}$ and $\pm$ superscript of $\vec{r}^{\pm}_{\parallel,0}$. However, we assume that height drift during a single measurement (which would make $z$ dependent on $\vec{r}_{\parallel}$) or in between the positive- and negative-bias measurements, as well as other experimental errors are negligible. Additionally, we note that the exact sample-tip distance $z$ is not determined in the experiment.

\begin{figure}[t]
	\centering
	\includegraphics[width=\textwidth]{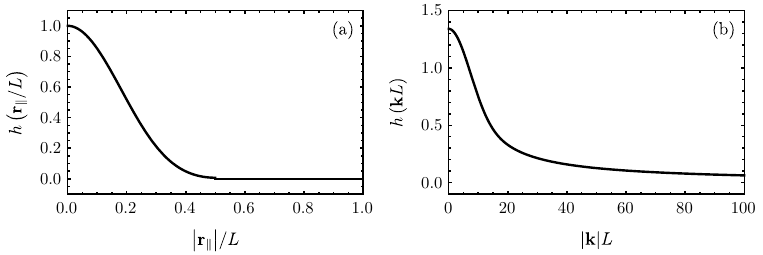}\subfloat{\label{app:fig:fft_kernel:a}}\subfloat{\label{app:fig:fft_kernel:b}}\caption{The circularly symmetric exact Blackman window function and its Fourier transform. (a) depicts the Blackman window function as defined in \cref{app:eqn:blackman_window_func}, while (b) shows its Fourier transform as given in \cref{app:eqn:blackman_to_mom}.}
	\label{app:fig:fft_kernel}
\end{figure}

Before processing, the experimental data is multiplied by a circularly symmetric exact Blackman window function~\cite{BLA59},
\begin{equation}
	\label{app:eqn:blackman_window_func}
	h \left( \frac{ \vec{r}_{\parallel}}{L} \right) = \begin{cases}
		\frac{1}{9304} \left( 4620 \cos \left( \frac{2 \pi \abs{\vec{r}_{\parallel}}}{L} \right) + 715 \cos \left( \frac{4 \pi \abs{\vec{r}_{\parallel}}}{L} \right) + 3969 \right) & \qq{for} \frac{\abs{\vec{r}_{\parallel}}}{L} \leq \frac{1}{2} \\
		0 & \qq{otherwise}
	\end{cases}.
\end{equation}
We also define the Fourier transform of the Blackman function
\begin{align}
	\tilde{h} \left( \vec{k} L \right) &= \frac{1}{L^2}\int \dd[2]{r_{\parallel}} h \left( \frac{ \vec{r}_{\parallel}}{L} \right) e^{i \vec{k} \cdot \vec{r}_{\parallel}}, \label{app:eqn:blackman_to_mom}\\
	h \left( \frac{ \vec{r}_{\parallel}}{L} \right) &= \frac{L^2}{(2 \pi)^2} \int \dd[2]{k} \tilde{h} \left( \vec{k} L \right) e^{-i \vec{k} \cdot \vec{r}_{\parallel}}. \label{app:eqn:blackman_to_real}
\end{align}
The Blackman function and its Fourier transform are illustrated in \cref{app:fig:fft_kernel}. The choice of a circular window function helps avoid introducing additional symmetry-breaking into the Fourier data due to the rectangular shape of the images. Furthermore, employing the smooth Blackman window function effectively suppresses Gibbs oscillations, which could arise from the abrupt edges of a circular step function.

After multiplying $A^{\pm}_{\text{meas}} \left( \vec{r} \right) \to A^{\pm}_{\text{meas}} \left( \vec{r} \right) h \left( \vec{r}_{\parallel} \right)$, the Blackman window function ensures that the measured CDD is zero outside the experimental region. Therefore, the measured CDD can be expressed as
\begin{equation}
	\label{app:eqn:measured_ldos_black}
	A^{\pm}_{\text{meas}} \left( \vec{r} \right) = A^{\pm} \left( T_{\pm} \vec{r}_{\parallel} + \vec{r}^{\pm}_{\parallel,0} + z \hat{\vec{z}} \right) h \left(  \frac{\vec{r}_{\parallel}}{L} \right).
\end{equation}
The filtered experimental data is then Fourier transformed using the Fast Fourier Transform method, yielding
\begin{equation}
	A^{\pm}_{\text{meas}} \left( \vec{k} \right) = \frac{1}{L^2} \int \dd[2]{r_{\parallel}} A^{\pm} \left( T_{\pm} \vec{r}_{\parallel} + \vec{r}^{\pm}_{\parallel,0} + z \hat{\vec{z}} \right) h \left(  \frac{\vec{r}_{\parallel}}{L} \right) e^{- i \vec{k} \cdot \vec{r}_{\parallel}},
\end{equation}
where the integration limit extends to infinity because the Blackman window function vanishes for $\abs{\vec{r}_{\parallel}} > \frac{L}{2}$. Using the Fourier representations of the CDD and Blackman function from \cref{app:eqn:ft_sp_func_to_real,app:eqn:blackman_to_real}, we rewrite
\begin{align}
	A^{\pm}_{\text{meas}} \left( \vec{k} \right) &= \int \dd[2]{r_{\parallel}} \sum_{\vec{Q}} \frac{1}{(2 \pi)^2} \int \dd[2]{k'} A^{\pm} \left( \vec{Q}, z \right) e^{i \vec{Q} \cdot \left( T_{\pm} \vec{r}_{\parallel} + \vec{r}^{\pm}_{\parallel,0} \right)} h \left(  \vec{k}' L \right) e^{ -i \vec{k}' \cdot \vec{r}_{\parallel}} e^{- i \vec{k} \cdot \vec{r}_{\parallel}} \nonumber \\
	&= \sum_{\vec{Q}} \int \dd[2]{k'} A^{\pm} \left( \vec{Q}, z \right) e^{i \vec{Q} \cdot \vec{r}^{\pm}_{\parallel,0}} \tilde{h} \left(  \vec{k}' L \right) \delta \left( T_{\pm}^{-1} \vec{Q} - \vec{k}' - \vec{k} \right) \nonumber \\
	&= \sum_{\vec{Q}} A^{\pm} \left( \vec{Q}, z \right) e^{i \vec{Q} \cdot \vec{r}^{\pm}_{\parallel,0}} \tilde{h} \left( L \left( T_{\pm}^{-1} \vec{Q} - \vec{k} \right) \right). \label{app:eqn:fourier_trafo_measured}
\end{align}
\Cref{app:eqn:fourier_trafo_measured} reveals that $A^{\pm}_{\text{meas}} \left( \vec{k} \right)$ represents the convolution of the exact (but skewed by $T_{\pm}$) CDD -- characterized by Dirac $\delta$-functions at the reciprocal lattice vectors $\vec{Q}$ -- and the Fourier transform of the Blackman function $\tilde{h} \left( \vec{k} L \right)$. Since $\tilde{h} \left( \vec{k} L \right)$ is sharply peaked at the origin, the measured $A^{\pm}_{\text{meas}} \left( \vec{k} \right)$ appears as an array of peaks at the skewed reciprocal lattice vectors $T^{-1}_{\pm} \vec{Q}$. Specifically, we note that because $\abs{L T_{\pm}^{-1} \vec{b}_{1}} \sim \frac{2\pi L}{\abs{\vec{a}_1}} \sim 55$ and, consequently, $\frac{\tilde{h} \left( L T_{\pm}^{-1} \vec{b}_{1} \right)}{h \left( \vec{0} \right)} \sim 0.08 \ll 1$, we can approximate the height of each peak by
\begin{equation}
	\label{app:eqn:measured_CDD_approximation}
	A^{\pm}_{\text{meas}} \left( T_{\pm}^{-1} \vec{Q}' \right) = \sum_{\vec{Q}} A^{\pm} \left( \vec{Q}, z \right) e^{i \vec{Q} \cdot \vec{r}^{\pm}_{\parallel,0}} \tilde{h} \left( L T_{\pm}^{-1} \left(  \vec{Q} - \vec{Q}' \right) \right) \approx A^{\pm} \left( \vec{Q}', z \right) e^{i \vec{Q}' \cdot \vec{r}^{\pm}_{\parallel,0}} \tilde{h} \left( \vec{0} \right).
\end{equation}
This approximation enables us to infer the \emph{true} CDD from the measured CDD as
\begin{equation}
	\label{app:eqn:infer_CDD_from_experiment}
	A^{\pm} \left( \vec{Q}, z \right) \approx \frac{1}{\tilde{h} \left( \vec{0} \right)} A^{\pm}_{\text{meas}} \left( T_{\pm}^{-1} \vec{Q} \right) e^{-i \vec{Q} \cdot \vec{r}^{\pm}_{\parallel,0}}.
\end{equation}

\subsubsection{Determining the unit cell lattice vectors in the experimental data}\label{app:sec:experimental:extract_CDD:lattice_vect}

In practice, however, the affine transformation corresponding to the experimental measurements is \emph{not} known, meaning that both $T_{\pm}$ and $\vec{r}^{\pm}_{\parallel,0}$ must also be inferred from the experimental data. To illustrate this process, let $T_{\pm}$ and $\vec{r}^{\pm}_{\parallel,0}$ represent the ``true'' affine transformation, while $T'_{\pm}$ and $\vec{r}^{\prime,\pm}_{\parallel,0}$ denote the affine transformation inferred from the STM measurements.

The linear transformation $T'_{\pm}$ can be determined by analyzing $A^{\pm}_{\text{meas}} \left( \vec{k} \right)$, which exhibits sharp peaks at $\vec{k} = T^{-1}_{\pm} \vec{Q}$ (or equivalently at $\vec{k} = \left( T^{\prime}_{\pm} \right)^{-1} \vec{Q}$), where $\vec{Q} \in \mathbb{Z} \vec{b}_1 + \mathbb{Z} \vec{b}_2$. This relationship implies that for any $m,n \in \mathbb{Z}$, there exist integers $m',n' \in \mathbb{Z}$ satisfying
\begin{equation}
	T^{-1}_{\pm} \left( m \vec{b}_1 + n \vec{b}_2 \right) = \left( T^{\prime}_{\pm} \right)^{-1} \left( m' \vec{b}_1 + n' \vec{b}_2 \right),
\end{equation}
or, alternatively,
\begin{equation}
	\label{app:eqn:equation_for_trafo_stm}
	T^{\prime}_{\pm} T^{-1}_{\pm} \left( m \vec{b}_1 + n \vec{b}_2 \right) = m' \vec{b}_1 + n' \vec{b}_2. 
\end{equation}
By assumption, both $T_{\pm}$ and $T^{\prime}_{\pm}$ are close to rotations, implying that $T^{\prime}_{\pm} T^{-1}_{\pm}$ is also approximately a rotation ({\it i.e.}{}, its determinant is close to 1). \Cref{app:eqn:equation_for_trafo_stm} admits six solutions
\begin{equation}
	T^{\prime}_{\pm} T^{-1}_{\pm} = \begin{pmatrix}
		\cos \left( \frac{2 \pi k}{6} \right) & -\sin \left( \frac{2 \pi k}{6} \right) \\
		\sin \left( \frac{2 \pi k}{6} \right) & \cos \left( \frac{2 \pi k}{6} \right) \\
	\end{pmatrix}, \qq{for} k \in \mathbb{Z}, 0 \leq k < 6.
\end{equation}
However, due to the $C_{3z}$ symmetry of \ch{NbSe2}, only two of these solutions are distinct, specifically
\begin{equation}
	T^{\prime}_{\pm} = T_{\pm} \qq{or}
	T^{\prime}_{\pm} = - T_{\pm}.
\end{equation}
Thus, by examining the peak positions in $A^{\pm}_{\text{meas}} \left( \vec{k} \right)$, the transformation $T^{\prime}_{\pm}$ can only be determined up to an overall sign. This sign ambiguity implies that the lattice vectors derived from the experimental data may either coincide with the conventional ones or differ by a minus sign. Since \ch{NbSe2} lacks $C_{2z}$ symmetry, this sign is significant and must be resolved before fitting the STM data.

\subsubsection{Determining the unit cell origin in the experimental data}\label{app:sec:experimental:extract_CDD:unit_cell_origin}

Before discussing how the sign of $T^{\prime}_{\pm}$ is determined, we consider the origin of the experimental STM data, namely $\vec{r}^{\pm}_{\parallel,0}$. To obtain this, we first note that \ch{NbSe2} is $C_{3z}$ symmetric, as discussed in \cref{app:sec:NbSe2_ab_initio:crystal_struct}. Consequently, from \cref{app:eqn:sym_constraint_FFT_CDD}, the Fourier transform of the true CDD must satisfy $A^{\pm} \left( \vec{Q}, z \right) = A^{\pm} \left( C_{3z} \vec{Q}, z \right)$. This symmetry condition can be applied to the largest nontrivial Fourier component of the measured STM data ({\it i.e.}{}, the first harmonic) to determine $\vec{r}^{\pm}_{\parallel,0}$. 

Defining 
\begin{equation}
	\vec{Q}_{n} = C^{n}_{3z} \vec{b}_{1}, \qq{for} n=0,1,2,
\end{equation}
we find from \cref{app:eqn:infer_CDD_from_experiment} that
\begin{equation}
	\label{app:eqn:system_of_eqn_for_r0_and_arg}
	\arg \left( A^{\pm} \left( \vec{b}'_1, z \right) \right) = \arg \left( A^{\pm}_{\text{meas}} \left( \left( T'_{\pm} \right)^{-1} \vec{Q}_n \right) \right) - \vec{Q}_n \cdot \vec{r}^{\pm}_{\parallel,0} \mod \left( 2 \pi \right), \qq{for} n=0,1,2,
\end{equation}
where 
\begin{equation}
	\vec{b}'_1 =\begin{cases}
		\vec{b}_1 & \qq{if} T'_{\pm} = T_{\pm} \\
		- \vec{b}_1 & \qq{if} T'_{\pm} = - T_{\pm}
	\end{cases}.
\end{equation}
These equations provide three distinct solutions for $\arg \left( A^{\pm} \left( \vec{b}'_1, z \right) \right)$ and the two components of $\vec{r}^{\pm}_{\parallel,0}$. Specifically, we deduce for $\arg \left( A^{\pm} \left( \vec{b}'_1, z \right) \right)$
\begin{equation}
	3 \arg \left( A^{\pm} \left( \vec{b}'_1, z \right) \right) = \sum_{n=0}^2 \arg \left( A^{\pm}_{\text{meas}} \left( \left( T'_{\pm} \right)^{-1} \vec{Q}_n \right) \right) - \left( \sum_{n=0}^2 \vec{Q}_n \right) \cdot \vec{r}^{\pm}_{\parallel,0},
\end{equation}
which implies three distinct solutions for $\arg \left( A^{\pm} \left( \vec{b}'_1, z \right) \right)$
\begin{equation}
	\arg \left( A^{\pm} \left( \vec{b}'_1, z \right) \right) = \frac{1}{3} \sum_{n=0}^2 \arg \left( A^{\pm}_{\text{meas}} \left( \left( T'_{\pm} \right)^{-1} \vec{Q}_n \right) \right) + \frac{2 \pi k}{3}, \qq{for} k = 0, 1, 2,
\end{equation}
which, in turn, result in three different solutions for $\vec{r}^{\pm}_{\parallel,0}$. Assuming $\vec{r}^{\pm}_{\parallel,0} = \beta^{\pm}_1 \vec{a}_1 + \beta^{\pm}_2 \vec{a}_2$, the first two equations in \cref{app:eqn:system_of_eqn_for_r0_and_arg} yield 
\begin{align}
	\arg \left( A^{\pm} \left( \vec{b}'_1, z \right) \right) &= \arg \left( A^{\pm}_{\text{meas}} \left( \left( T'_{\pm} \right)^{-1} \vec{Q}_0 \right) \right) - 2 \pi \beta^{\pm}_1, \\
	\arg \left( A^{\pm} \left( \vec{b}'_1, z \right) \right) &= \arg \left( A^{\pm}_{\text{meas}} \left( \left( T'_{\pm} \right)^{-1} \vec{Q}_1 \right) \right) - 2\pi \left( \beta^{\pm}_2 - \beta^{\pm}_1 \right).
\end{align}
These equations can be solved for the unit cell origin coordinates
{\small
\begin{align}
	\beta^{\pm}_1 &= \frac{1}{2 \pi} \left( - \frac23 \arg \left( A^{\pm}_{\text{meas}} \left( \left( T'_{\pm} \right)^{-1} \vec{Q}_0 \right) \right) + \frac13 \arg \left( A^{\pm}_{\text{meas}} \left( \left( T'_{\pm} \right)^{-1} \vec{Q}_1 \right) \right) + \frac13 \arg \left( A^{\pm}_{\text{meas}} \left( \left( T'_{\pm} \right)^{-1} \vec{Q}_2 \right) \right) \right) + \frac{2 k}{3},  \label{app:eqn:origin_fft_choice_beta_1}\\
	\beta^{\pm}_2 &= \frac{1}{2 \pi} \left( - \frac13 \arg \left( A^{\pm}_{\text{meas}} \left( \left( T'_{\pm} \right)^{-1} \vec{Q}_0 \right) \right) - \frac13 \arg \left( A^{\pm}_{\text{meas}} \left( \left( T'_{\pm} \right)^{-1} \vec{Q}_1 \right) \right) + \frac23 \arg \left( A^{\pm}_{\text{meas}} \left( \left( T'_{\pm} \right)^{-1} \vec{Q}_2 \right) \right) \right) + \frac{2 k}{3}. \label{app:eqn:origin_fft_choice_beta_2}
\end{align}}This shows that three distinct solutions exist for the unit cell origin as determined from the experiment. 

\subsubsection{Fixing the unit cell in the experimental data}\label{app:sec:experimental:extract_CDD:fix_unit_cell}

The analysis presented above can be summarized more intuitively as follows:
\begin{itemize}
	\item By analyzing the periodicity of the experimentally measured CDD, the lattice vectors of the sample can be identified, albeit only up to an overall minus sign.
	\item The origin of the unit cell corresponds to one of the three $C_{3z}$-symmetric sites that can be identified in the experimentally measured CDD. These three choices for the unit cell origin correspond to the three distinct solutions for $\vec{r}^{\pm}_{\parallel,0}$ derived in \cref{app:eqn:origin_fft_choice_beta_1,app:eqn:origin_fft_choice_beta_2}.
\end{itemize}
Given the experimentally measured CDD, there are six possible unit cell choices, out of which only one corresponds to the conventional unit cell described in \cref{app:sec:NbSe2_ab_initio:crystal_struct} and employed throughout this work. These six choices arise from the possible assignments of the three $C_{3z}$-symmetric sites of \ch{NbSe2} from \cref{app:eqn:c3z_pos_nbse2} to the three $C_{3z}$-symmetric points of the measured CDD ({\it i.e.}{}, there are $3! = 6$ permutations of the sites).

The six unit cell choices are not equivalent. Since the atomic positions are not directly observable in STM experiments, we rely on the spectral weight hierarchy discussed in \cref{app:sec:NbSe2_ab_initio:dft_simulations} to assign the $C_{3z}$-symmetric sites of the unit cell ({\it i.e.}{}, $1a$, $1b$, and $1c$) to the corresponding points in the measured CDD. The simulated CDD of the \ch{NbSe2} quasi-flat band shown in \cref{app:fig:DFT-CDD-diffH} consistently indicates that the $1b$ (Se) site is the ``brightest'', followed by the $1c$ (empty) site, with the $1a$ (Nb) site being the ``darkest''. This brightness hierarchy is observed across sample-tip distances $z/\si{\angstrom} \in \left[2, 5\right]$ and enables the precise identification of the conventional \ch{NbSe2} unit cell in the experimental CDD.

This identification is further validated in \cref{fig:validate_pos} of the main text, where we compare the local density of states (LDOS) of the three $C_{3z}$-symmetric sites at two bias voltages: one near the Fermi level and another at $V \approx \SI{-2}{\volt}$. \textit{Ab initio} simulations show that at both bias voltages, the $1b$ (Se) site is the brightest due to the proximity of the Se atoms to the STM tip. Around the Fermi level, tunneling primarily occurs into the quasi-flat band, which forms an obstructed atomic limit with a Wannier center at the $1c$ (empty) site, making the $1c$ site brighter than the $1a$ (Nb) site. At lower bias voltages $V \approx \SI{-2}{\volt}$, tunneling occurs into both the mirror-even bands with Wannier centers at the $1a$ (Nb) site and the mirror-odd bands with Wannier centers at the $1c$ (empty) site (see \cref{app:fig:DFTbands}). As a result, the LDOS at the $1a$ (Nb) and $1c$ (empty) sites become comparable. This behavior, predicted by the simulated LDOS at different energies in \cref{app:fig:DFT-CDD-diffE}, is also reproduced experimentally.

In the experiment, we measure conductance $\dv{I}{V}$ maps over the same sample region at two different bias voltages: one near \SI{0}{\volt} and another near \SI{-2}{\volt}. According to the Tersoff-Hamann approximation in \cref{app:sec:theory_stm:TH_approx:approximation}, the conductance at a bias voltage $V$ is proportional to the sample's LDOS at energy $\omega = \abs{e}V$. Using the data at $V \approx \SI{0}{\volt}$, we identify the three $C_{3z}$-symmetric sites based on the brightness hierarchy $1b \text{ (Se)}> 1c \text{ (empty)} > 1a \text{ (Nb)}$. Since the atomic positions remain fixed across the two bias voltages, this identification holds for both measurements. The experimental results, shown in \cref{fig:validate_pos:c}, agree with the simulated data in \cref{app:fig:DFT-CDD-diffE}.

In summary, we identify the atomic positions in the STM data based on the simulated CDD from DFT for the quasi-flat band, and this identification is corroborated by conductance measurements at lower bias voltages.

\subsection{Fitting the correlation function from experimental data}\label{app:sec:experimental:fit}

Once the conventional unit cell has been identified in the experimental data, we proceed to extract the correlation function $\rho_{ij} \left( \Delta \vec{R} \right)$ from the Fourier transformation of the quasi-flat band CDD, given by $A \left( \vec{Q}, z_{\text{meas}} \right) = A^+ \left( \vec{Q}, z_{\text{meas}} \right) + A^- \left( \vec{Q}, z_{\text{meas}} \right)$. Here, $z_{\text{meas}}$ denotes the \emph{unknown} tip-sample distance, distinguishing it from the \emph{assumed} tip-sample distance $z$.

As outlined in the main text, we fit only three primary matrix elements of the orbital correlation function (and their symmetry-related counterparts), denoted as $O_i$ (for $0 \leq i \leq 2$): $O_0 = \rho_{11} \left( \vec{0} \right)$, $O_1 = \rho_{12} \left( \vec{0} \right)$, and $O_2 = \rho_{23} \left( \vec{a}_1 \right)$. These are selected while approximating all other elements to zero. The inclusion of $O_0$ and $O_2$ is motivated by their dominance according to the \textit{ab initio}{} simulations presented in \cref{app:tab:OPs-3band-model}. Although $O_1$ is smaller, it is included due to its relevance to the correlator associated with the fictitious UA phase of the \ch{NbSe2} quasi-flat band, as detailed in \cref{app:sec:NbSe2_ab_initio:wanniers:limits}. Further-range correlators are excluded to mitigate risks of overfitting.

We employ the constrained minimization procedure from \cref{app:eqn:general_constraint_minimization} to fit the correlators $O_i$ (for $0 \leq i \leq 2$). However, the constraint in \cref{app:eqn:general_constraint_minimization} can be eliminated through an efficient parameterization of the rescaled correlation function $\tilde{\rho}_{ij} \left( \Delta \vec{R} \right)$. Using the normalization conditions from \cref{app:eqn:normalization_rho_linear,app:eqn:normalization_rho_square} and assuming that correlators other than $O_0$, $O_1$, and $O_2$ are negligible, we derive the following explicit constraints:
\begin{equation}
	\label{app:eqn:constraint_nbse2_explicit}
	3 O_0 = 1 \qq{and}
	3 O_0^2 + 6 O_1^2 + 6 O_2^2 = 1.
\end{equation}
The space of solutions that satisfy these constraints can be efficiently parameterized as
\begin{equation}
	O_0 = \frac{1}{3}, \quad
	O_1 = \frac{1}{3} \cos \theta, \quad
	O_2 = \frac{1}{3} \sin \theta,
\end{equation}
indicating that the nonzero elements of the \emph{rescaled} correlation function can be expressed as
\begin{equation}
	\label{app:eqn:rescaled_rho_parameterization}
	\tilde{\rho}_{11} \left( \vec{0} \right) = \frac{\gamma}{3}, \quad
	\tilde{\rho}_{12} \left( \vec{0} \right) = \frac{\gamma}{3} \cos \theta, \quad
	\tilde{\rho}_{23} \left( \vec{0} \right) = \frac{\gamma}{3} \sin \theta,
\end{equation}
with $\gamma > 0$. All other elements of $\tilde{\rho}_{ij} \left( \Delta \vec{R} \right)$ that are not symmetry-related to the ones in \cref{app:eqn:rescaled_rho_parameterization} are set to zero.

Using the efficient parameterization in \cref{app:eqn:rescaled_rho_parameterization}, the constrained minimization problem in \cref{app:eqn:general_constraint_minimization} becomes an unconstrained one
\begin{equation}
	\label{app:eqn:constraint_minimization_nbse2}
	\min_{\gamma,\theta} \sum_{\vec{Q}} \abs{ A \left( \vec{Q}, z_{\text{meas}} \right) - \sum_{\Delta \vec{R}, i, j} B_{ij} \left( \vec{Q}, z, \Delta \vec{R} \right) \tilde{\rho}_{ij} \left( \Delta \vec{R} \right)}^2 .
\end{equation}
Finally, after normalizing according to \cref{app:eqn:normalization_rho_linear}, the orbital correlation function associated with the \ch{NbSe2} quasi-flat band is obtained. In the experiment, the accurate value of the tip-sample distance, denoted by $z_{\text{meas}}$ in \cref{app:eqn:constraint_minimization_nbse2}, is not known. Therefore, we treat the \emph{assumed} tip-sample distance $z$ as an additional fitting parameter. As explained in the main text, we fit the orbital correlation function assuming different tip-sample distances $z$ and then select the result with the minimal error, as detailed below. 

For a given assumed height $z$ and the corresponding fitted rescaled correlation function $\tilde{\rho}_{ij} \left( \Delta \vec{R} \right)$, we define two metrics for the fitting error. The first metric quantifies the error between the experimentally measured $A \left( \vec{Q}, z_{\text{meas}} \right)$ and the computed $A \left( \vec{Q}, z \right)$ using the fitted $\tilde{\rho}_{ij} \left( \Delta \vec{R} \right)$
\begin{equation}
	\epsilon = \sqrt{ \frac{ \sum_{\vec{Q}} \abs{ A \left( \vec{Q}, z_{\text{meas}} \right) - \sum_{\Delta \vec{R}, i, j} B_{ij} \left( \vec{Q}, z, \Delta \vec{R} \right) \tilde{\rho}_{ij} \left( \Delta \vec{R} \right)}^2}{\sum_{\vec{Q}} \abs{ A \left( \vec{Q}, z_{\text{meas}} \right)}^2}}.
\end{equation}
Since some symmetry-breaking effects exist in the experimental data while the fitted CDD is necessarily symmetric, it is useful to introduce an error metric that isolates the fitting error \emph{without} accounting for the experimental symmetry-breaking. To achieve this, we first symmetrize the measured CDD $A \left( \vec{Q}, z_{\text{meas}} \right)$. For each group of reciprocal lattice vectors $\left\lbrace \vec{Q}_i \right\rbrace$ related by the unitary crystalline symmetries of \ch{NbSe2} (namely $C_{3z}$ and $m_x$), the \emph{symmetrized} CDD is computed as
\begin{equation}
	A_{\text{sym}} \left( \vec{Q}_i, z_{\text{meas}} \right) = \frac{\sum_{i=1}^{N_{\vec{Q}}} A \left( \vec{Q}_i, z_{\text{meas}} \right) }{N_{\vec{Q}}}, \qq{for} 1 \leq i \leq N_{\vec{Q}},
\end{equation}
where $N_{\vec{Q}}$ is the number of symmetry-related reciprocal lattice vectors in the group $\left\lbrace \vec{Q}_i \right\rbrace$. Using this symmetrized CDD, we define the second error metric
\begin{equation}
	\epsilon_{\text{sym}} = \sqrt{ \frac{ \sum_{\vec{Q}} \abs{ A_{\text{sym}} \left( \vec{Q}, z_{\text{meas}} \right) - \sum_{\Delta \vec{R}, i, j} B_{ij} \left( \vec{Q}, z, \Delta \vec{R} \right) \tilde{\rho}_{ij} \left( \Delta \vec{R} \right)}^2}{\sum_{\vec{Q}} \abs{ A_{\text{sym}} \left( \vec{Q}, z_{\text{meas}} \right)}^2}}.
\end{equation}
The error between the fitted CDD and the unsymmetrized experimental data is relatively flat as a function of $z$. Therefore, we employ the minimum value of $\epsilon_{\text{sym}}$ to determine the correct tip height, as well as the corresponding fitted correlation function. 
 
\begin{table}[t]
	\centering
	\begin{tabular}{|c|c|c|c|c|c|c|}
		\hline
		Tip & Set & $\left(V_-, V_+ \right) / \si{\volt}$ & fitted $z / \si{\angstrom}$ & $O_0$ & $O_1$ & $O_2$  \\ 
		\hline
		\multirow{8}{*}{1} & 1 & $(-0.320,  0.60)$ &  3.30 &      0.333 &      0.014 &      0.333 \\ 
		& 2 & $(-0.320,  0.60)$ &  3.30 &      0.333 &      0.018 &      0.333 \\ 
		& 3 & $(-0.320,  0.60)$ &  3.30 &      0.333 &      0.021 &      0.333 \\ 
		& 4 & $(-0.320,  0.60)$ &  3.30 &      0.333 &      0.020 &      0.333 \\ 
		& 5 & $(-0.320,  0.60)$ &  3.30 &      0.333 &      0.019 &      0.333 \\ 
		& 6 & $(-0.320,  0.60)$ &  3.30 &      0.333 &      0.018 &      0.333 \\ 
		& 7 & $(-0.200,  0.30)$ &  3.30 &      0.333 &      0.022 &      0.333 \\ 
		& 8 & $(-0.200,  0.10)$ &  3.30 &      0.333 &      0.015 &      0.333 \\ \hline
		\multirow{3}{*}{2} & 1 & $(-0.320,  0.62)$ &  3.30 &      0.333 &      0.124 &      0.309 \\ 
		& 2 & $(-0.250,  0.55)$ &  3.25 &      0.333 &      0.129 &      0.307 \\ 
		& 3 & $(-0.150,  0.25)$ &  3.25 &      0.333 &      0.132 &      0.306 \\ \hline 
		\multirow{8}{*}{3} & 1 & $(-0.320,  0.63)$ &  3.30 &      0.333 &      0.034 &      0.332 \\ 
		& 2 & $(-0.320,  0.63)$ &  3.25 &      0.333 &      0.032 &      0.332 \\ 
		& 3 & $(-0.320,  0.62)$ &  3.25 &      0.333 &      0.030 &      0.332 \\ 
		& 4 & $(-0.320,  0.62)$ &  3.30 &      0.333 &      0.030 &      0.332 \\ 
		& 5 & $(-0.320,  0.62)$ &  3.30 &      0.333 &      0.032 &      0.332 \\ 
		& 6 & $(-0.250,  0.42)$ &  3.25 &      0.333 &      0.036 &      0.331 \\ 
		& 7 & $(-0.150,  0.25)$ &  3.25 &      0.333 &      0.043 &      0.331 \\ 
		& 8 & $(-0.050,  0.05)$ &  3.70 &      0.333 &      0.080 &      0.324 \\ 
		\hline
	\end{tabular}
	\caption{Summary of the fitted tip distances $z$ and the inter-orbital correlators $O_{i}$ (for $1 \leq i \leq 3$) derived from all experimental data sets measured in this work. For each STM tip, multiple sets of measurements were performed (each at both positive and negative bias voltages) across different regions of the sample. Additionally, measurements where the bias voltages bracket only part of the quasi-flat band are included, as these also align with the OA nature of the quasi-flat band.}
	\label{app:tab:fitting_results}
\end{table}

Experimentally, we perform multiple measurements using three distinct STM tips, across various regions of the sample. The results are summarized in \cref{app:tab:fitting_results}, where we tabulate the best-fitting results ({\it i.e.}{}, those with the smallest fitting error) for each data set. All measurements consistently yield inter-orbital correlators $O_1 \approx 0$ and $O_2 \approx \frac{1}{3}$, confirming the OA nature of the \ch{NbSe2} quasi-flat band, as detailed in \cref{app:sec:NbSe2_ab_initio:wanniers:limits}. The fitted tip distances $z$ range between $\SI{3}{\angstrom}$ and $\SI{4}{\angstrom}$. These findings unambiguously demonstrate that the inter-orbital correlation functions extracted from the experimental measurements validate the obstructed nature of the quasi-flat band in \ch{NbSe2}.

\end{document}